\documentclass{article}
\usepackage{amsmath,amssymb}
\usepackage{amsthm}
\usepackage{mathrsfs}
\usepackage{accents}
\usepackage[margin=3cm]{geometry}
\usepackage[utf8]{inputenc}
\usepackage{xcolor}
\usepackage{mathtools}
\DeclarePairedDelimiter{\ceil}{\lceil}{\rceil}

\usepackage{hyperref}
\usepackage{graphicx}
\usepackage{caption}
\usepackage{subcaption}

\hypersetup{
    colorlinks,
    citecolor=black,
    filecolor=black,
    linkcolor=black,
    urlcolor=black
}
\usepackage{cleveref}
\usepackage[style=alphabetic]{biblatex}
\usepackage{float}
\usepackage{verbatim}
\usepackage[normalem]{ulem}
\newcommand{\stkout}[1]{\ifmmode\text{\sout{\ensuremath{#1}}}\else\sout{#1}\fi}

\usepackage{graphicx}
\usepackage{latexsym}
\usepackage[dvips]{epsfig}
\usepackage{wasysym}
\usepackage{bm}
\usepackage{slashed}
\usepackage{yhmath} 
\usepackage{stmaryrd}
\usepackage{appendix}

\theoremstyle{remark}
\newtheorem{remark}{Remark}

\def\bma{{\bm a}}

\def\bme{{\bm e}}

\def\bmu{{\bm u}}

\def\bmzero{{\bm 0}}
\def\bmone{{\bm 1}}

\def\bmA{{\bm A}}
\def\bmB{{\bm B}}
\def\bmC{{\bm C}}
\def\bmD{{\bm D}}

\def\bmQ{{\bm Q}}

\def\bmX{{\bm X}}
\def\bmZ{{\bm Z}}

\def\bmeta{{\bm \eta}}

\def\bmphi{{\bm \phi}}

\def\bmsigma{{\bm \sigma}}

\def\bmPhi{{\bm \Phi}}

\def\bmpartial{{\bm \partial}}

\allowdisplaybreaks

\newcounter{mnotecount}

\newcommand{\mnotex}[1]
{\protect{\stepcounter{mnotecount}}$^{\mbox{\footnotesize $\bullet$\themnotecount}}$ 
\marginpar{
\raggedright\tiny\em
$\!\!\!\!\!\!\,\bullet$\themnotecount: #1} }

\title{Controlled regularity at future null infinity from past asymptotic initial data: massless fields}
\author{G. Taujanskas\footnote{Electronic address: \texttt{taujanskas@dpmms.cam.ac.uk}} ~and J. A. Valiente Kroon\footnote{Electronic address: \texttt{j.a.valiente-kroon@qmul.ac.uk}}}
\date{\today}

\newtheorem{theorem}{Theorem}[section]

\newtheorem{lemma}[theorem]{Lemma}
\newtheorem{proposition}[theorem]{Proposition}

\newcommand{\scri}[0]{\mathscr{I}}
\renewcommand{\d}{\operatorname{d}\!}
\newcommand{\la}{\lesssim}
\renewcommand{\leq}{\leqslant}
\renewcommand{\geq}{\geqslant}

\crefformat{equation}{(#2#1#3)}
\crefrangeformat{equation}{(#3#1#4)--(#5#2#6)}
\crefmultiformat{equation}{(#2#1#3)}%
{ and~(#2#1#3)}{, (#2#1#3)}{ and~(#2#1#3)}

\begin{document}

\maketitle


\begin{abstract}
We study the relationship between asymptotic characteristic initial data at past null infinity and the regularity of solutions at future null infinity for the massless linear spin-$s$ field equations on Minkowski space. By quantitatively controlling the solutions on a causal rectangle reaching the conformal boundary, we relate the (generically singular) behaviour of the solutions near past null infinity, future null infinity, and spatial infinity. Our analysis uses Friedrich's cylinder at spatial infinity together with a careful Gr\"onwall-type estimate that does not degenerate at the intersection of null infinity and the cylinder (the so-called critical sets).
\end{abstract}

\setcounter{tocdepth}{1}
\tableofcontents

\section{Introduction}
\label{Section:Introduction}

In the early 1960s, Penrose famously observed \cite{Pen63,Pen65a} that techniques from conformal geometry could be used to study the asymptotic regime of fields in general relativity. Since zero rest-mass fields carry no preferred length scale, their equations of motion exhibit an essential conformal invariance, making Penrose's technique particularly suited in the study of the asymptotics of massless fields. By attaching to a spacetime a conformal boundary ``at infinity", Penrose was able to translate questions regarding the asymptotic behaviour of physical fields into questions about the \emph{local} regularity of suitably rescaled (or \emph{unphysical}) fields, now evolving on a conformally rescaled (unphysical) spacetime. Spacetimes for which the procedure of attaching such a boundary was possible were termed \emph{asymptotically simple}, and since its inception the notion has been instrumental in the development of the current understanding of the asymptotic behaviour of massless fields, including gravity \cite{Ger76,PenRin86,Fri04}.

An important application of the idea of asymptotic simplicity has been to \emph{scattering problems}, the study of the relationship between the asymptotics of fields in the distant past and the distant future. Penrose's conformal boundary---called \emph{null infinity}\footnote{Null infinity, denoted $\scri$, derives its name as a result of being the locus of all endpoints of inextendible null geodesics in the spacetime. It is a co-dimension one hypersurface which may be spacelike, timelike, or null, corresponding to the sign of the cosmological constant $\Lambda$ being positive, negative, or zero, respectively. When $\Lambda \geq 0$, $\scri$ has a past component and a future component, denoted $\scri^-$ and $\scri^+$.}---fits very naturally here, and was first understood in relation to the classical Lax--Phillips scattering theory\footnote{Classical scattering theory has a long history going back to at least the works of Lax and Phillips \cite{LaxPhillips1964,LaxPhillips1967} and Friedlander \cite{Friedlander1962,Friedlander1964,Friedlander1967} in the 1960s. For a more complete bibliography, see for example \cite{Taujanskas2018} and references therein.} by Friedlander in \cite{Fried80}. In the case of zero cosmological constant, Friedlander showed that scattering problems could be formulated as \emph{characteristic initial value problems} with data at past or future null infinity, and the ideas of such \emph{conformal} formulations of scattering were later taken up by Baez, Segal and Zhou \cite{BaezSegalZhou1990} and subsequently many others (see e.g. \cite{NicTau22} and references therein). In cases when the field equations are linear or the data is highly regular, characteristic initial value problems are known to be well-posed \cite{Hor90,Ren90,ChrPae12,ChrPae13b,ChruscielCabetWafo2014}, giving rise to a \emph{scattering operator}\footnote{We gloss over the issue of asymptotic completeness here, i.e. that a scattering operator should be invertible.}: a map from, say, data on $\scri^-$ to induced data on $\scri^+$. A detailed understanding of this scattering operator, however, requires an understanding of how the structure of the data on $\scri^-$ is related to the structure of the solution near $\scri^+$.

In general, this is a non-trivial task, in part due to the behaviour of the fields near spatial infinity $i^0$. In Penrose's picture \cite{Pen65a} of smoothly compactified Minkowski space, the two components of null infinity meet at a point, $i^0$, where the integral curves of the generators of $\scri^\pm$ intersect. Spatial infinity $i^0$ is therefore a caustic, and, unless the data is supported away from $i^0$, solutions generally evolve non-smoothly near $i^0$. It is now understood that generic behaviour is at best \emph{polyhomogeneous} near spatial and null infinity \cite{Fri98a,HinVas20}, i.e. the asymptotic expansions of fields involve polynomial and logarithmic\footnote{It was observed in \cite{Val04a,Val04d} that there are two conceptually distinct classes of logarithmic terms. One is a consequence of the caustic nature of $i^0$, and is the focus of the present paper. The second arises from the singularity in the Weyl tensor at $i^0$ on spacetimes with non-zero ADM mass. In full general relativity these two classes are unavoidably intertwined, but the second is of course absent on Minkowski space.} terms of the form $x^n \log^m x$. In an effort to resolve the detailed structure of spatial infinity, in 1998 Friedrich \cite{Fri98a} introduced\footnote{Before that, Ashtekar and Hansen \cite{AshHan78} had put forward a similar construction of a \emph{hyperboloid at spatial infinity}, which turns out to be closely related to Friedrich's cylinder \cite{MagVal21}.} a different conformal representation of $i^0$, now known as the \emph{cylinder at spatial infinity}, or the \emph{F-gauge}. In general, the utility of this conformal gauge lies in the fact that it permits the formulation of a regular initial value problem in a neighbourhood of spatial infinity for Friedrich's \emph{conformal Einstein field equations} \cite{Fri98b}. In full generality, Friedrich's cylinder is constructed using conformal geodesics, and so carries a strong geometric underpinning. In the case of Minkowski space, however, the cylinder may be constructed using an ad hoc conformal transformation, which we recap in \Cref{Section:CylinderMinkowski}. Briefly, the key idea is to blow up the point $i^0$ to a $(1+2)$-dimensional submanifold $\mathcal{I}$ with topology $(-1,1) \times \mathbb{S}^2$, which then becomes a total characteristic of the field equations. The ends $\mathcal{I}^\pm = \{ \pm 1\} \times \mathbb{S}^2$ of this cylinder intersect $\scri^\pm$; on these spheres---called the \emph{critical sets}---the field equations lose symmetric hyperbolicity and therefore degenerate. It is this degeneracy that is at the heart of the logarithmic divergences which we study in this paper.

Friedrich's cylinder has been used to study the polyhomogeneity of spin-0, spin-1 and spin-2 fields starting from spacelike Cauchy data in \cite{Fri98a,Fri03b,Val03a,MinMacVal22,GasVal22,Val07b,Val09a}. The case of the characteristic initial value problem near $i^0$, for Friedrich's conformal Einstein equations, has been investigated by Paetz \cite{Pae18}, who studied conditions on the characteristic data which guarantee smoothness at the critical sets $\mathcal{I}^\pm$. The degeneracy of the spin-$s$ equations at $\mathcal{I}^\pm$ arises as a loss of rank in the matrix multiplying the time derivatives, destroying symmetric hyperbolicity. Since standard estimates for symmetric hyperbolic systems (which rely on very general properties of the principal part of the equations \cite{Joh91}) now fail, new estimates are needed. In this paper we study the linear massless spin-$s$ equations on Minkowski space, for\footnote{Here are throughout the paper, $s\neq 0$. Although results similar to the ones presented in this paper appear obtainable for spin-0 fields, we do not handle the case of the wave equation as it is not in the form of a zero rest-mass free field equation $\nabla_A{}^{A'} \phi_{AB\ldots M} = 0$, which manifestly exhibits structures on which our estimates rely.} $s \in \frac{1}{2} \mathbb{N}$, starting from characteristic data on $\scri^-$, and relate the structural properties of solutions in the regions near $\scri^-$, $i^0$, and $\scri^+$. Our starting point are the estimates of Friedrich \cite{Fri03b}, which we generalize to all non-zero spins\footnote{While the most attention in the literature appears to have been devoted to spin-2 fields (gravity), we believe that the validity of our results for all $s$ may find applications in related areas. In particular, the case of Dirac fields ($s=1/2$) may be relevant for Lorentzian index theory \cite{BarGauduchonMoroianu2005,BarStrohmaier2020}.}, this being an easy extension with the relevant structure present in the equations for any $s$.  Friedrich's estimates \cite{Fri03b}, valid near $i^0$ and $\scri^+$, exploit the specific lower order structure of the spin-$2$ equations, and are necessarily weaker in the sense that they allow polyhomogeneous solutions with logarithmic divergences at $\mathcal{I}^\pm$. Having extended these to all spins, we then proceed to adapt the ideas in \cite{Fri03b} to construct estimates near $\scri^-$. This is the main contribution of our paper. Specifically, we prescribe initial data for spin-$s$ fields on a portion of past null infinity which extends all the way to $i^0$, and on an outgoing null hypersurface $\underline{\mathcal{B}}_\varepsilon$ emanating from $\scri^-$ (see \Cref{OptimalExistenceDomain}), and follow Luk's strategy\footnote{Luk's approach, for the Einstein equations, relies on the well-known observation that when written in a double null gauge, these form a symmetric hyperbolic system with a hierarchical structure. A similar hierarchical structure, for the spin-$s$ equations, is key to obtaining the estimates in \cite{Fri03b} as well as the ones presented in this paper.} \cite{Luk12} to construct an optimal existence domain for the characteristic initial value problem. In this domain we prove Gr\"onwall-type estimates for the solution by carefully arranging the sign of a suitable constant, which we achieve by taking sufficiently many derivatives in the boundary-defining variable $\tau = t/r$.
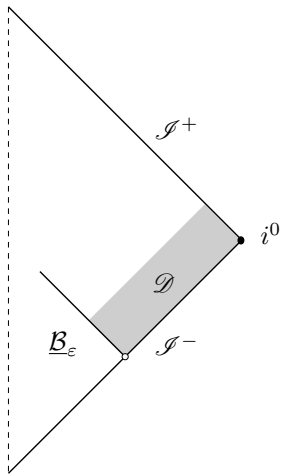
\begin{figure}[h]
\begin{center}
\includegraphics[width=0.45\textwidth]{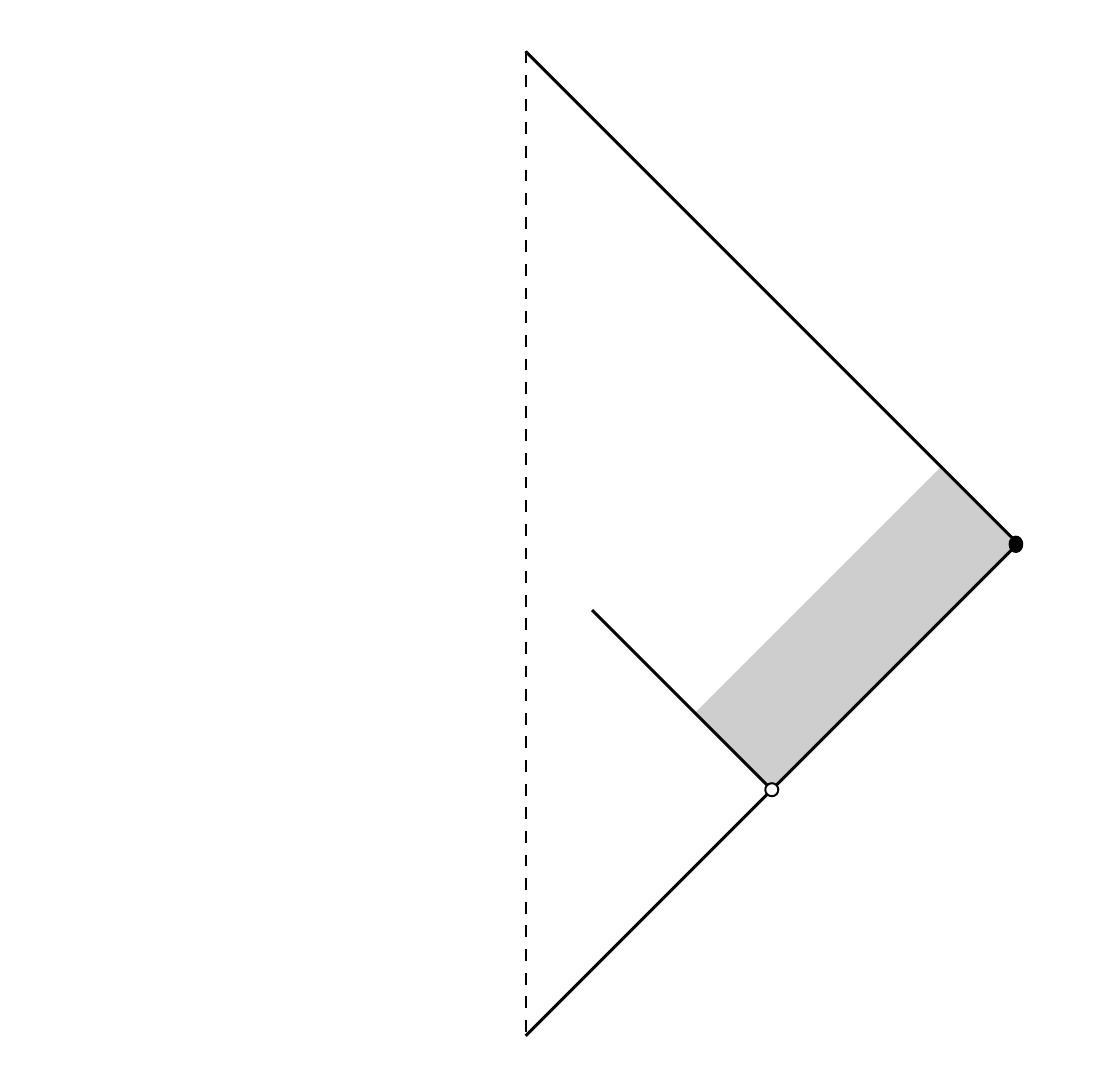}
\put(-50,135){$\mathscr{I}^+$}
\put(-50,53){$\mathscr{I}^-$}
\put(-10,95){$i^0$}
\put(-90,55){$\underline{\mathcal{B}}_\varepsilon$}
\put(-51,77){$\mathscr{D}$}
\end{center}
\caption{The Penrose diagram of Minkowski space showing the existence domain $\mathscr{D}$. The main technical challenge in this analysis is the degeneracy of the spin-$s$ equations at $i^0$.}
\label{OptimalExistenceDomain}
\end{figure}

The necessity for weaker estimates near $i^0$ reflects the general fact that solutions there acquire polyhomogeneous expansions \cite{Fri98b}. The computation of these expansions is done by postulating an Ansatz in terms of spin-weighted spherical harmonics, and then using the total characteristic nature of the cylinder $\mathcal{I}$ at spatial infinity. As a consequence, the spin-$s$ equations reduce to intrinsic transport equations on $\mathcal{I}$, and, at least formally, the terms in the expansion can be computed by solving Jacobi differential equations up and down the cylinder. The mentioned estimates then show that these expansions represent true solutions. The logarithmic terms in the solutions arise as logarithmic singularities in Jacobi polynomials of the second kind as one solves the ODEs on $\mathcal{I}$. An inspection of the computation of the asymptotic expansions reveals that, for our Ansatz\footnote{Our Ansatz, given in \Cref{Section:Expansions_In_Harmonics}, roughly postulates that the $p$-th derivative of the field contains no spherical harmonics of order higher than $p$. We make this assumption on the basis of simplifying the presentation. In the general case when all spherical harmonics are permitted, logarithmic terms will arise for each spherical harmonic with order $\geq p$.}, the logarithmic terms are, at each order in the expansion, associated to a specific ``top order" spherical harmonic. Moreover, at each order these logarithmic terms are regularized in a very specific way by multiplication by a polynomial expression which vanishes at $\mathcal{I}^\pm$. The higher
the order in the expansion, the higher the order of this smoothing polynomial---consequently, the logarithmic divergences become milder as one looks higher in the expansion.

\subsection*{Main result}

The \textbf{main result} of this article is as follows. Given asymptotic characteristic initial data on $\mathscr{I}^-$ for the massless spin-$s$ equations which possesses a regular asymptotic expansion towards $\mathcal{I}^-$, there exists a unique solution to the spin-$s$ equations in a neighbourhood of spatial infinity which contains pieces of both past and future null infinity, as shown in \Cref{OptimalExistenceDomain}, and moreover, one is able to control the regularity of the solution at $\scri^+$ in terms of specific properties of the data on $\scri^-$. The regularity of the solution at $\scri^+$ is determined not solely by the regularity of the data on $\scri^-$, but also by its multipolar structure at the critical set $\mathcal{I}^-$.  In particular, simply requiring additional regularity of the characteristic data does not necessarily result in improved behaviour at future null infinity. This is consistent with previous results obtained in the analysis of the corresponding Cauchy problem \cite{Fri03b,Val03a}. A detailed statement of our main result is given in \Cref{MainTheorem}.

\subsubsection*{Strategy of the proof}
As mentioned, in our analysis we use Friedrich's cylinder at spatial infinity, the F-gauge. In this gauge the causal domain depicted in \Cref{OptimalExistenceDomain} corresponds to the (light and dark) grey area in \Cref{Fig:CylinderIntro}. We divide this domain in two subdomains, the \emph{lower domain} $\underline{\mathcal{N}}_\epsilon$ in light grey, and the \emph{upper domain} $\mathcal{N}_1$ in dark grey, which are separated by a spacelike hypersurface $\mathcal{S}_{-1+\varepsilon}$ terminating at the cylinder at spatial infinity $\mathcal{I}$. \textbf{On the upper domain} $\mathcal{N}_1$ we look for solutions which can be written as a formal asymptotic expansion around $\mathcal{I}$ plus a remainder term. By using the total characteristic nature of $\mathcal{I}$, the terms in this formal expansion are, at least in principle, explicitly computable, and the regularity of the expansion can be controlled by fine-tuning the multipolar structure of the initial data. For the remainder, by generalizing the estimates of \cite{Fri03b}, we ensure its control all the way up to $\mathscr{I}^+$ in terms of the regularity of the Cauchy data. There is a connection here between the regularity of the remainder and the order of the formal expansion; the remainder becomes more regular as the order of the expansion increases. \textbf{On the lower domain} $\underline{\mathcal{N}}_\varepsilon$, we again look for solutions in the form of an asymptotic expansion plus a remainder. Here, however, the geometric/structural properties of the setting require an expansion with respect to the boundary defining function of past null infinity, $1+\tau = 1+ t/r$, rather than that of the cylinder $\mathcal{I}$. For simplicity, we assume that this expansion takes the form of a regular Taylor expansion in powers of $1+\tau$. In particular, we assume the data on $\scri^-$ is sufficiently regular. To control the remainder, we construct estimates using an adaptation of the techniques in \cite{Fri03b}. This construction proceeds in two stages: first, by making use of a bootstrap assumption on the radiation field, we estimate the other $2s-1$ components of the spin-$s$ field; in the second stage, we prove the bootstrap bound on the radiation field (recall that the radiation field encodes the freely specifiable data for the spin-$s$ equations on $\mathscr{I}^-$). As we mentioned in the introduction, these estimates are obtained by arranging the sign of a certain constant to be negative, which is achieved by taking derivatives in $\tau$ and carefully using the lower order structure of the equations. As in the case of the upper domain, the regularity of this remainder depends on the order of the asymptotic expansion in a similar way. \textbf{Finally}, we stitch together the solutions on the lower and upper domains by ensuring that the lower solution has enough control at the spatial hypersurface $\mathcal{S}_{-1 + \epsilon}$ to be able to apply our estimates in the upper domain. We thus obtain a statement controlling the solution up to future null infinity in terms of the asymptotic characteristic initial data on past null infinity. The existence and uniqueness of solutions in each domain is proved using standard last slice arguments.

\begin{figure}[h]
\begin{center}
\includegraphics[width=0.48\textwidth]{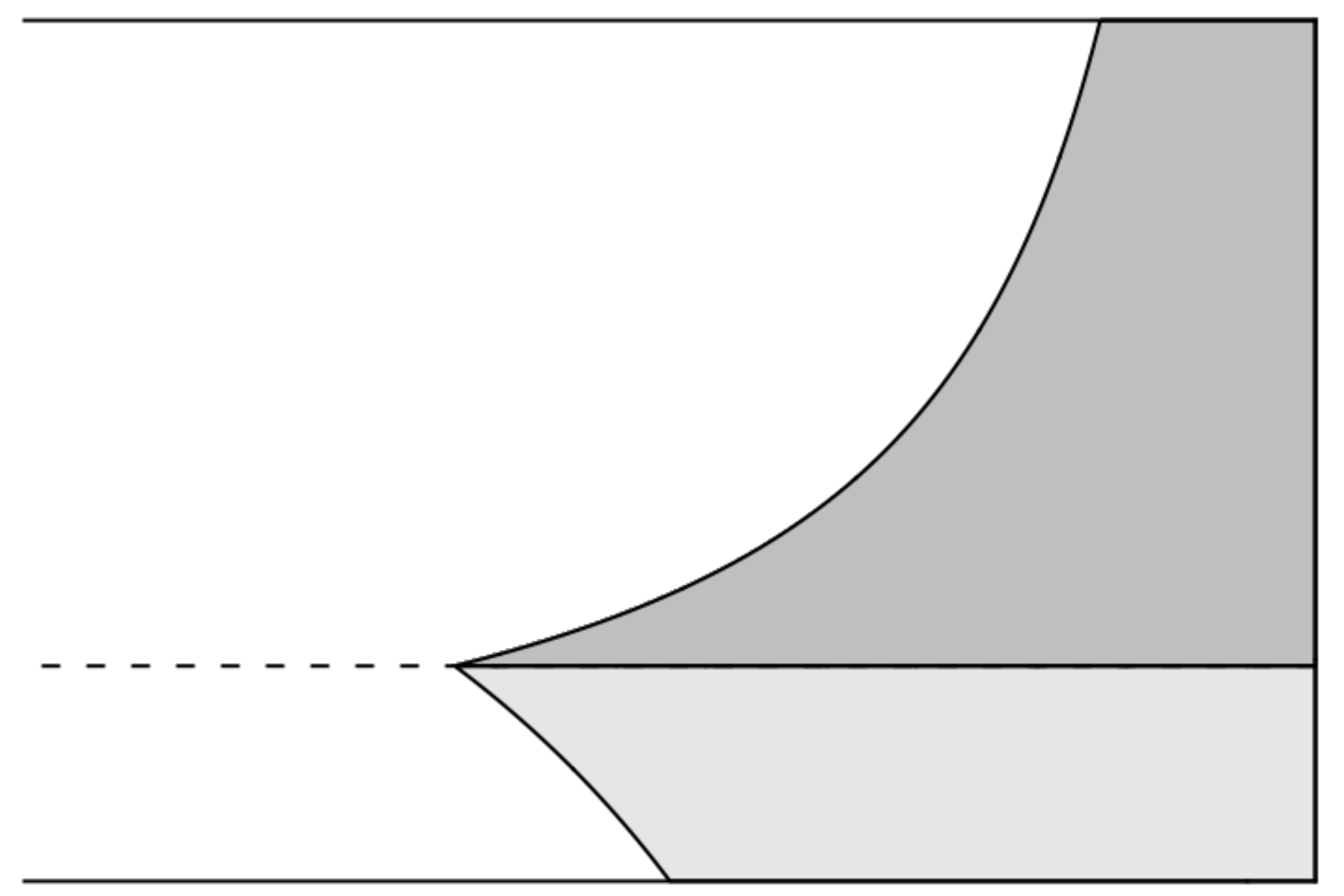}
\put(-100,147){$\mathscr{I}^+$}
\put(-100,-12){$\mathscr{I}^-$}
\put(0,70){$\mathcal{I}$}
\put(0,145){$\mathcal{I}^+$}
\put(-7,137){$\bullet$}
\put(0,-10){$\mathcal{I}^-$}
\put(-7,0){$\bullet$}
\put(-190,45){$\mathcal{S}_{-1+\varepsilon}$}
\put(-140,15){$\underline{\mathcal{B}}_\varepsilon$}
\put(-70,15){$\underline{\mathcal{N}}_\varepsilon$}
\put(-45,70){$\mathcal{N}_1$}
\put(-90,80){$\mathcal{B}_1$}

\end{center}
\caption{A depiction of the existence domain $\mathscr{D} = \underline{\mathcal{N}}_\epsilon \cup \mathcal{N}_1$ in the F-gauge. This is not a Penrose diagram but rather a coordinate diagram. The domain $\mathscr{D}$ consists of the union of two domains: the lower domain $\underline{\mathcal{N}}_\varepsilon$, which arises as the development of an asymptotic characteristic initial value problem, and which is bounded in the future by a spacelike hypersurface $\mathcal{S}_{-1+\varepsilon}$ which reaches the cylinder at spatial infinity $\mathcal{I}$; and the upper domain $\mathcal{N}_1$, which arises as the development of a standard Cauchy initial value problem from data on  $\mathcal{S}_{-1+\varepsilon}$. The estimates constructed in the main text will control the solution up to $\mathscr{I}^+$ despite the degeneracy of the equations at $\mathcal{I}^\pm$. }
\label{Fig:CylinderIntro}
\end{figure}
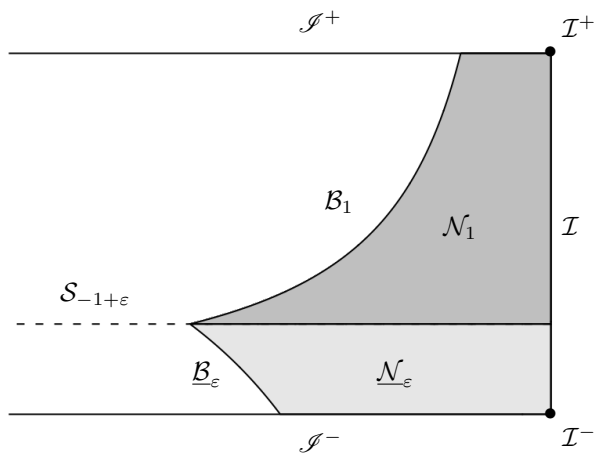

\subsection*{Outline of the article}

In \Cref{Section:GeometricSetUp} we provide a succinct discussion of Friedrich's framework of the cylinder at spatial infinity as well as lay out general properties of the spin-$s$ field equations. In \Cref{Section:ForwardEstimates} we construct the estimates in the upper domain, where, along the way, we extend Friedrich's estimates \cite{Fri03b} to arbitrary spins $s \in \frac{1}{2} \mathbb{N}$. In \Cref{Section:CharacteristicProblem} we discuss the set-up of the asymptotic characteristic initial value problem. Here we also discuss the difficulties that arise when attempting to analyse the behaviour of solutions near spatial infinity. \Cref{Section:BackwardEstimates} contains the main insight of our paper and provides the construction of our estimates in the lower domain. In \Cref{Section:FromPastToFuture} we combine the estimates in the lower and upper domains to establish the control of solutions at future null infinity in terms of past asymptotic initial data. We conclude in \Cref{Section:Conclusions} with some remarks and prospective directions for future research. In an effort to be accessible to a wider audience, we also include four appendices containing material essential for the analysis but whose inclusion in the main text would hinder the flow of the reading. Much (but not all) of this material is scattered throughout the various references noted in the introduction, and we hope that collecting it here will make the subject more approachable. \Cref{Appendix:SU2} discusses properties of $\mathrm{SU}(2)$ which are relevant to the construction of the estimates in the main text. \Cref{Appendix:SpinSEqns} provides a detailed derivation of the spin-$s$ field equations in the F-gauge. \Cref{Appendix:SolutionJets} gives a detailed account of the construction of asymptotic expansions in a neighbourhood of the cylinder at spatial infinity---these expansions are central for the construction of the solution in the upper domain. Finally, \Cref{Appendix:ExpansionsNullInfinity} presents the construction of asymptotic expansions in a neighbourhood of past null infinity---these expansions are key to the construction of the solution in the lower domain.

\subsection*{Conventions and notation}

Our conventions are consistent with \cite{CFEBook}. In particular, our metric signature is $(+, -, -, -)$, and the Riemann curvature tensor associated with the Levi-Civita connection of a metric $g_{ab}$ is defined by \mbox{$[\nabla_a, \nabla_b] u^d = R^d_{\phantom{d}cab} u^c$}. For a given spin dyad $\{ o^A, \iota^A \}$ with $o_A \iota^A = 1$, we write $\epsilon_\bmA{}^A$, $\bmA \in \{ 0, 1 \}$, to denote $\epsilon_0{}^A = o^A$ and $\epsilon_1{}^A = \iota^A$. Spinorial indices are raised and lowered using the antisymmetric $\epsilon$-spinor $\epsilon_{AB} = o_A \iota_B - \iota_A o_B$ (with inverse $\epsilon^{AB} = o^A \iota^B - \iota^A o^B$), e.g. $\xi_B = \xi^A \epsilon_{AB}$, using the convention that contracted indices should be ``adjacent, descending to the right". As usual, the spacetime metric $g_{ab}$ decomposes as $g_{ab} = \epsilon_{AB} \bar{\epsilon}_{A' B'}$, where $\bar{\epsilon}_{A'B'} = \overline{\epsilon_{AB}}$. The spin dyad $\epsilon_\bmA{}^A$ gives rise to a tetrad of null vectors $\bme_{\bmA \bmA'} = \bme_{\bmA \bmA'}{}^{AA'} \partial_{AA'} = \epsilon_\bmA{}^A \bar{\epsilon}_{\bmA'}{}^{A'} \partial_{A A'}$. The spin connection coefficients $\Gamma_{\bmA \bmA'}{}^\bmB{}_\bmC$ are then defined as $\Gamma_{\bmA \bmA'}{}^\bmB{}_\bmC = - \epsilon_\bmC{}^Q \nabla_{\bmA \bmA'} \epsilon^\bmB{}_Q$, where $\nabla_{\bmA \bmA'} = \bme_{\bmA \bmA'}{}^{AA'} \nabla_{A A'}$. In Newman--Penrose language, the spin connection coefficients $\Gamma_{\bmA \bmA' \bmB \bmC} = \Gamma_{\bmA \bmA' \bmC \bmB}$ are exactly the spin coefficients $\gamma_{\bmA \bmA' \bmB \bmC}$ (see \cite{PenRin86}, Summary of Vol. 1). When integrating over $\mathrm{SU}(2)$, $\mu$ denotes the normalized Haar measure on $\mathrm{SU}(2)$. Throughout the majority of the paper we will be working on \emph{rescaled} Minkowski space (in the F-gauge), and we shall denote objects (fields, connections, spin dyad) on this spacetime plainly. When working on physical Minkowski space, we will denote objects with a tilde, e.g. $\tilde{\phi}$.

\section{Geometric setup}
\label{Section:GeometricSetUp}
In this section we discuss the general geometric setup for our analysis in a neighbourhood of spatial and null infinity.

\subsection{Point representation of spatial infinity: the Penrose gauge}

Let $(y^\mu)$ be Cartesian coordinates on the Minkowski spacetime
$(\mathbb{R}^4,\tilde{\bmeta})$ with the standard metric  $\tilde{\bmeta} = \tilde{\eta}_{\mu \nu} \mathbf{d} y^\mu \otimes \mathbf{d} y^\nu$, where $\tilde{\eta}_{\mu \nu} = \operatorname{diag}(1, -1, -1, -1)$. In spherical coordinates we have that 
\[ 
\tilde{\bmeta} = \mathbf{d} t\otimes \mathbf{d} t- \mathbf{d} r\otimes \mathbf{d} r - r^2 \bmsigma, 
\]
where $t = y^0$, $r^2 = \sum_{i=1}^3 y_i^2$, and $\bmsigma$ denotes the standard metric on $\mathbb{S}^2$. On the  region
\[
\mathcal{N} \equiv \{ (y^\mu) \in \mathbb{R}^4 \; | \; \tilde{\eta}_{\mu\nu}
y^\mu y^\nu <0 \}
\]
which includes the asymptotic region around spatial infinity, we consider the coordinate inversion
\[
y^\mu \mapsto x^\mu = - \frac{y^\mu }{y_\nu y^\nu},
\]
and formally extend the domain of validity of the coordinates
$(x^\mu)$ to include the set $\mathscr{I}' = \{
x^\mu x_\mu =0 \} = \{ y^\mu y_\mu = - \infty \}$. The set $\mathscr{I}'$ formally forms the part of the boundary of $\mathcal{N}$ at infinity. Defining $\rho^2 \equiv \sum_{i=1}^3 x_i^2 $, the metric $\tilde{\bmeta}$ reads
\begin{equation}
\label{inverted_Minkowski_metric}
 \tilde{\bmeta} = \frac{1}{(x^\sigma x_\sigma)^2}
   \tilde{\eta}_{\mu\nu} \mathbf{d} x^\mu \otimes \mathbf{d} x^\nu = \frac{1}{\left( \rho^2 -(x^0)^2
   \right)^2}\left( \mathbf{d} \, (x^0)\otimes \mathbf{d} \, (x^0)  -
   \mathbf{d} \rho \otimes \mathbf{d} \rho  -\rho^2 \bmsigma \right),
\end{equation}
where on $\mathcal{N}$ the inverted coordinates $x^0$ and $\rho$ are given in terms of $t$ and $r$ by
\[ 
\rho = \frac{r}{r^2 - t^2}, \qquad x^0 = \frac{t}{r^2 - t^2}. 
\]
The right-hand side of \eqref{inverted_Minkowski_metric} is a conformally rescaled Minkowski metric in the new coordinates $(x^\mu)$, where one notices that on $\mathcal{N}$ the conformal factor 
\[
\Xi = - x^\mu x_\mu = - \frac{1}{y^\mu y_\mu} = \rho^2 - (x^0)^2
\]
extends smoothly to $\scri'$. Moreover, the metric $\eta'_{\mu \nu} \equiv \Xi^2 \tilde{\eta}_{\mu \nu}$ (which is just the Minkowski metric again) does too, and now contains at its origin the point
\[
i^0 \equiv \{ (x^\mu)\in \mathbb{R}^4 \; | \; x^\mu=0  \}
\]
called \emph{spatial infinity}, formally a point on the $\scri'$ part of the boundary of $\mathcal{N}$. The sets
\[
\mathscr{I}^\pm \equiv \mathscr{I}' \cap \{ (x^\mu) \in \mathbb{R}^4
\; |\; \eta_{\mu\nu}x^\mu x^\nu =0, \; \pm x^0 >0 \}
\]
are called the future and past \emph{null infinity}, respectively, of Minkowski space (near $i^0$), and are the hypersurfaces which null geodesics in $\mathcal{N}$ approach asymptotically (see \Cref{AsymptoticRegion}).

\begin{figure}[t]
\begin{center}
\includegraphics[width=0.6\textwidth]{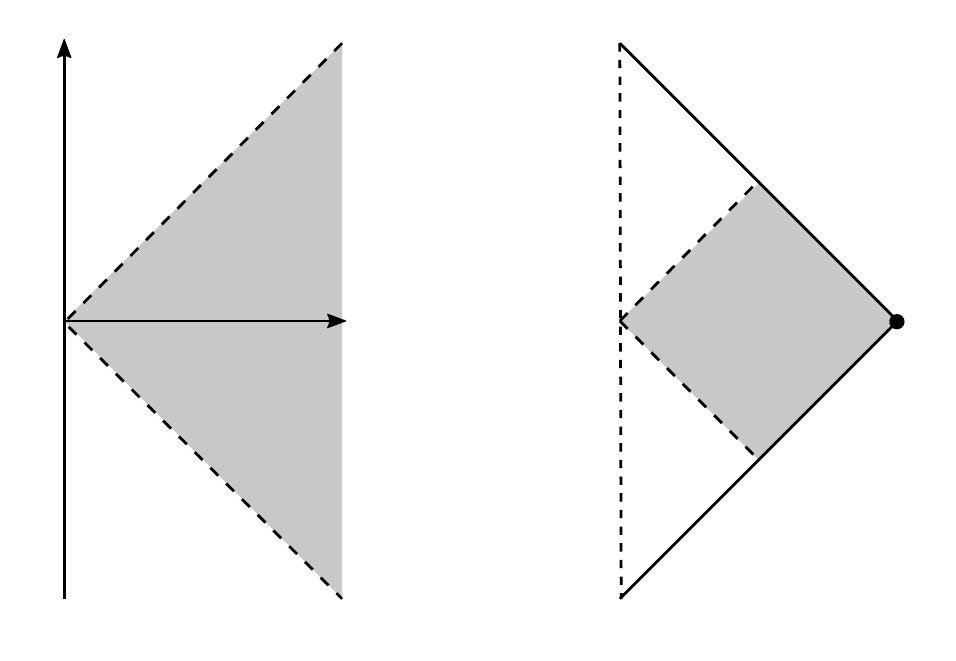}
\put(-58,92){$\mathcal{N}$}
\put(-35,120){$\scri^+$}
\put(-35,65){$\scri^-$}
\put(-10,92){$i^0$}
\put(-200,110){$\mathcal{N}$}
\put(-165,93){$r$}
\put(-250,180){$t$}
\end{center}
\caption{The point representation of spatial infinity $i^0$ with its neighbourhood $\mathcal{N}$.  The left diagram depicts the region $\mathcal{N}$ in physical space. The right diagram depicts a conformal extension of the region which includes spatial infinity and a portion of null infinity.}\label{AsymptoticRegion}
\end{figure}
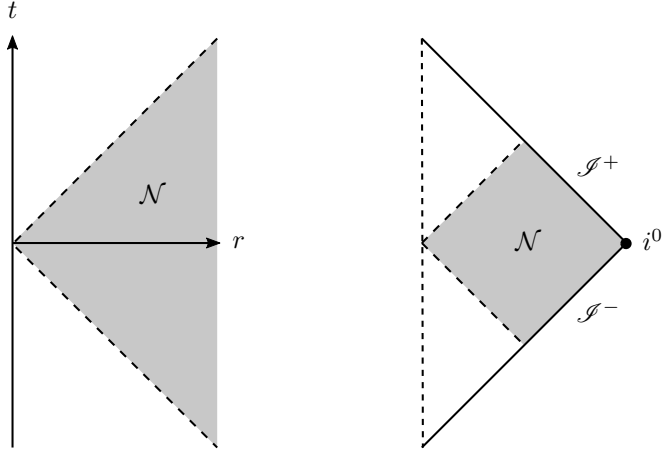

\subsection{Cylinder representation of spatial infinity: the F-gauge}
\label{Section:CylinderMinkowski}
In the construction above, the endpoint $i^0$ of all \emph{spacelike} geodesics in Minkowski space is collapsed to a point, which for our purposes turns out to be somewhat unnatural (compare this, for example, to $\scri^\pm$, the end\emph{surfaces} of null geodesics). Instead, it is useful to blow up $i^0$ to a cylinder as follows. Set\footnote{Note that in terms of the physical coordinates $t$ and $r$, the new coordinate $\tau$ is given by $\tau = t/r$.} $x^0 = \rho \tau$, and define a new conformal factor\footnote{In terms of physical coordinates $t$ and $r$, this new conformal factor is simply $\Theta = 1/r$. This conformal scale in fact appears in the literature on the theory of conformal scattering \cite{MasNic04,MasNic12,NicTau22}, albeit without the use of the F-coordinates $\tau$ and $\rho$. It is the definition of these coordinates that is at the core of the blow-up of $i^0$ to a cylinder.}
\[ 
\Theta \equiv \frac{1}{\rho} \Xi = \rho \left( 1 - \tau^2 \right). 
\]
We then define the unphysical metric
\begin{align*}
    \bmeta & \equiv \Theta^2 \tilde{\bmeta} \\
    & = \frac{1}{\rho^2} \left( \rho^2 \mathbf{d} \tau \otimes \mathbf{d}\tau + \tau \rho \left( \mathbf{d} \tau \otimes \mathbf{d} \rho + \mathbf{d} \rho \otimes \mathbf{d} \tau \right) - (1- \tau^2) \, \mathbf{d} \rho\otimes\mathbf{d} \rho  - \rho^2 \bmsigma \right).
\end{align*}
We call the coordinates $\tau$ and $\rho$ the \emph{F-coordinates} on Minkowski space. In terms of $\tau$ and $\rho$, the region $\mathcal{N}$ then reads
\[
\mathcal{N} =\left\{ (\tau, \rho) \times \mathbb{S}^2 \, | \, -1 <\tau < 1, \; \rho>0 \right\} \simeq (-1, 1)_\tau \times (0, \infty)_\rho \times \mathbb{S}^2
\]
and 
\[
\mathscr{I}^\pm = \left\{ (\tau, \rho) \times \mathbb{S}^2 \, | \, \tau = \pm 1, \; \rho>0 \right\} \simeq (0, \infty)_\rho \times \mathbb{S}^2. 
\]
It is also convenient to introduce the sets
\begin{align*}
&\mathcal{I} \equiv \left\{ (\tau, \rho) \times \mathbb{S}^2 \, | \,  |\tau| <1, \; \rho=0  \right\} \simeq (-1,1)_\tau \times \mathbb{S}^2, \\
&\mathcal{I}^0 \equiv \left\{(\tau, \rho) \times \mathbb{S}^2 \, | \,  \tau =0, \; \rho=0 \right\} \simeq \mathbb{S}^2, \\
&\mathcal{I}^\pm \equiv \left\{(\tau, \rho) \times \mathbb{S}^2 \, | \, \tau =\pm 1, \; \rho=0  \right\} \simeq \mathbb{S}^2.
\end{align*}
In this representation the point $i^0 = \{ x^\mu = 0 \} = \{\rho = 0, \, x^0 = 0\}$ has therefore been blown to the cylinder $\mathcal{I}$. Note that although $\mathcal{I}$ is at a finite $\rho$-coordinate, it is still at infinity with respect to the metric $\bmeta$. We call $\mathcal{I}$ the \emph{cylinder at spatial infinity} and $\mathcal{I}^\pm$ the \emph{critical sets}.

Let us denote the standard initial hypersurface by
\[
\tilde{\mathcal{S}} \equiv \{ (y^\mu)\in \mathbb{R}^4 \; |\;  t=y^0=0 \}
\]
and set
\[
\mathcal{S} \equiv \tilde{\mathcal{S}}\cup \mathcal{I}, \qquad
\overline{\mathcal{I}} \equiv \mathcal{I}\cup \mathcal{I}^+\cup \mathcal{I}^-, \qquad
\overline{\mathcal{N}} \equiv \mathcal{N} \cup \mathscr{I}^+\cup
\mathscr{I}^-\cup \overline{\mathcal{I}},
\]
and consider $\rho$, $\tau$ and the spherical coordinates $\bmsigma$ as
coordinates on 
\[
\overline{\mathcal{N}} \simeq [-1,1]_\tau \times [0,\infty)_\rho \times \mathbb{S}^2_\bmsigma.
\]
In these coordinates the
expressions for $\Theta$ and $\Xi=\rho \Theta$, and the coordinate
expression for the inverse metric $\bmeta^\sharp$ all extend smoothly to all points of
$\overline{\mathcal{N}}$. Observe, however, that $\bmeta$ itself degenerates at $\rho = 0$, so that $\bmeta$ is only smooth on
$\mathcal{N}\cup\mathscr{I}^+\cup\mathscr{I}^-$. We call this conformal completion of $\mathcal{N}$, together with the F-coordinates $\tau$ and $\rho$, the \emph{F-gauge} \cite{Fri98a,Fri98b}.

\begin{figure}[t]
\begin{center}
\includegraphics[width=\textwidth]{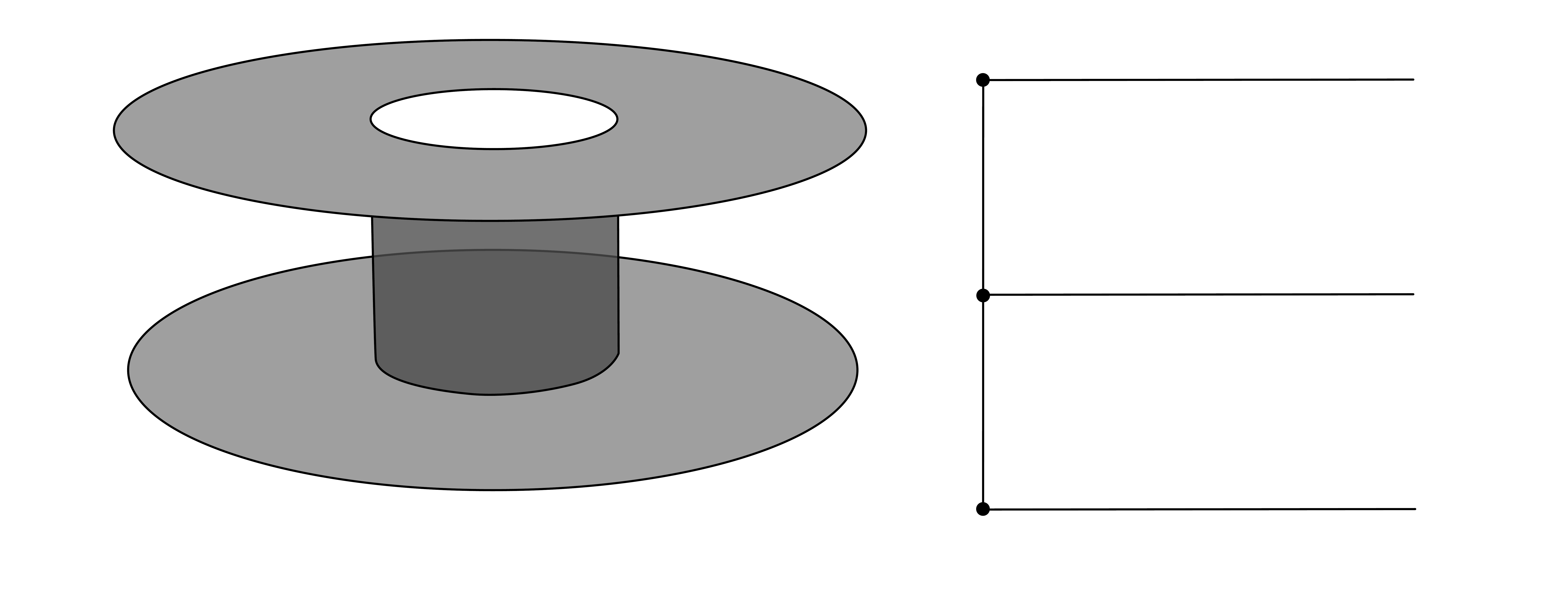}
\put(-181,142){$\mathcal{I}^+$}
\put(-181,19){$\mathcal{I}^-$}
\put(-181,80){$\mathcal{I}^0$}
\put(-100,150){$\scri^+$}
\put(-100,10){$\scri^-$}
\put(-100,88){$\tau = 0$}
\put(-250,125){$\scri^+$}
\put(-250,55){$\scri^-$}
\put(-320,65){$\mathcal{I}$}
\end{center}
\caption{The cylinder at spatial infinity in the F-gauge. On the left, a depiction suppressing one angular dimension. On the right, a cross section. This type of 2-dimensional diagram is a convenient way of depicting the locations of various hypersurfaces in the F-gauge. Observe that this is a \emph{coordinate diagram} and not a Penrose diagram. In particular, null geodesics do not correspond to lines with slope of $\pm 45^\circ$. }
\label{Fig:CylinderBackground}
\end{figure}
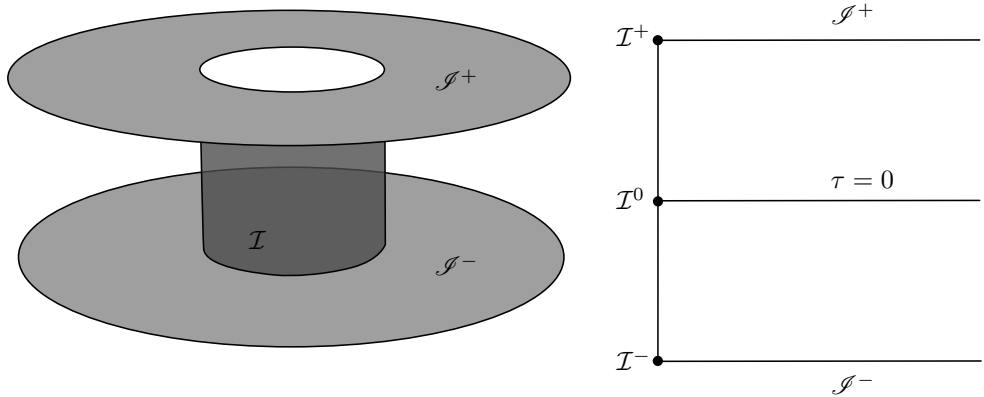

\subsection{Null frame near $\mathcal{I}^0$}
\label{Section:NullFrameNearI0}

In order to perform estimates, we will write the field equations in terms of a null frame $\{ \bme_\bma = \bme_\bma^\mu \bmpartial_\mu \}$. We make use of the Infeld-van der Waerden symbols to \emph{spinorise} the frame indices, $\bme_\bma \mapsto \bme_{\bmA\bmA'}^\mu \bmpartial_\mu$, and introduce a frame $\bme_{\bmA\bmA'}$ satisfying
\[
\bmeta(\bme_{\bmA\bmA'},\bme_{\bmB\bmB'}) =
\epsilon_{\bmA\bmB}\epsilon_{\bmA'\bmB'}.
\]
Specifically, we choose
\begin{equation}
\bme_{\bmzero\bmzero'} = \frac{1}{\sqrt{2}}\left(
(1-\tau)\partial_\tau + \rho \partial_\rho  \right), \qquad \bme_{\bmone\bmone'} =\frac{1}{\sqrt{2}}\left(
(1+\tau)\partial_\tau - \rho \partial_\rho  \right)
\label{FramesMinkowski} 
\end{equation}
and complex vector fields $\bme_{\bmzero\bmone'}$ and $\bme_{\bmone\bmzero'}$ which are tangent to the spheres $\mathbb{S}^2_{\tau, \rho}$ of constant $\tau$ and $\rho$. However, non-vanishing smooth vector fields on $\mathbb{S}^2$ cannot be defined globally, so we lift the dimension of the spheres by considering \emph{all possible choices} of such vector fields. Specifically, for each $(\tau, \rho)$ we consider the Hopf map $p : \mathrm{SU}(2) \simeq \mathbb{S}^3 \to \mathbb{S}^2 \simeq \mathrm{SU}(2) / \mathrm{U}(1)$, which defines a principal $\mathrm{U}(1)$-bundle over $\mathbb{S}^2$. This map induces a group of rotations $\bme_{\bmzero\bmone'} \mapsto e^{i \alpha} \bme_{\bmzero\bmone'}$, $ \alpha \in \mathbb{R}$, which leaves the tangent bundle of $\mathbb{S}^2_{\tau, \rho}$ invariant. By lifting each $2$-sphere along the Hopf map, we obtain a 5-dimensional manifold $[-1, 1]_\tau \times (0, \infty)_\rho \times \mathrm{SU}(2)$; note that this is a 5-dimensional submanifold of the bundle of normalised spin frames on $\overline{\mathcal{N}} \setminus \overline{\mathcal{I}}$.

The geometric structures on $\overline{\mathcal{N}}\setminus \overline{\mathcal{I}}$ can be naturally lifted to the 5-dimensional manifold introduced in the previous paragraph. Considering $\tau$, $\rho$ and $t^\bmA{}_\bmB \in \mathrm{SU}(2)$ as coordinates on this manifold, the lifted vector fields in \eqref{FramesMinkowski} have the same coordinate expressions as before. Allowing $\rho$ to take the value $0$, one can extend all geometric structures to include this value. Thus, one considers $(\tau,\rho,t^\bmA{}_\bmB)$ as coordinates on the extended manifold, which we again denote by
\[  
\overline{\mathcal{N}} \equiv [-1,1]_\tau \times [0,\infty)_\rho \times
  \mathrm{SU}(2)_{t^\bmA_{\phantom{\bmA}\bmB}}.
\]
In this setting now
\[
\mathcal{I} \simeq (-1,1)_\tau \times \mathrm{SU}(2), \quad
\mathcal{I}^0 \simeq \mathrm{SU}(2),\quad \mathcal{I}^\pm \simeq
\mathrm{SU}(2), \quad \text {and} \quad \mathscr{I}^\pm \simeq (0, \infty)_\rho \times
\mathrm{SU}(2)
\]
as subsets of $\overline{\mathcal{N}}$. To complete this construction, it remains to choose a frame of complex vector fields on $\mathrm{SU}(2)$. We set
\[
\bme_{\bmzero\bmone'} = - \frac{1}{\sqrt{2}} \bmX_+ \quad \text{and} \quad
\bme_{\bmone\bmzero'} =- \frac{1}{\sqrt{2}} \bmX_-,
\]
where $\bmX_\pm$ are vector fields on $\mathrm{SU}(2)$ defined
in \Cref{Appendix:SU2} (briefly, they are complex linear combinations of two of the three left-invariant vector fields on $\mathrm{SU}(2)$). The connection form on the bundle of normalised spin
frames then defines connection coefficients $\Gamma_{\bmA \bmA' \bmB \bmC } = \Gamma_{\bmA \bmA' \bmC \bmB }$ with respect to the frame
$\bme_{\bmA\bmA'}$, the only independent non-zero values of which are
\[
     \Gamma_{\bmzero\bmzero'\bmzero\bmone} = \Gamma_{\bmone\bmone'\bmzero\bmone} = - \frac{1}{2 \sqrt{2}}.
\]
In Newman--Penrose language, these correspond to the spin coefficients $\varepsilon$ and $\gamma$, which are gauge quantities with respect to tetrad rescalings in the compacted spin coefficient formalism---see §4.12, \cite{PenRin84}.

\subsection{Spin-$s$ equations in the F-gauge}
Let
\[
\phi_{A_1\dots A_{2s}} =\phi_{(A_1\dots A_{2s})}
\]
be a totally symmetric spinor of valence $2s$, with $s \in \frac{1}{2} \mathbb{N}$, where we take $\mathbb{N} = \{ 1, \, 2, \, 3, \, \ldots \}$, i.e. we do not consider the spin zero case (the wave equation). The massless spin-$s$ equations for $\phi$ on $(\mathcal{N}, \boldsymbol{\eta})$ then read
\begin{equation}
\nabla^Q{}_{A'} \phi_{QA_1\dots A_{2s-1}}=0.
    \label{MasslessSpinS}
\end{equation}
It is worth recalling that the spin-$s$ equations are conformally invariant. That is, $\phi_{A_1\ldots A_{2s}}$ satisfies \eqref{MasslessSpinS} on $(\mathcal{N}, \boldsymbol{\eta})$, where $\boldsymbol{\eta} = \Theta^2 \tilde{\boldsymbol{\eta}}$ is the unphysical Minkowski metric, if and only if the \emph{physical} spin-$s$ field
\[ \tilde{\phi}_{A_1 \ldots A_{2s}} = \Theta \phi_{A_1 \ldots A_{2s}} \]
satisfies
\[ \tilde{\nabla}^Q{}_{A'} \tilde{\phi}_{QA_1 \ldots A_{2s-1}} = 0 \]
on $(\mathcal{N}, \tilde{\boldsymbol{\eta}})$ \cite{PenRin84}. Conclusions about $\phi$ can therefore be easily translated into conclusions about the physical field $\tilde{\phi}$.

In the F-gauge equation \eqref{MasslessSpinS} is equivalent to the system
\begin{subequations}
\begin{eqnarray}
&& A_k \equiv (1+\tau)\partial_\tau \phi_{k+1} -\rho \partial_\rho \phi_{k+1} -\bmX_+ \phi_{k} +(k+1-s)\phi_{k+1} =0, \label{AEqnMain} \\
&& B_k \equiv (1-\tau)\partial_\tau \phi_k + \rho\partial_\rho\phi_k -\bmX_-\phi_{k+1} +(k-s)\phi_k =0, \label{BEqnMain}
\end{eqnarray}
\end{subequations}
for $k=0, \, \ldots, \, 2s-1$, where $\phi_k$ are the $2s$ independent components of the spinor $\phi_{A_1\cdots A_{2s}}$ with respect to a spin dyad $\{ o^A, \iota^A \}$ corresponding to the frame $\{ \bme_{\bmA \bmA'} \}$. The derivation of equations \eqref{AEqnMain} and \eqref{BEqnMain} is given in \Cref{Appendix:SpinSEqns}. We note here that the component $\phi_{2s} = \phi_{A_1 \ldots A_{2s}} \iota^{A_1} \ldots \iota^{A_{2s}}$ is the \emph{outgoing radiation field} on $\scri^+$, while $\phi_0 = \phi_{A_1 \ldots A_{2s}} o^{A_1} \ldots o^{A_{2s}}$ is the \emph{incoming radiation field} on $\scri^-$.

\section{Estimates near $\scri^+$}
\label{Section:ForwardEstimates}

In this section we construct estimates for the massless spin-$s$ equations which allow us to control the solutions up to and including $\mathscr{I}^+$ in terms of suitable norms on initial data prescribed on a Cauchy hypersurface. In doing so we recap the estimates of Friedrich \cite{Fri03b} and extend them to fields of arbitrary spin $s$.

\subsection{Geometric setup}
We use the F-gauge as discussed in \Cref{Section:CylinderMinkowski,Section:NullFrameNearI0}. Given $t \in (-1,1]$, $t > \tau_\star \in (-1, 1) $, and $\rho_\star>0$, one defines the sets
\begin{eqnarray*}
&& \mathcal{N}_t \equiv \bigg\{ (\tau, \rho, t^\bmA{}_\bmB) \, | \, \tau_\star \leq \tau \leq t, \; 0\leq \rho \leq
   \frac{\rho_\star}{1+\tau}, \; t^\bmA{}_\bmB \in \mathrm{SU}(2) \bigg\} \subset [-1, 1]_\tau \times [0, \infty)_\rho \times \mathrm{SU}(2), \\
&& \mathcal{S}_t \equiv \bigg\{ (\tau, \rho, t^\bmA{}_\bmB) \, | \, \tau=t, \; 0\leq \rho \leq \frac{\rho_\star}{1+t}, \; t^\bmA{}_\bmB \in \mathrm{SU}(2) \bigg\} , \\
&& \mathcal{B}_t \equiv \bigg\{(\tau, \rho, t^\bmA{}_\bmB) \, | \, \tau_\star \leq \tau \leq t, \; \rho = \frac{\rho_\star}{1+\tau}, \; t^\bmA{}_\bmB \in \mathrm{SU}(2) \bigg\}, \\
&& \mathcal{I}_t \equiv \bigg\{ (\tau, \rho, t^\bmA{}_\bmB) \, | \, \tau_\star \leq \tau \leq t, \; \rho=0, \; t^\bmA{}_\bmB \in \mathrm{SU}(2) \bigg\} ,
\end{eqnarray*}
and $\mathcal{S}_\star \equiv \mathcal{S}_{\tau_\star}$. A schematic depiction of these sets is given in \Cref{Fig:Cylinder_Forward_Estimates}.

\begin{figure}[H]
\begin{center}
\includegraphics[scale=0.4]{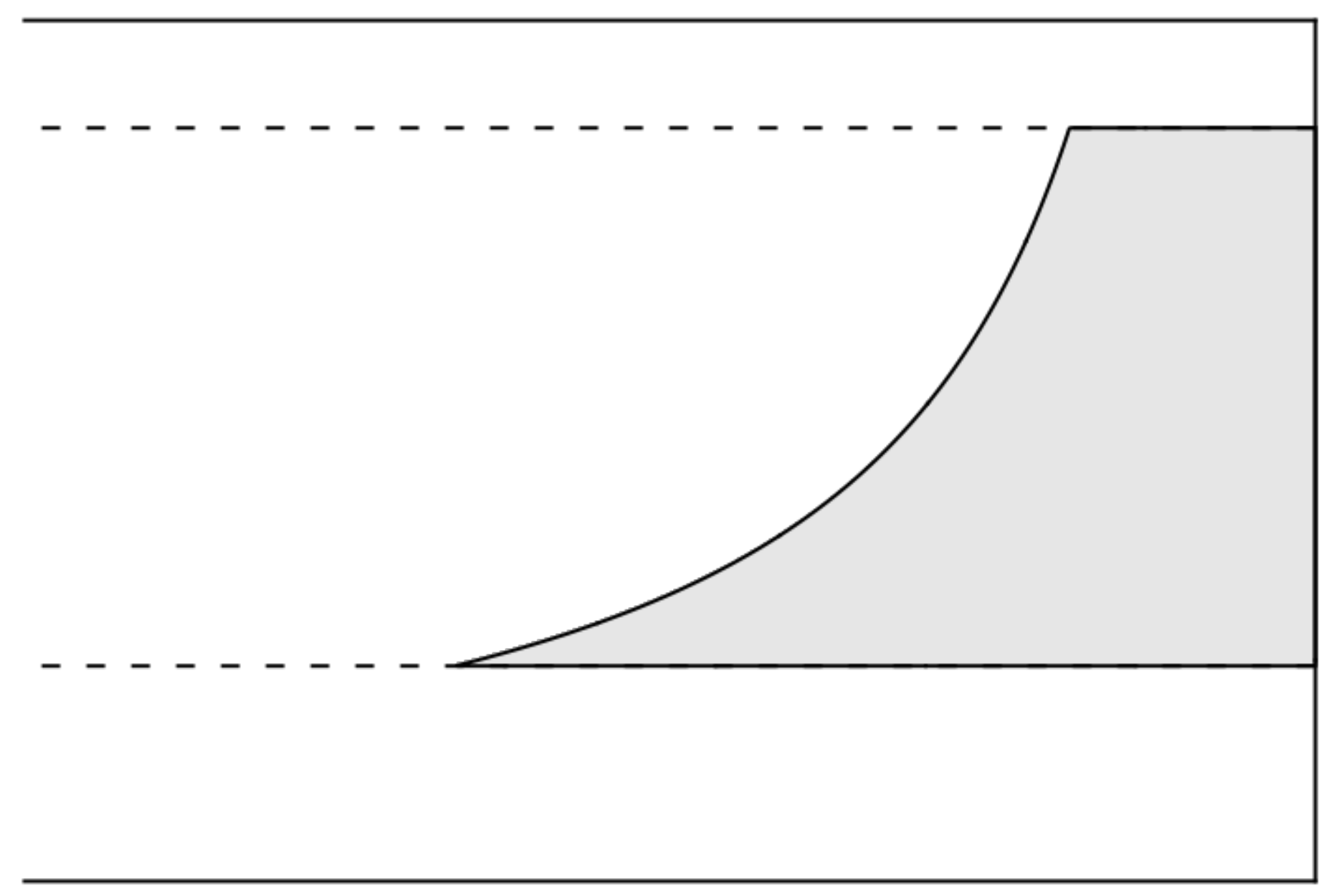}
\put(-120, 148){$\scri^+$}
\put(-120, -10){$\scri^-$}
\put(-40, 75){$\mathcal{N}_t$}
\put(2, 75){$\mathcal{I}_t$}
\put(2, 145){$\mathcal{I}^+$}
\put(-7, 139){$\bullet$}
\put(2, -7){$\mathcal{I}^-$}
\put(-7, 0){$\bullet$}
\put(-70, 25){$\mathcal{S}_\star$}
\put(-88, 75){$\mathcal{B}_t$}
\put(-28, 128){$\mathcal{S}_t$}
\end{center}
\caption{Schematic depiction of the domain used in the construction of estimates controlling the solutions to the spin-$s$ equations near $\mathscr{I}^+$ in terms of Cauchy initial data.}
\label{Fig:Cylinder_Forward_Estimates}
\end{figure}

\subsection{Construction of estimates}
Given non-negative integers $p$, $p'$, $q$, $q'$ and a multiindex $\alpha=(\alpha_1,\alpha_2,\alpha_3)$, we consider the differential operators 
\[
D \equiv D^{q,p,\alpha} \equiv \partial^q_\tau \partial^p_\rho \bmZ^\alpha \quad \text{and} \quad D' \equiv D^{q',p',\alpha} \equiv \partial^{q'}_\tau \partial^{p'}_\rho \bmZ^\alpha,
\]
where $\bmZ^\alpha$ denotes the vector fields on
$\mathrm{SU}(2)$ introduced in \Cref{Appendix:TechnicalLemmaSU2C}. We then have the following estimates controlling the solutions to the spin-$s$ equations up to $\mathscr{I}^+$.

\begin{proposition}
\label{thm:forward_estimates}
Let $\tau_\star \in (-1, 1)$, $t > \tau_\star$, and consider the field $\phi_{A_1\cdots A_{2s}}$ satisfying equations \eqref{AEqnAppendix} and \eqref{BEqnAppendix} in the region $\mathcal{N}_t$. Let $(p, m) \in \mathbb{N} \times \mathbb{N} $ be such that 
\[ p > m + s,\]
and suppose that 
\[ \sum_{k=0}^{2s} \int_{\mathcal{S}_\star} \sum_{ q' + p' + |\alpha| \leq m } |D'(\partial_\rho^p \phi_k)|^2 \, \d \rho \wedge \d \mu < \infty.  \]
Then there exists a constant $C_{p,m,s} > 0$ which is independent of $t$ such that for all $0 \leq k \leq 2s$
\begin{equation}
\label{Spin2BasicEstimate}
\| \partial_\rho^p \phi_k \|^2_{H^m(\mathcal{N}_t)} \equiv \int_{\mathcal{N}_t} \sum_{ q' + p' + |\alpha| \leq m } |D'(\partial^p_\rho \phi_k)|^2 \, \mathrm{dv} \leq C_{p,m,s} \sum_{k=0}^{2s} \int_{\mathcal{S}_\star} \sum_{ q' + p' + |\alpha| \leq m } |D'(\partial^p_\rho \phi_k)|^2 \, \d \rho \wedge \d \mu.
\end{equation}
\end{proposition}

\begin{proof} Applying $D$ to the equations \eqref{AEqnAppendix} and \eqref{BEqnAppendix}, multiplying by $\overline{D \phi_k}$ and $\overline{D \phi_{k+1}}$ respectively, and summing the real parts, one obtains the identity
\[  2 \operatorname{Re} \left( \overline{D \phi_{k+1}} D A_k[\phi] + \overline{D \phi_k} D B_k[\phi] \right) = 0.  \]
Commuting $D$ into the operators $A_k[\phi]$ and $B_k[\phi]$, this may be rewritten as
\begin{align}
\label{differentialidentity}
\begin{split}
    0 &= \partial_\tau \left( (1+ \tau) | D \phi_{k+1} |^2 + (1-\tau) |D \phi_k |^2 \right) + \partial_\rho \left( - \rho | D \phi_{k+1} |^2 + \rho |D \phi_k |^2 \right) \\
    &\phantom{=} - \bmZ^\alpha \bmX_+ ( \partial_\tau^q \partial_\rho^p \phi_k ) \bmZ^\alpha ( \partial_\tau^q \partial_\rho^p \bar{\phi}_{k+1} ) - \bmZ^\alpha (\partial_\tau^q \partial_\rho^p \phi_k ) \bmZ^\alpha \bmX_+ (\partial_\tau^q \partial_\rho^p \bar{\phi}_{k+1} ) \\
    &\phantom{=} - \bmZ^\alpha \bmX_- ( \partial_\tau^q \partial_\rho^p \bar{\phi}_k ) \bmZ^\alpha (\partial_\tau^q \partial_\rho^p \phi_{k+1} ) - \bmZ^\alpha (\partial_\tau^q \partial_\rho^p \bar{\phi}_k ) \bmZ^\alpha \bmX_- (\partial_\tau^q \partial_\rho^p \phi_{k+1} ) \\
    &\phantom{=} + 2(k+1-s+q-p) | D \phi_{k+1} |^2 + 2(k-s-q+p) |D\phi_k|^2.
\end{split}
\end{align}
It is worth noting at this stage that the numerical factors in the last two terms in \eqref{differentialidentity}, arising from the explicit appearance of the coordinates $\tau$ and $\rho$ in equations \eqref{AEqnAppendix} and \eqref{BEqnAppendix}, are a key aspect of the calculation. In particular, their signs may be manipulated by choosing the values of $q$ and $p$ appropriately. We now integrate \cref{differentialidentity} over $\mathcal{N}_t$ (cf. \Cref{Fig:Cylinder_Forward_Estimates}) with respect to the \emph{Euclidean} volume element $\mathrm{dv} = \d \tau \wedge \d \rho \wedge \d \mu$, where $\d \mu$ is the Haar measure on $\mathrm{SU}(2)$. We have
\begin{align*}
    0 & = \int_{\mathcal{N}_t} \left( \begin{array}{c} \partial_\tau \\ \partial_\rho \end{array} \right) \cdot \left( \begin{array}{c} (1+\tau) |D \phi_{k+1}|^2 + (1-\tau) |D \phi_k|^2 \\ -\rho | D \phi_{k+1}|^2 + \rho | D \phi_k |^2 \end{array} \right) \mathrm{dv} \\
    &\phantom{=} - \int_{\mathcal{N}_t} \Big( \bmZ^\alpha \bmX_+ ( \partial_\tau^q \partial_\rho^p \phi_k ) \bmZ^\alpha ( \partial_\tau^q \partial_\rho^p \bar{\phi}_{k+1} ) + \bmZ^\alpha (\partial_\tau^q \partial_\rho^p \phi_k ) \bmZ^\alpha \bmX_+ (\partial_\tau^q \partial_\rho^p \bar{\phi}_{k+1} ) \\
    &\phantom{=} \hspace{2em} + \bmZ^\alpha \bmX_- ( \partial_\tau^q \partial_\rho^p \bar{\phi}_k ) \bmZ^\alpha (\partial_\tau^q \partial_\rho^p \phi_{k+1} ) + \bmZ^\alpha (\partial_\tau^q \partial_\rho^p \bar{\phi}_k ) \bmZ^\alpha \bmX_- (\partial_\tau^q \partial_\rho^p \phi_{k+1} ) \Big) \, \mathrm{dv} \\
    & \phantom{=} + 2(k+1-s+q-p) \int_{\mathcal{N}_t} |D \phi_{k+1} |^2 \, \mathrm{dv} + 2(k-s-q + p) \int_{\mathcal{N}_t} |D \phi_k |^2 \, \mathrm{dv}.
\end{align*}
Using \Cref{Lemma:IntegrationSU2C}, it is straightforward to check that the second integral in the above equality will vanish when summed over $|\alpha| \leq m'$ for a given $m'$; as we will employ such a summation shortly, we simply write $\text{angular}(\alpha)$ to denote the second integral in the above. Next, using the Euclidean divergence theorem, one obtains
\begin{align*}
0 &= \int_{\mathcal{S}_t} \left( (1+t) |D \phi_{k+1}|^2 + (1-t)|D \phi_k|^2 \right) \d \rho \wedge \d \mu \\
&\phantom{=} - \int_{\mathcal{S}_\star} \left( (1+ \tau_\star) |D \phi_{k+1}|^2 + (1-\tau_\star) |D \phi_k |^2 \right) \d \rho \wedge \d \mu \\
&\phantom{=} + \int_{\mathcal{I}_t} \rho \left( |D\phi_k|^2 - |D \phi_{k+1} |^2 \right) \d \tau \wedge \d \mu \\
&\phantom{=} + \int_{\mathcal{B}_t} \left( \begin{array}{c} (1+\tau) |D \phi_{k+1}|^2 + (1-\tau) |D \phi_k|^2 \\ -\rho | D \phi_{k+1}|^2 + \rho | D \phi_k |^2 \end{array} \right) \cdot \nu \left( \begin{array}{c} \rho \\ 1+ \tau \end{array} \right) \d \mathcal{B} + \text{angular}(\alpha) \\
&\phantom{=} + 2 (p-q+k-s) \int_{\mathcal{N}_t} | D \phi_k |^2 \, \mathrm{dv} - 2 (p-q-k-1+s) \int_{\mathcal{N}_t} | D \phi_{k+1} |^2 \, \mathrm{dv},
\end{align*}
where $\d \mathcal{B}$ is the induced measure on $\mathcal{B}_t$, $\nu \equiv (\rho^2 + (1+\tau)^2 )^{-1/2}$ is a normalisation factor for the outward normal to $\mathcal{B}_t$, and we note that the integral over $\mathcal{I}_t$ vanishes due to the factor of $\rho$ in the integrand. Further, the integral over $\mathcal{B}_t$ simplifies to just $\int_{\mathcal{B}_t} 2 \nu \rho | D \phi_k |^2 \d \mathcal{B}$, and in particular is manifestly non-negative. Altogether therefore
\begin{align*}
&\int_{\mathcal{S}_t} \left( (1+t) |D \phi_{k+1}|^2 + (1-t)|D \phi_k|^2 \right) \d \rho \wedge \d \mu + 2 (p-q+k-s) \int_{\mathcal{N}_t} | D \phi_k |^2 \, \mathrm{dv}\\
& \leq \int_{\mathcal{S}_\star} \left( (1+ \tau_\star) |D \phi_{k+1}|^2 + (1-\tau_\star) |D \phi_k |^2 \right) \d \rho \wedge \d \mu + 2 (p-q-k-1+s) \int_{\mathcal{N}_t} | D \phi_{k+1} |^2 \, \mathrm{dv} \\
& - \text{angular}(\alpha).
\end{align*}

We now perform the relabelling
\[ q \longrightarrow q', \qquad p \longrightarrow p' + p, \]
under which $D \longrightarrow D' \partial_\rho^p$, where $D' = \partial_\tau^{q'} \partial_\rho^{p'} \bmZ^\alpha$, and sum both sides over 
\[ \mho \equiv \left\{ (q', p', \alpha) \, | \, q' + p' + |\alpha| \leq m \right\}. \]
Then $\sum_\mho \text{angular}(\alpha) = 0$, and we obtain
\begin{align*} (1+t) & \sum_\mho \int_{\mathcal{S}_t} |D'(\partial_\rho^p \phi_{k+1}) |^2 \, \d \rho \wedge \d \mu + (1-t) \sum_\mho \int_{\mathcal{S}_t} |D'(\partial_\rho^p \phi_k)|^2 \, \d \rho \wedge \d \mu \\
& + 2\sum_\mho (p'+p-q'+k-s) \int_{\mathcal{N}_t} |D'(\partial_\rho^p \phi_k)|^2 \, \mathrm{dv} \\
& \leq  (1+\tau_\star) \sum_\mho \int_{\mathcal{S}_\star} |D'(\partial_\rho^p \phi_{k+1}|^2 \, \d \rho \wedge \d \mu + (1-\tau_\star) \sum_\mho \int_{\mathcal{S}_\star} |D'(\partial_\rho^p \phi_k)|^2 \, \d \rho \wedge \d \mu \\
& +2 \sum_\mho (p' + p - q' - k -1 +s ) \int_{\mathcal{N}_t} |D'(\partial_\rho^p \phi_{k+1})|^2 \, \mathrm{dv}.
\end{align*}
We now choose $p$ and $m$ such that
\[ p > m+s  \]
and use simple uniform bounds for the numerical factors in the above,
\begin{align}
\label{Spin2IntermmediateInequality1}
\begin{split}
(1+t) & \sum_\mho \int_{\mathcal{S}_t} |D'(\partial_\rho^p \phi_{k+1}) |^2 \, \d \rho \wedge \d \mu + (1-t) \sum_\mho \int_{\mathcal{S}_t} |D'(\partial_\rho^p \phi_k)|^2 \, \d \rho \wedge \d \mu \\
& + 2(p - m -s + k) \sum_\mho \int_{\mathcal{N}_t} |D'(\partial_\rho^p \phi_k)|^2 \, \mathrm{dv} \\
& \leq  (1+\tau_\star) \sum_\mho \int_{\mathcal{S}_\star} |D'(\partial_\rho^p \phi_{k+1})|^2 \, \d \rho \wedge \d \mu + (1-\tau_\star) \sum_\mho \int_{\mathcal{S}_\star} |D'(\partial_\rho^p \phi_k)|^2 \, \d \rho \wedge \d \mu \\
& +2 (p+m+s-1) \sum_\mho  \int_{\mathcal{N}_t} |D'(\partial_\rho^p \phi_{k+1})|^2 \, \mathrm{dv}.
\end{split}
\end{align}
With the restriction $p>m+s$ on $p$ and $m$, the bulk integral on the left-hand side is positive, so we deduce for $0 \leq k \leq 2s -1$ that
\begin{align}
    \label{Spin2IntermmediateInequality2}
    \begin{split}
    \int_{\mathcal{S}_t} \sum_\mho |D'(\partial_\rho^p \phi_{k+1})|^2 \, \d \rho \wedge \d \mu & \la \int_{\mathcal{S}_\star} \sum_\mho \left( |D'(\partial_\rho^p \phi_{k+1})|^2 + |D'(\partial_\rho^p \phi_{k})|^2 \right) \d \rho \wedge \d \mu \\
    & \phantom{\la} + 2(p+m+s-1) \int_{\tau_\star}^t \left( \int_{\mathcal{S}_\tau} \sum_\mho |D'(\partial_\rho^p \phi_{k+1})|^2 \, \d \rho \wedge \d \mu \right) \d \tau.
    \end{split}
\end{align}
The inequality \eqref{Spin2IntermmediateInequality2} provides an estimate for all components of $\phi_{A_1\cdots A_{2s}}$ except $\phi_0$, the incoming radiation field. For this, we note that \eqref{Spin2IntermmediateInequality1} also implies
\begin{align}
\label{Spin2IntermmediateInequality3}
\begin{split}
2(p-m-s) \int_{\mathcal{N}_t} \sum_\mho | D'(\partial_\rho^p \phi_0) |^2 \, \mathrm{dv} & \la \int_{\mathcal{S}_\star} \sum_\mho \left( |D'(\partial_\rho^p \phi_0 ) |^2 + | D'(\partial_\rho^p \phi_1 ) |^2 \right) \d \rho \wedge \d \mu \\
& \phantom{\la} + \int_{\mathcal{N}_t} \sum_\mho | D' (\partial_\rho^p \phi_1 ) |^2 \, \mathrm{dv},
\end{split}
\end{align}
so that the norm $\| \partial_\rho^p \phi_0 \|^2_{H^m(\mathcal{N}_t)}$ will be controlled if the norm $\| \partial_\rho^p \phi_1 \|^2_{H^m(\mathcal{N}_t)}$ is controlled. Now applying Gr\"onwall's inequality to \eqref{Spin2IntermmediateInequality2}, we obtain
\[ \int_{\mathcal{S}_t} \sum_\mho | D'(\partial_\rho^p \phi_{k+1} )|^2 \, \d \rho \wedge \d \mu \la_{p,m,s} \int_{\mathcal{S}_\star} \sum_\mho \left( |D'(\partial_\rho^p \phi_{k+1})|^2 + |D'(\partial_\rho^p \phi_{k})|^2 \right) \d \rho \wedge \d \mu, \]
which, upon integration, gives
\[ \int_{\mathcal{N}_t} \sum_\mho |D'(\partial_\rho^p \phi_{k+1} )|^2 \, \mathrm{dv} \la_{p,m,s} \int_{\mathcal{S}_\star} \sum_\mho \left( |D'(\partial_\rho^p \phi_{k+1})|^2 + |D'(\partial_\rho^p \phi_k )|^2 \right) \d \rho \wedge \d \mu \]
for $0 \leq k \leq 2s -1$. This, combined with \eqref{Spin2IntermmediateInequality3}, then gives the estimate for $\phi_0$,
\[ \int_{\mathcal{N}_t} \sum_\mho |D'(\partial_\rho^p \phi_{0} )|^2 \, \mathrm{dv} \la_{p,m,s} \int_{\mathcal{S}_\star} \sum_\mho \left( |D'(\partial_\rho^p \phi_{1})|^2 + |D'(\partial_\rho^p \phi_0 )|^2 \right) \d \rho \wedge \d \mu. \]
We therefore conclude that for all $0 \leq k \leq 2s$ and for $p>m+s$
\begin{equation}
\tag{\ref{Spin2BasicEstimate}}
\| \partial_\rho^p \phi_k \|^2_{H^m(\mathcal{N}_t)} \leq C_{p,m,s} \sum_{k=0}^{2s} \int_{\mathcal{S}_\star} \sum_\mho |D'(\partial_\rho^p \phi_k)|^2 \, \d \rho \wedge \d \mu,
\end{equation}
where in particular the constant $C_{p,m,s}$ is independent of $t$.
\end{proof}

\begin{remark} Observe that the right-hand side of \eqref{Spin2BasicEstimate} in principle contains large numbers of derivatives in $\tau$. Using the spin-$s$ equations \eqref{AEqnAppendix} and \eqref{BEqnAppendix}, however, these can be eliminated entirely at the expense of producing derivatives in $\rho$ and along $\mathrm{SU}(2)$. That is, the right-hand side may be estimated by the initial data and its tangential derivatives on $\mathcal{S}_\star$.
\end{remark}

\begin{remark} As the right-hand side of \eqref{Spin2BasicEstimate} does not depend on $t$, one has that the left-hand side is uniformly bounded as $t \to 1$. Therefore for any $t \in(-1, 1]$, as the boundary $\partial \mathcal{N}_t$ of the 5-dimensional domain $\mathcal{N}_t$ is Lipschitz, the Sobolev embedding theorem gives for $m, r \in \mathbb{N}$, $\alpha \in [0,1)$ such that $m \geq r + \alpha + \frac{5}{2}$ the continuous embedding
\[ H^m(\mathcal{N}_t) \hookrightarrow C^{r, \alpha}(\mathcal{N}_t). \]
In particular, if $\alpha \leq \frac{1}{2}$ we have
\[ H^{r+ 3}(\mathcal{N}_t) \hookrightarrow C^{r, \alpha}(\mathcal{N}_t). \]
Note that it is important here that $\mathcal{N}_t$ is an open set, i.e. we consider $\mathcal{N}_t$ without its boundary.
\end{remark}

\subsection{Asymptotic expansions near $\mathcal{I}$}
\label{Section:ExpansionsUpperDomain}

Using estimate \eqref{Spin2BasicEstimate} and the Sobolev embedding theorem as described above, one finds that for any $0 \leq k \leq 2s$, $r \in \mathbb{N}$, and $0 \leq \alpha \leq \frac{1}{2}$,
\[
\partial^p_\rho \phi_k \in C^{r,\alpha} (\mathcal{N}_1) \quad
\text{for} \quad p\geq r + s + 4,
\]
provided the data on $\mathcal{S}_\star$ is sufficiently regular, as per \Cref{thm:forward_estimates}. This allows one to ensure the desired regularity of the solution near $\scri^+$ by providing sufficiently regular data on $\mathcal{S}_\star$. Now, using the wave equations of \Cref{Appendix:SolutionJets} on $\mathcal{I}$, one may solve for the functions
\[ \phi_k^{(p')} (\tau, t^\bmA{}_\bmB) \equiv \partial_\rho^{p'} \phi_k |_{\mathcal{I}}, \]
$0 \leq k \leq 2s$, on $\mathcal{I}$ for any finite order $p' \in \mathbb{N} \cup \{ 0 \}$.  The functions $\phi_k^{(p')}$ are computable explicitly as a consequence of the fact that the cylinder at spatial infinity $\mathcal{I}$ is a total characteristic of the spin-$s$ equations. Regarding these functions as functions on $\mathcal{N}_1$ which are independent of $\rho$, one may integrate $\partial^p_\rho \phi_k$ with respect to $\rho$ to obtain for $p \geq r +s +4$ the expansion
\begin{equation}
\phi_k = \sum_{p'=|k-2s|}^{p-1} \frac{1}{p'!} \phi^{(p')}_k \rho^{p'} + J^p(\partial_\rho^p \phi_k ) 
\label{FExpansion}
\end{equation}
on $\mathcal{N}_1$, where $J$ denotes the operator
\[ J \, : \, f \mapsto J(f) = \int_0^\rho f(\tau, \rho', t^\bmA{}_\bmB) \d \rho'. \]
It follows then that
\begin{equation}
\label{Taylor_expansion_in_rho_in_bulk}
J^p(\partial^p_\rho \phi_k) = \phi_k - \sum_{p'=0}^{p-1} \frac{1}{p'!} \phi^{(p')}_k \rho^{p'} \in C^{r,\alpha}(\mathcal{N}_1) \quad \text{for} \quad p \geq r + s + 4
\end{equation}
for any given $r \in \mathbb{N}$. One therefore obtains an expansion of the solution to the spin-$s$ equations in terms of the explicitly known functions $\phi_k^{(p')}$, $0 \leq k \leq 2s$, $0 \leq p' \leq p-1$, and a remainder of prescribed smoothness. We call such expansions $\emph{F-expansions}$. The details of the explicit computation of the terms in the F-expansions is given in \Cref{Appendix:SolutionJets}.

\begin{remark} Notice that the functions $\phi_k^{(p')}$ completely control whether the solution $\phi$ extends smoothly to $\scri^+$ or acquires logarithmic singularities there. Since the gauge in our construction is defined in terms of smooth conformally invariant structures (conformal geodesics and associated conformal Gaussian coordinates, see e.g. \cite{MagVal21}) and extends smoothly through $\scri^\pm$, these singularities are not a result of the choice of gauge. Rather, they are generated at all orders by the interplay between the structure of the data on $\mathcal{S}_\star$ near $\mathcal{I}$ and the nature of the evolution along $\mathcal{I}$.\end{remark}

\begin{remark} Inspecting the proof of \Cref{thm:forward_estimates}, one notices that the bound in \eqref{Taylor_expansion_in_rho_in_bulk} can in fact be improved to
\[ p \geq r + s + 4 - k, \]
provided one is prepared to consider the regularity of the solution component-by-component. In particular, setting $p=0=r$ shows that the components $\phi_{k\geq s +4}$ are always continuous. This set of components is non-empty for $s \geq 4$.
\end{remark}

\subsection{Existence of solutions}
In this subsection we provide a brief argument showing the existence of solutions to the spin-$s$ equations on the domain $\mathcal{N}_1$. The proof is a classical \emph{last slice argument} which combines the local existence result for symmetric hyperbolic systems with a contradiction argument.

As we are looking for solutions of the form \eqref{FExpansion} and the formal expansions are explicitly known, we only need to argue the existence of the derivatives $\partial_\rho^p\phi_k$. Given that the components $\phi_k$ satisfy a symmetric hyperbolic system (cf. \cref{StandardEvolutionFirst,StandardEvolutionLast}), it follows then that $\partial_\rho^p\phi_k$ also satisfy a symmetric hyperbolic system with the same structural properties as \cref{StandardEvolutionFirst,StandardEvolutionLast}.  Now, assume that on $\mathcal{S}_\star$ one has
\[
(\partial_\rho^p\phi_k)_\star \in H^{m}(\mathcal{S}_\star), \qquad p>m+s.
\]
The existence theorem for symmetric hyperbolic systems (see e.g. \cite{Kat75b})  then shows that there exists a $T>0$ such that
\[
\partial_\rho^p \phi_k (\tau, \cdot) \in H^{m}(\mathcal{S}_\tau), \qquad 0\leq \tau <T,
\]
and either $\partial_\rho^p \phi_k (\tau, \cdot) \not\in H^m(S_\tau)$ for $\tau >T$, or $T$ is not maximal, the dependence on $\tau$ being continuous. Assume now the existence of a time $\tau_\ast < 1$ (i.e. a \emph{last slice}) such that
\[
\partial_\rho^p \phi_k (\tau, \cdot) \not\in H^{m}(\mathcal{S}_{\tau}) \quad \mbox{for} \quad \tau\in[\tau_\star,1].
\]
It follows then that
\[
\int_{\tau_\ast}^1 \| \partial_\rho^p \phi_k \|^2_{H^m(\mathcal{S}_t)} \, \mathrm{d}t = \infty.
\]
As this integral is bounded by $\| \partial_\rho^p \phi_k\|_{H^m(\mathcal{N}_1)}$, the existence of a last slice contradicts the fact that $\partial_\rho^p\phi_k\in H^{m}(\mathcal{N}_1)$, which is true by \Cref{thm:forward_estimates}.

\section{Asymptotic characteristic initial value problem}
\label{Section:CharacteristicProblem}

The purpose of this section is to provide a formulation of the \emph{asymptotic characteristic initial value problem }for the massless spin-$s$ equations \eqref{MasslessSpinS}. By this we understand a setting in which suitable initial data is prescribed on a portion of past null infinity $\mathscr{I}^-$ and an incoming null hypersurface $\underline{\mathcal{B}}_\varepsilon$ intersecting $\mathscr{I}^-$ at a cut 
\begin{eqnarray*}
&& \mathcal{C}_\star \equiv  \mathscr{I}^-\cap \underline{\mathcal{B}}_\varepsilon\\
&& \phantom{\mathcal{C}_\star}= \{ (\tau, \rho) \times \mathbb{S}^2 \in \mathcal{N} \; | \; \tau =-1,\;\; \rho=\rho_\star \},
\end{eqnarray*}
with $\rho_\star$ a constant. 

\subsection{Freely specifiable data on $\scri^-$}
Equations \eqref{AEqnMain} and \eqref{BEqnMain} are, respectively, transport equations along outgoing and incoming null geodesics in the Minkowski spacetime, written in the F-gauge. In particular, at $\mathscr{I}^-$ one has that 
\begin{equation}
(A_k)|_{\scri^-} = -\rho \partial_\rho (\phi_{k+1})|_{\mathscr{I}^-} - \bmX_+(\phi_k)|_{\mathscr{I}^-} + (k+1-s) (\phi_{k+1})|_{\mathscr{I}^-}=0, \qquad k=0, \, \ldots, \, 2s-1,
\label{TransportEqnNullInfty}
\end{equation}
where $(\phi_k)|_{\mathscr{I}^-}$ denotes the restriction of $\phi_k$ to $\mathscr{I}^-$; \eqref{TransportEqnNullInfty} is a set of $2s$ transport equations along the null generators of $\mathscr{I}^-$. The key observation is that only $(\phi_0)|_{\mathscr{I}^-}$ is not fixed by these transport equations. This allows one to identify $\phi_0$ as the freely specifiable data on $\mathscr{I}^-$. That is, given $(\phi_0)|_{\mathscr{I}^-}$, one can solve the ODEs \eqref{TransportEqnNullInfty} one by one, at least in a neighbourhood of $\mathcal{C}_\star$, starting with the equation $A_0|_{\mathscr{I}^-}=0$. One obtains $(\phi_{k+1})|_{\mathscr{I}^-}$ for $k=0, \, \ldots, \, 2s-1$, provided that the initial values at the cut $(\phi_{k+1})_\star\equiv (\phi_{k+1})|_{\mathcal{C}_\star}$ are also known. 

Next, let $\underline{\mathcal{B}}_\varepsilon$ be the incoming null hypersurface whose generators are tangent to the vector $\bme_{\bmzero\bmzero'}$ and which intersects $\mathscr{I}^-$ at $\mathcal{C}_\star$. On $\underline{\mathcal{B}}_\varepsilon$ one has that
\begin{equation}
B_k|_{\underline{\mathcal{B}}_\varepsilon}=\bme_{\bmzero\bmzero'} (\phi_k)|_{\underline{\mathcal{B}}_\varepsilon} - \bmX_-(\phi_{k+1})|_{\underline{\mathcal{B}}_\varepsilon} +(k-s) (\phi_k)|_{\underline{\mathcal{B}}_\varepsilon}=0
    \label{TransportEqnBepsilon}
\end{equation}
for $k = 0, \, \ldots, \, 2s -1$. These are transport equations along the null generators of $\underline{\mathcal{B}}_\varepsilon$, in which only the component $(\phi_{2s})|_{\underline{\mathcal{B}}_\varepsilon}$ satisfies no differential condition. Accordingly, it can be identified with the freely specifiable data on $\underline{\mathcal{B}}_\varepsilon$. Similarly, one can then solve the ODEs \eqref{TransportEqnBepsilon} in sequence (starting from $B_{2s-1}$) to obtain $(\phi_k)|_{\underline{\mathcal{B}}_\varepsilon}$ for all $0\leq k \leq 2s$, provided the initial values $(\phi_k)|_\star \equiv (\phi_k)|_{\mathcal{C}_\star}$, $0 \leq k \leq 2s-1$, at the cut $\mathcal{C}_\star$ are known. 

We therefore have the following:

\begin{lemma}
\label{Lemma:FreData}
 The full set of initial data for the asymptotic characteristic initial value problem for the massless spin-$s$ equations can be computed from the reduced data set consisting of
 \begin{eqnarray*}
 && \phi_0 \quad \mbox{on}\quad \mathscr{I}^-, \\
 && \phi_{2s} \quad \mbox{on}\quad  \underline{\mathcal{B}}_\varepsilon, \quad \text{and} \\
 && \phi_1, \, \ldots, \,\phi_{2s-1}\quad \text{on} \quad \mathcal{C}_\star = \scri^- \cap \underline{\mathcal{B}}_\varepsilon. 
 \end{eqnarray*}
 \end{lemma}

\subsection{Symmetric hyperbolicity and the local existence theorem}
For reasons that will become clear shortly, in order to discuss the existence of solutions to the asymptotic characteristic initial value problem in a neighbourhood of $\mathcal{C}_\star$, it is convenient to use the evolution equations in a slightly different conformal gauge than the one discussed in \Cref{Section:NullFrameNearI0}. More precisely, we make use of the evolution equations \Cref{StandardEvolutionFirstGeneral,StandardEvolutionLastGeneral} as given in \Cref{Appendix:MoreGeneralGauge} and assume that the smooth function $\mu$ is no longer identically equal to $1$. A long but straightforward calculation shows that in this gauge the spin-$s$ equations \eqref{MasslessSpinS} are equivalent to an evolution system of the form
\begin{equation}
\mathbf{A}^\mu \partial_\mu \bmPhi =\mathbf{B}\cdot \bmPhi
    \label{MasslessSpinSSHS}
\end{equation}
where
\begin{eqnarray*}
&& \mathbf{A}^\mu \partial_\mu = \left( \begin{array}{cccccc} \sqrt{2} \hat{\bme}_{\bmzero \bmzero'} & - \mu \bmX_- & 0 & \cdots & \cdots & 0 \\
-\mu \bmX_+ & 2 \partial_{\hat{\tau}} & - \mu \bmX_- & 0 & \cdots & 0 \\
0 & -\mu \bmX_+ & 2 \partial_{\hat{\tau}} & - \mu \bmX_- & \ddots & \vdots \\
\vdots & \ddots & \ddots & \ddots & \ddots & 0 \\
0 & \cdots &  0 &- \mu \bmX_+ & 2 \partial_{\hat{\tau}} & - \mu \bmX_- \\
0 & \cdots & \cdots & 0 & - \mu \bmX_+ & \sqrt{2} \hat{\bme}_{\bmone \bmone'}
\end{array} \right), \\
&& \bmPhi = \big(\phi_0, \, \ldots, \, \phi_{2s}  \big)^t,
\end{eqnarray*}
and $\mathbf{B}$ is a constant $2s\times 2s$ matrix. It can be readily verified that the above matrix is Hermitian. Moreover,
\[
\mathbf{A}^0 \equiv \mbox{diag}\big(1-\kappa'\hat{\tau}, \, 2, \, \ldots, \, 2, \, 1+\kappa'\hat{\tau} \big).
\]
This matrix is positive definite away from $\mathcal{I}^\pm$. As $\kappa'=\mu'\rho +\mu$ with $\mu(0)=1$, but not identically $1$ away from $\rho=0$, it follows that
\begin{equation}
\mathbf{A}^0|_{\mathcal{I}^-}= \mbox{diag}\big(2, \, 2, \, \ldots, \, 2, \, 0 \big).
\label{DegenerateMatrix}
\end{equation}
An analogous expression holds on $\mathcal{I}^+$. Thus, the system \eqref{MasslessSpinSSHS} is symmetric hyperbolic away from the critical sets $\mathcal{I}^\pm$, but degenerates at $\mathcal{I}^\pm$.

\begin{remark}
For the choice $\mu=1$ corresponding to the basic F-gauge discussed in \Cref{Section:CylinderMinkowski}, one has that
\[
\mathbf{A}^0|_{\mathscr{I}^-}= \mbox{diag}\big(2, \, 2, \, \ldots,\, 2, \, 0 \big).
\]
That is, in this representation the corresponding symmetric hyperbolic system degenerates on the \emph{whole} of $\mathscr{I}^-$, which creates problems when attempting to solve the system \eqref{MasslessSpinSSHS}. The use of the more general gauge of \Cref{Appendix:MoreGeneralGauge} shows that the degeneration of $\mathbf{A}^0$ on the whole of $\scri^-$ in the case $\mu =1$ is a \emph{coordinate} singularity. While this more singular representation ($\mu=1$) does not lend itself naturally to the discussion of local existence in a neighbourhood of $\mathcal{C}_\star=\mathscr{I}^-\cap \underline{\mathcal{B}}_\varepsilon$, it gives rise to a straightforward construction of asymptotic initial data. Moreover, as we do in \Cref{Section:ForwardEstimates,Section:BackwardEstimates}, it allows us to construct robust estimates.  
\end{remark}

Using the general theory set up in \cite{Ren90} (see also \cite{CFEBook}, \S12.5.3), one can therefore obtain the following local existence theorem.

\begin{proposition}
\label{Proposition:CVIPRendall}
Given a smooth choice of reduced initial data as in \Cref{Lemma:FreData}, there exists a neighbourhood $\mathcal{U}\subset D^+(\mathscr{I}^-\cup \underline{\mathcal{B}}_\varepsilon)$ of $\mathcal{C}_\star$ in which the asymptotic characteristic initial value problem for the massless spin-$s$ equation has a unique smooth solution $\phi_k$ for $k=0,\ldots,2s$, where $D^+(S)$ denotes the future domain of dependence of the set $S \subset \mathcal{N}$.
\end{proposition}

In fact, using a strategy similar to that in \cite{Luk12}, it is possible to ensure the existence of solutions along a narrow causal rectangle along $\mathscr{I}^-$ bounded away from $\mathcal{I}$ (see \Cref{Fig:CausalDiamonds}). However, the degeneracy in the matrix $\mathbf{A}^0$ at $\mathcal{I}^-$ as given in \eqref{DegenerateMatrix} precludes one from obtaining a domain of existence which includes the past critical point (see \Cref{Remark:range_causal_diamond}). In order to make this domain extension statement more precise, let us introduce some further notation. Given $(\tau_\bullet,\rho_\bullet)\in [-1,1]\times [0,\infty)$, let 
\[
\mathcal{S}_{\tau_\bullet,\rho_\bullet}\equiv \big\{ (\tau,\rho,t^\bmA{}_\bmB) \in \mathcal{N}\;|\; \tau=\tau_\bullet, \; \rho=\rho_\bullet, \;t^\bmA{}_\bmB\in \mathrm{SU}(2)   \big\}.
\]
For a given choice of $(\tau_\bullet,\rho_\bullet)$, the set $\mathcal{S}_{\tau_\bullet,\rho_\bullet}$ naturally projects to a 2-sphere via the Hopf map. Suppose that $(\tau_\bullet,\rho_\bullet)$ are chosen so that $\mathcal{S}_{\tau_\bullet,\rho_\bullet}$ lies in the chronological future of $\mathcal{C}_\star =\mathcal{S}_{-1,\rho_\star}$. The \emph{causal diamond} defined by the sets $\mathcal{S}_{\tau_\bullet,\rho_\bullet}$ and $\mathcal{C}_\star$ is given by
\[
\mathscr{D}_{\tau_\bullet,\rho_\bullet} = J^+(\mathcal{C}_\star)\cap J^-(\mathcal{S}_{\tau_\bullet,\rho_\bullet}).
\]
A schematic representation of an example of such a causal diamond $\mathscr{D}_{\tau_\bullet,\rho_\bullet}$ is provided in Figure \ref{Fig:CausalDiamonds}. The set $\mathcal{S}_{\tau_\bullet,\rho_\bullet}$ can be thought of as the set of points of intersection of future-directed null geodesics emanating from $\mathcal{S}_{-1,\rho_\circ}\subset \mathscr{I}^-$ with the future-directed null geodesics emanating from $\mathcal{S}_{\tau_\ast,\rho_\ast}\subset \underline{\mathcal{B}}_\varepsilon$, where $\rho_\circ$, $\tau_\ast$, and $\rho_\ast$ depend on $(\tau_\bullet,\rho_\bullet)$. Conversely, given points $(-1,\rho_\circ,t^\bmA{}_\bmB)\in \mathscr{I}^-$ and $(\tau_\ast,\rho_\ast,t^\bmA{}_\bmB)\in \underline{\mathcal{B}}_\varepsilon$, there exist coordinates $(\tau_\bullet,\rho_\bullet)$ depending on $(\rho_\circ,\tau_\ast,\rho_\ast)$ such that $\mathcal{S}_{\tau_\bullet,\rho_\bullet}$ is the set of intersections of the future-directed null geodesics emanating from $\mathcal{S}_{-1,\rho_\circ}$ and $\mathcal{S}_{\tau_\ast,\rho_\ast}$. In this notation one then has the following.

\begin{proposition}
\label{Proposition:CIVPLuk}
Given $\rho_\circ\in (0,\rho_\star)$, consider smooth initial data for the asymptotic characteristic initial value problem for the spin-$s$ equations on $\underline{\mathcal{B}}_\varepsilon\cup \mathscr{I}^-_{[\rho_\circ,\rho_\star]}$ where
\[
\mathscr{I}^-_{[\rho_\circ,\rho_\star]}\equiv \{ p\in \mathscr{I}^- \;|\; \rho(p) \in [\rho_\circ,\rho_\star] \}.
\]
Then there exists a unique smooth solution to the massless spin-$s$ field equations on the causal diamond $\mathscr{D}_{\tau_\bullet,\rho_\bullet}$. 
\end{proposition}

\begin{figure}[H]
\centering
     \begin{subfigure}[b]{0.5\textwidth}
     \centering
\includegraphics[width=\textwidth]{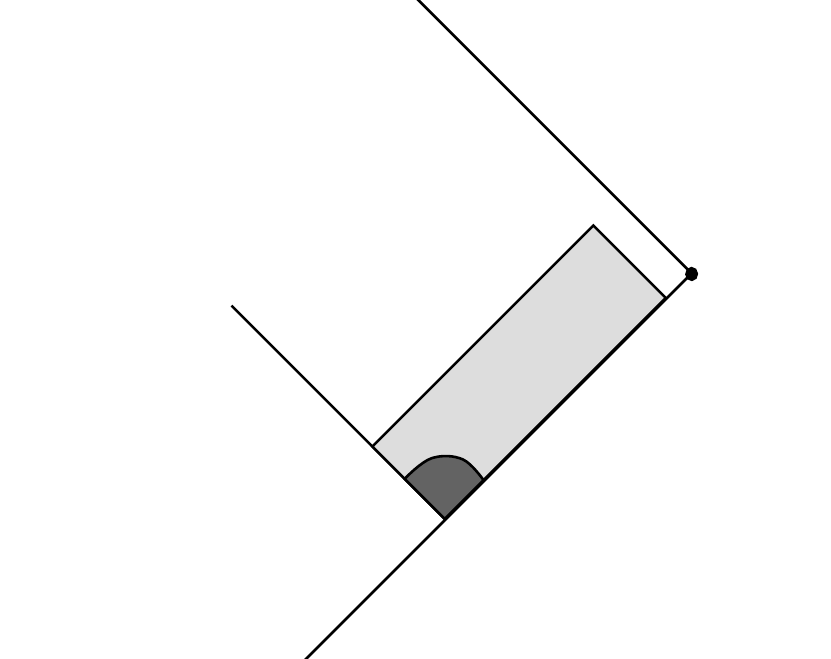}
\put(-75, 150){$\scri^+$}
\put(-75, 50){$\scri^-$}
\put(-95,75){$\mathscr{D}_{\tau_\bullet,\rho_\bullet}$}
\put(-106,43){\textcolor{white}{$\mathcal{U}$}}
\put(-30, 101){$i^0$}
\put(-65.5, 112){$\bullet$}
\put(-92, 117){$\mathcal{S}_{\tau_\bullet, \rho_\bullet}$}
\put(-122, 36){$\underline{\mathcal{B}}_\varepsilon$}
\caption{The domain of existence in the Penrose gauge.}
\end{subfigure} 
\begin{subfigure}[b]{0.45\textwidth}
\centering
\includegraphics[width=\textwidth]{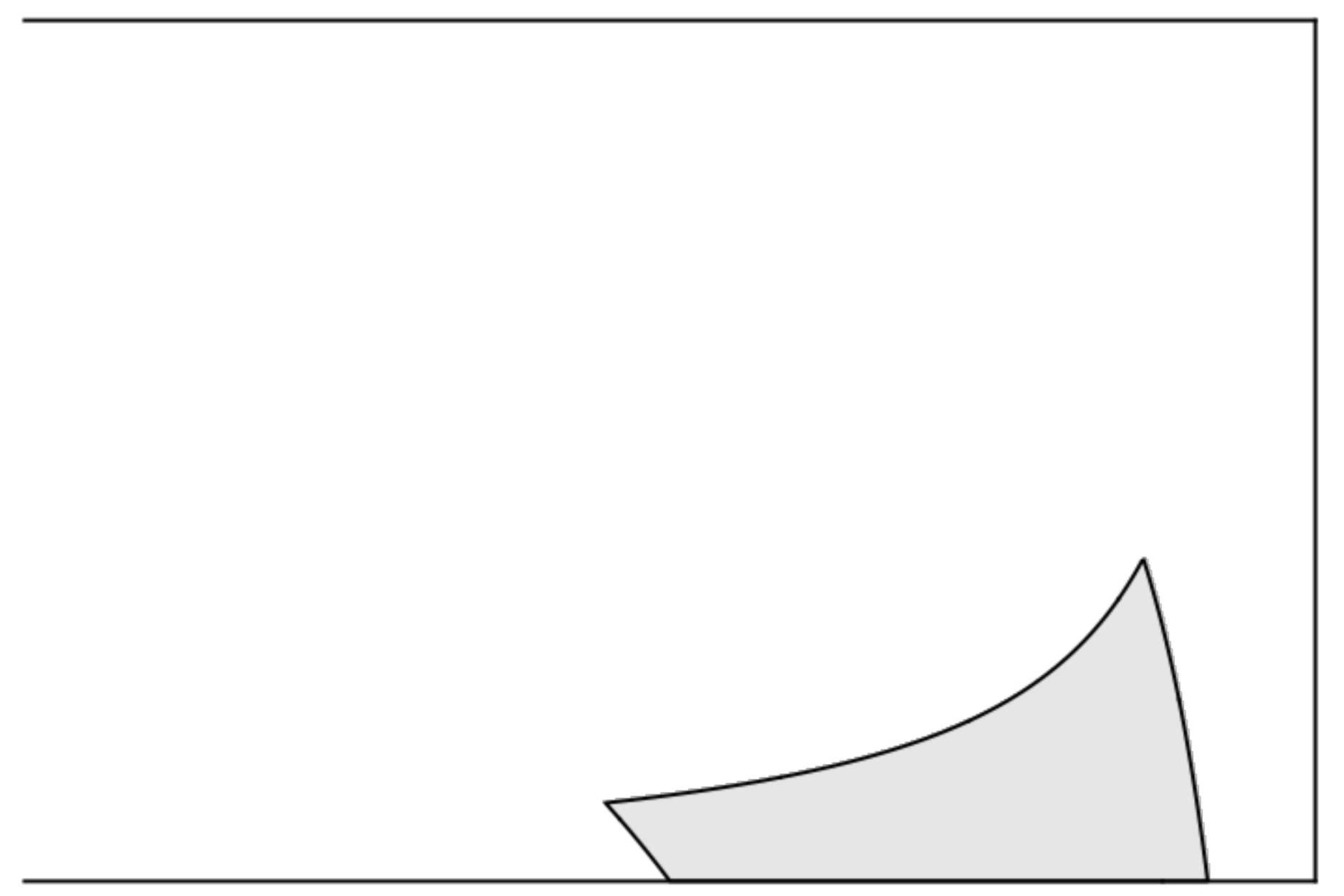}
\put(-150, 135){$\scri^+$}
\put(-150, -10){$\scri^-$}
\put(0,60){$\mathcal{I}$}
\put(0,-10){$\mathcal{I}^-$}
\put(-6.5,0){$\bullet$}
\put(0,135){$\mathcal{I}^+$}
\put(-6.5,128){$\bullet$}
\put(-70,7){$\mathscr{D}_{\tau_\bullet,\rho_\bullet}$}
\caption{The domain of existence in the F-gauge.}
\end{subfigure}
\caption{The domains of existence of Propositions \ref{Proposition:CVIPRendall} and \ref{Proposition:CIVPLuk}. Observe that these domains do not include spatial infinity.}\label{Fig:CausalDiamonds}
\end{figure}
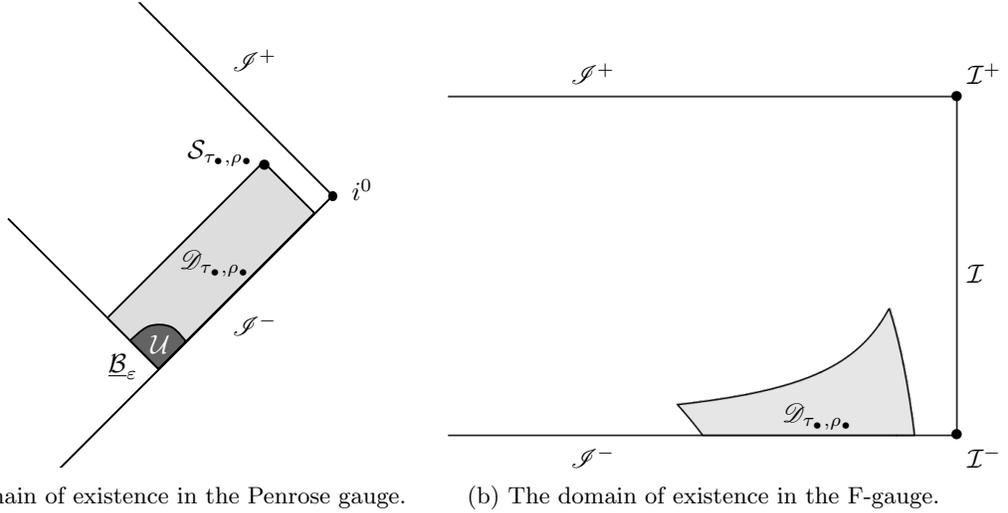

\begin{remark} \label{Remark:range_causal_diamond} The conclusion of \Cref{Proposition:CIVPLuk} ensures the existence of solutions to the spin-$s$ equations on a causal diamond which can get arbitrarily close to spatial infinity, but cannot actually reach it. The development of estimates to deal with this degeneracy is the objective of the next section.
\end{remark}

\section{Estimates near \texorpdfstring{$\scri^-$}{scriminus}}
\label{Section:BackwardEstimates}
In this section we construct estimates which allow us to control the solutions to the spin-$s$ equations in terms of asymptotic characteristic initial data on past null infinity. At the core of these estimates is a combination of the methods developed in \cite{Luk12} to obtain an optimal existence result for the characteristic initial value problem, and the strategy adopted in \Cref{Section:ForwardEstimates} for the standard initial value problem near spatial infinity. We recall once more that the degeneracy of the equations at the critical point $\mathcal{I}^-$ prevents a direct application of the results of \cite{Luk12}, necessitating a different treatment of the region which includes spatial infinity.

\subsection{Estimating the radiation field}
The component $\phi_{2s}$ (the \emph{outgoing radiation field}) is the most problematic to estimate from $\scri^-$ in our setting, as the na\"ive energy estimate for equation \eqref{AEqnAppendix} for $k=2s-1$ degenerates at past null infinity $\scri^- = \{ \tau = -1 \}$. However, we may produce estimate analogous to those of \Cref{Section:ForwardEstimates}.

Given $\rho_\star >0$, we consider the following hypersurfaces in  $\overline{\mathcal{N}}$:
\begin{align*}
\scri^-_{\rho_\star} &\equiv \scri^- \cap \{ 0 \leq \rho \leq \rho_\star \} , \\
\underline{\mathcal{B}}_\varepsilon &\equiv  \left\{ (\tau, \rho, t^\bmA{}_{\bmB}) \, | \, -1 \leq \tau \leq -1 + \varepsilon, ~ \rho = \frac{\rho_\star}{1-\tau}, ~ t^\bmA{}_{\bmB} \in \mathrm{SU}(2) \right\}, \\
\mathcal{S}_{-1+\varepsilon} &\equiv  \left\{ (\tau, \rho, t^\bmA{}_{\bmB}) \, | \, \tau = -1 + \varepsilon, ~ 0 \leq \rho \leq \frac{\rho_\star}{2-\epsilon}, ~ t^\bmA{}_\bmB \in \mathrm{SU}(2) \right\},
\\
\mathcal{I}_\varepsilon &\equiv \bigg\{ (\tau, \rho, t^\bmA{}_{\bmB}) \, | \, -1 \leq \tau \leq -1 + \varepsilon, ~ \rho = 0, ~ t^\bmA{}_\bmB \in \mathrm{SU}(2) \bigg\}.
\end{align*}
Observe that the set $\underline{\mathcal{B}}$ is a short incoming null hypersurface intersecting $\scri^-$ at $\rho = \rho_\star$. Moreover, let $\underline{\mathcal{N}}_\varepsilon$ denote the spacetime slab bounded by the hypersurfaces $\underline{\mathcal{B}}_\varepsilon$, $\mathscr{I}^-_{\rho_\star}$, $\mathcal{S}_{-1+\varepsilon}$ and $\mathcal{I}_\varepsilon$ (see \Cref{Fig:Cylinder_Backward_Estimates}). As in the previous section, given $m\geq 0$ we let
\[
\mho \equiv \{ (q', p', \alpha) \in \mathbb{N} \times \mathbb{N} \times \mathbb{N}^3 \, : \, q' + p' + |\alpha| \leq m \}.
\]

\begin{figure}[H]
\begin{center}
\includegraphics[scale=0.4]{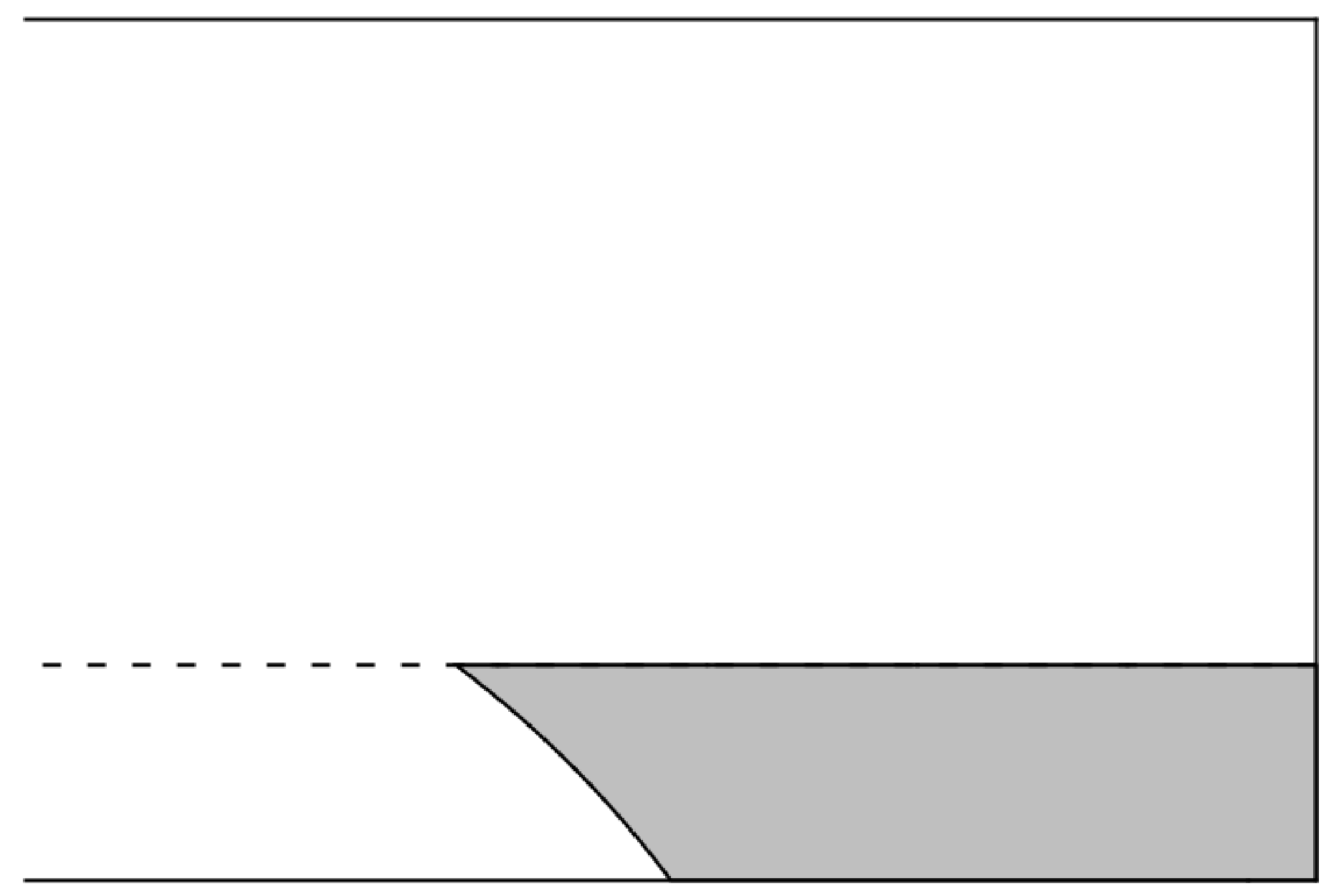}
\put(-120, 148){$\scri^+$}
\put(-65, -10){$\scri^-_{\rho_\star}$}
\put(2, 15){$\mathcal{I}_\varepsilon$}
\put(-65, 43){$\mathcal{S}_{-1+\varepsilon}$}
\put(-140, 15){$\underline{\mathcal{B}}_\varepsilon$}
\put(-65, 15){$\underline{\mathcal{N}}_\varepsilon$}
\end{center}
\caption{We perform energy estimates from $\scri^-_{\rho_\star} \cup \underline{\mathcal{B}}_\varepsilon$ to a Cauchy surface $\mathcal{S}_{-1+\varepsilon}$ for some $\varepsilon \ll 1$.}
\label{Fig:Cylinder_Backward_Estimates}
\end{figure}
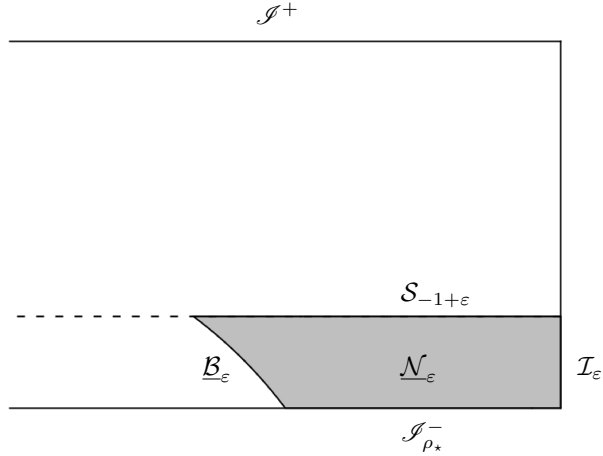

\noindent In terms of the above we have the following estimate for the outgoing radiation field.

\begin{proposition} \label{thm:radiation_field_estimate} Let $\varepsilon > 0$. Suppose that for $m, \, q, \, p \in \mathbb{N}$ and $(q', p', \alpha) \in \mathbb{N} \times \mathbb{N} \times \mathbb{N}^3$ satisfying
\[ q' + p' + |\alpha| \leq m \quad \text{and} \quad m + p < s + q  \]
we have the bound 
\begin{equation}
    \label{characteristic_data_assumption}
    \sum_\mho \int_{\scri^-_{\rho_\star}} |D'( \partial_\tau^q \partial^p_\rho \phi_k)|^2 \d \rho \wedge \d \mu + \sum_\mho \int_{\underline{\mathcal{B}}_\varepsilon} |D'(\partial_\tau^q \partial^p_\rho\phi_k)|^2 \d \underline{\mathcal{B}} \leq \Omega_\star
\end{equation}
on the characteristic data for $0 \leq k \leq 2s$,
for some $\Omega_\star > 0$. Additionally, assume that the bootstrap bound 
\begin{equation} 
\label{bootstrap_bound}
    \| \partial_\tau^q \partial_\rho^p \phi_k \|^2_{H^m(\underline{\mathcal{N}}_\varepsilon)} < \Omega_\star
\end{equation}
holds for $0 \leq k \leq 2s -1$. Then there exists a constant $C>0$ such that
\begin{equation}
    \label{radiation_field_bound_scri_minus}
    \| \partial_\tau^q \partial_\rho^p \phi_{2s} \|^2_{H^m(\underline{\mathcal{N}}_\varepsilon)} \leq C \Omega_\star.
\end{equation}
\end{proposition}

\begin{remark}
The bootstrap bound \eqref{bootstrap_bound} is proven in \Cref{prop:bootstrap_bound} assuming only the bound \eqref{characteristic_data_assumption}.
\end{remark}

\begin{proof} As before, we write
\[ D \equiv D^{q, p, \alpha} \equiv \partial_\tau^q \partial_\rho^p \bmZ^\alpha, \]
and commute $D$ into the equations \eqref{AEqnAppendix} and \eqref{BEqnAppendix} to arrive at the identity
\begin{align}
\tag{\ref{differentialidentity}}
\begin{split}
    0 &= \partial_\tau \left( (1+ \tau) | D \phi_{k+1} |^2 + (1-\tau) |D \phi_k |^2 \right) + \partial_\rho \left( - \rho | D \phi_{k+1} |^2 + \rho |D \phi_k |^2 \right) \\
    &\phantom{=} - \bmZ^\alpha \bmX_+ ( \partial_\tau^q \partial_\rho^p \phi_k ) \bmZ^\alpha ( \partial_\tau^q \partial_\rho^p \bar{\phi}_{k+1} ) - \bmZ^\alpha (\partial_\tau^q \partial_\rho^p \phi_k ) \bmZ^\alpha \bmX_+ (\partial_\tau^q \partial_\rho^p \bar{\phi}_{k+1} ) \\
    &\phantom{=} - \bmZ^\alpha \bmX_- ( \partial_\tau^q \partial_\rho^p \bar{\phi}_k ) \bmZ^\alpha (\partial_\tau^q \partial_\rho^p \phi_{k+1} ) - \bmZ^\alpha (\partial_\tau^q \partial_\rho^p \bar{\phi}_k ) \bmZ^\alpha \bmX_- (\partial_\tau^q \partial_\rho^p \phi_{k+1} ) \\
    &\phantom{=} + 2(k+1-s+q-p) | D \phi_{k+1} |^2 + 2(k-s-q+p) |D\phi_k|^2,
\end{split}
\end{align}
where $0 \leq k \leq 2s-1$. Integrating \eqref{differentialidentity} over the region $\underline{\mathcal{N}}_\varepsilon$  against the 5-form $\mathrm{dv} \equiv \d \tau \wedge \d \rho \wedge \d \mu$, we have
\begin{align*}
    0 & = \int_{\underline{\mathcal{N}}_\varepsilon} \left( \begin{array}{c} \partial_\tau \\ \partial_\rho \end{array} \right) \cdot \left( \begin{array}{c} (1+\tau) |D \phi_{k+1}|^2 + (1-\tau) |D \phi_k|^2 \\ -\rho | D \phi_{k+1}|^2 + \rho | D \phi_k |^2 \end{array} \right) \mathrm{dv} \\
    &\phantom{=} + 2 (k-s-q+p) \int_{\underline{\mathcal{N}}_\varepsilon} | D \phi_k |^2 \, \mathrm{dv} + 2(k+1-s+q-p) \int_{\underline{\mathcal{N}}_\varepsilon} | D \phi_{k+1}|^2 \, \mathrm{dv} \\
    &\phantom{=} - \int_{\underline{\mathcal{N}}_\varepsilon} \Big( \bmZ^\alpha \bmX_+ ( \partial_\tau^q \partial_\rho^p \phi_k ) \bmZ^\alpha ( \partial_\tau^q \partial_\rho^p \bar{\phi}_{k+1} ) + \bmZ^\alpha (\partial_\tau^q \partial_\rho^p \phi_k ) \bmZ^\alpha \bmX_+ (\partial_\tau^q \partial_\rho^p \bar{\phi}_{k+1} ) \\
    &\phantom{=} \underbrace{\hspace{2em} + \bmZ^\alpha \bmX_- ( \partial_\tau^q \partial_\rho^p \bar{\phi}_k ) \bmZ^\alpha (\partial_\tau^q \partial_\rho^p \phi_{k+1} ) + \bmZ^\alpha (\partial_\tau^q \partial_\rho^p \bar{\phi}_k ) \bmZ^\alpha \bmX_- (\partial_\tau^q \partial_\rho^p \phi_{k+1} ) \Big) \, \mathrm{dv}}_{\text{angular}(\alpha)}.
\end{align*}

\noindent Using \Cref{Lemma:IntegrationSU2C}, as in \Cref{Section:ForwardEstimates} we will have $\sum_{|\alpha| \leq m'} \text{angular}(\alpha) = 0$ for any $m' \in \mathbb{N}$. We will therefore not write the angular terms out in detail in the following computations. Using the Euclidean divergence theorem, we then find
\begin{align}
\label{scri_minus_pre_estimate_identity}
\begin{split}
    0 & =\int_{\mathcal{S}_{-1+\varepsilon}} \varepsilon | D \phi_{k+1} |^2 + (2-\varepsilon) | D \phi_k|^2 \, \d \rho \wedge \d \mu - \int_{\scri^-_{\rho_*}} 2 | D \phi_k |^2 \, \d \rho \wedge \d \mu  \\
    &\phantom{=} - \int_{\mathcal{I}_\varepsilon} \rho \left( - | D \phi_{k+1} |^2 + |D \phi_k |^2  \right) \d \tau \wedge \d \mu \\
    &\phantom{=} + \int_{\underline{\mathcal{B}}_\varepsilon} \left( \begin{array}{c} (1+\tau) |D \phi_{k+1}|^2 + (1-\tau) |D \phi_k|^2 \\ -\rho | D \phi_{k+1}|^2 + \rho | D \phi_k |^2 \end{array} \right) \cdot \nu \left( \begin{array}{c} -\rho \\ 1-\tau \end{array} \right) \d \underline{\mathcal{B}} + \text{angular}(\alpha) \\
    & \phantom{=} + 2(k-s-q+p) \int_{\underline{\mathcal{N}}_\varepsilon} | D \phi_k |^2 \, \mathrm{dv} + 2(k+1-s+q-p) \int_{\underline{\mathcal{N}}_\varepsilon} | D \phi_{k+1} |^2 \, \mathrm{dv},
\end{split}
\end{align}
where $\nu \equiv (\rho^2 + (1-\tau)^2)^{-1/2}$ is a normalisation factor for the outward normal to $\underline{\mathcal{B}}_\varepsilon$, $\d \underline{\mathcal{B}}$ is the induced measure on $\underline{\mathcal{B}}_\varepsilon$, and we note that the integral over $\mathcal{I}_\varepsilon$ vanishes due to the factor of $\rho$ in the integrand. In a slight deviation from the computations of \Cref{Section:ForwardEstimates}, here we perform the relabeling
\[ q \longrightarrow q + q', \qquad p \longrightarrow p + p', \]
so that $D \longrightarrow D' \partial_\tau^q \partial_\rho^p$, where $D' = \partial_\tau^{q'} \partial_\rho^{p'} \bmZ^\alpha$, and sum both sides of the above equality over $\mho$. The integrals over angular terms vanish, and noting that the integral of $(2-\varepsilon)|D'(\partial_\tau^q \partial_\rho^p \phi_k )|^2$ over $\mathcal{S}_{-1 + \varepsilon}$ is non-negative, we have the estimate
\begin{align*}
    \varepsilon \sum_\mho \int_{\mathcal{S}_{-1+\varepsilon}} | D'(\partial_\tau^q \partial_\rho^p \phi_{k+1})|^2 \, \d \rho \wedge \d \mu & \leq 2 \sum_\mho \int_{\scri^-_{\rho_*}} | D'(\partial_\tau^q \partial_\rho^p \phi_k) |^2 \, \d \rho \wedge \d \mu \\
    &\phantom{=}+ 2 \sum_\mho \int_{\underline{\mathcal{B}}_\varepsilon} \nu \rho | D'(\partial_\tau^q \partial_\rho^p \phi_k)|^2 \, \d \underline{\mathcal{B}} \\
    &\phantom{=} + 2 \sum_\mho (q + q' - p' - p + s - k) \int_{\underline{\mathcal{N}}_\varepsilon} | D'(\partial_\tau^q \partial^p_\rho \phi_k)|^2 \, \mathrm{dv} \\
    &\phantom{=} + 2 \sum_\mho (- q - q' + p' + p + s - k - 1 ) \int_{\underline{\mathcal{N}}_\varepsilon} | D'(\partial_\tau^q \partial^p_\rho \phi_{k+1})|^2 \, \mathrm{dv} \\
    & \leq 2 \sum_\mho \int_{\scri^-_{\rho_*}} | D'(\partial_\tau^q \partial_\rho^p \phi_k) |^2 \, \d \rho \wedge \d \mu \\
    &\phantom{=}+ 2 \sum_\mho \int_{\underline{\mathcal{B}}_\varepsilon} \nu \rho | D'(\partial_\tau^q \partial_\rho^p \phi_k)|^2 \, \d \underline{\mathcal{B}} \\
    &\phantom{\leq} + C^{(1)}_{m,p,s,k} \| \partial_\tau^q \partial_\rho^p \phi_k \|^2_{H^m(\underline{\mathcal{N}}_\varepsilon)} + C^{(2)}_{m,p,s,k} \| \partial_\tau^q \partial^p_\rho \phi_{k+1} \|^2_{H^m(\underline{\mathcal{N}}_\varepsilon)},
\end{align*}
where 
\[
C^{(2)}_{m,p,s,k} \equiv 2 (m-q+p+s-1-k)
\]
and 
\[
\| \psi \|^2_{H^m(\underline{\mathcal{N}}_\varepsilon)} \equiv  \sum_\mho \int_{\mathcal{N}_\varepsilon} |D' \psi |^2 \, \mathrm{dv}.
\]

Recall now that we assumed that the bootstrap bound \eqref{bootstrap_bound} holds---that is, we have that 
\[
    \| \partial_\tau^q \partial^p_\rho \phi_k \|^2_{H^m(\underline{\mathcal{N}}_\varepsilon)} \leq \Omega_\star
\]
for $0 \leq k \leq 2s -1$. Plugging in this bootstrap bound and assumption \eqref{characteristic_data_assumption} into the above estimate, we obtain
\[ 
\varepsilon \sum_\mho \int_{\mathcal{S}_{-1+\varepsilon}} | D'(\partial_\tau^q \partial^p_\rho \phi_{k+1} ) |^2 \, \d \rho \wedge \d \mu \leq C \left( \Omega_\star + C^{(2)}_{m, p, s, k} \| \partial_\tau^q \partial^p_\rho \phi_{k+1} \|^2_{H^m(\underline{\mathcal{N}}_\varepsilon)} \right) .
\]
The surfaces $\mathcal{S}_{-1+t}$, $t \in (0, \varepsilon)$, foliate $\underline{\mathcal{N}}_\varepsilon$, so this estimate is of the form
\[ 
\varepsilon f'_k(\varepsilon) \leq C \left( \Omega_\star + C^{(2)}_{m,p,s,k} f_k(\varepsilon) \right), 
\]
where $f_k(\varepsilon) \equiv \| \partial_\tau^q \partial^p_\rho \phi_{k+1} \|^2_{H^m(\underline{\mathcal{N}}_\varepsilon)}$. This may be rewritten as
\begin{equation} 
\label{Gronwall_radiation_field_scri_minus}
\frac{\d}{\d \varepsilon} \left( \varepsilon^{-C C^{(2)}_{m,p,s,k}} f_k(\varepsilon) \right) \leq C \Omega_\star \varepsilon^{-C C^{(2)}_{m,p,s,k} - 1},
\end{equation}
which is integrable near $\varepsilon = 0$ if $C^{(2)}_{m,p,s,k} < 0$. This is satisfied for $k=2s-1$ provided
\begin{equation}
\label{m_p_bound_scri_minus}
m + p < s + q .
\end{equation}
Assuming \eqref{m_p_bound_scri_minus} holds and integrating \eqref{Gronwall_radiation_field_scri_minus}, we arrive at
\[ 
f_{2s-1}(\varepsilon) \leq \frac{\Omega_\star}{-C^{(2)}_{m,p,s,2s-1}} \la \Omega_\star. 
\]
This is the required estimate.
\end{proof}

\subsection{The boostrap bound}

To prove the bootstrap bound \eqref{bootstrap_bound}, we return to the identity \eqref{scri_minus_pre_estimate_identity}. 

\begin{proposition} \label{prop:bootstrap_bound} Let $\varepsilon > 0$. Suppose that for $m, \, q, \, p \in \mathbb{N}$ and all $(q', p', \alpha) \in \mathbb{N} \times \mathbb{N} \times \mathbb{N}^3$ satisfying
\[ q' + p' + |\alpha| \leq m \quad \text{and} \quad m + p < s + q \]
we have the bound \eqref{characteristic_data_assumption} on characteristic data for $0 \leq k \leq 2s$. Then there exists a constant $C>0$  depending on $m$, $p$, $s$ and $\varepsilon$ such that
\[ 
\| \partial_\tau^q \partial_\rho^p \phi_k \|^2_{H^m(\underline{\mathcal{N}}_\varepsilon)} \leq C \Omega_\star 
\]
for $0 \leq k \leq 2s-1$, where $\underline{\mathcal{N}}_\varepsilon$ is as in \Cref{thm:radiation_field_estimate}. In other words, the bootstrap estimate \eqref{bootstrap_bound} holds.

\end{proposition}

\begin{proof} Performing the shifts $q \to q + q'$ and $p \to p + p'$ as in \Cref{thm:radiation_field_estimate} and summing the identity \eqref{scri_minus_pre_estimate_identity} over $\mho$, we deduce (this time dropping the manifestly positive $\varepsilon | D'( \partial_\tau^q \partial^p_\rho \phi_{k+1}) |^2$ term) that 
\begin{equation}
    \label{scri_minus_bootstrap_estimate_pre_gronwall}
    \sum_\mho \int_{\mathcal{S}_{-1+\varepsilon}} |D'(\partial_\tau^q \partial^p_\rho \phi_k) |^2 \, \d \rho \wedge \d \mu \la \Omega_* + C^{(1)}_{m,p,s,k} \| \partial_\tau^q \partial^p_\rho \phi_k \|^2_{H^m(\mathcal{N}_\varepsilon)} + C^{(2)}_{m,p,s,k} \| \partial_\tau^q \partial^p_\rho \phi_{k+1} \|^2_{H^m(\underline{\mathcal{N}}_\varepsilon)},
\end{equation}
where as before
\[ 
C^{(2)}_{m,p,s,k} = 2 (m - q + p + s - 1 - k) 
\]
and $C^{(1)}_{m,p,s,k} = 2 (m + q - p + s - k)$.

We begin with $k = 2s-1$. Then $C^{(2)}_{m,p,s,2s-1} < 0$ is guaranteed by the assumption that $m + p < s + q$, and the estimate \eqref{scri_minus_bootstrap_estimate_pre_gronwall} reads
\[ 
g_{2s-1}(\varepsilon) \la \Omega_\star + C^{(1)}_{m,p,s,2s-1} \int_0^\varepsilon g_{2s-1}(t) \d t , 
\]
where 
\[
g_k(t) \equiv \sum_\mho \int_{\mathcal{S}_{-1+t}} | D'(\partial_\tau^q \partial^p_\rho \phi_{k}) |^2 \, \d \rho \wedge \d \mu.
\]
Using Gr\"onwall's lemma, we find
\[ \| \partial_\tau^q \partial^p_\rho \phi_{2s-1} \|^2_{H^m(\underline{\mathcal{N}}_\varepsilon)} = \int_0^\varepsilon g_{2s-1}(t) \d t \la \Omega_* e^{C_{m,p,s} \varepsilon} \la_{m,p,s,\varepsilon} \Omega_\star. \]
We can now return to \eqref{scri_minus_bootstrap_estimate_pre_gronwall} for $k=2s-2$. We have, using the above estimate for $\phi_{2s-1}$, 
\[ g_{2s-2}(\varepsilon) \la \Omega_\star + C^{(1)}_{m,p,s,2s-2} \int_0^\varepsilon g_{2s-2}(t) \d t, \]
and so deduce that
\[ \| \partial_\tau^q \partial_\rho^p \phi_{2s-2} \|^2_{H^m(\underline{\mathcal{N}}_\varepsilon)} \la \Omega_\star. \]
Proceeding inductively for $0 \leq k \leq 2s-3$, we conclude that there exists a constant $C = C(m,p,s,\varepsilon)>0$ such that
\[ \| \partial_\tau^q \partial^p_\rho \phi_k \|^2_{H^m(\underline{\mathcal{N}}_\varepsilon)} \leq C \Omega_\star \]
for $0 \leq k \leq 2s-1$.
\end{proof}

\subsection{Asymptotic expansions near $\scri^-$}
\label{sec:expansions_near_past_null_infinity}

The estimates obtained in \Cref{Section:BackwardEstimates} may now be used to control the solutions to the spin-$s$ field equations in a narrow causal diamond extending along $\scri^-$ and containing a portion of $\mathcal{I}$. The strategy is to decompose the solution into a formal expansion around $\scri^-$, the regularity of the terms of which will be known explicitly, and a remainder whose regularity is controlled by the estimates on $\partial_\tau^q  \phi_k$. 

It is clear from the condition $m+p < s + q$ in \Cref{thm:radiation_field_estimate} that requiring $\rho$-derivatives near $\scri^-$ hinders the regularity sought on $\mathcal{S}_{-1+\varepsilon}$. Therefore setting $p=0$ (and writing $m \to m + 3$ for convenience) in \Cref{thm:radiation_field_estimate}, we have $\forall \, 0 \leq k \leq 2s$
\begin{equation}
\partial_\tau^q \phi_k\in H^{m+3}(\underline{\mathcal{N}}_\varepsilon) \quad \text{if} \quad m+3<s+q.
\label{EstimateLowerDomain}
\end{equation}
Sobolev embedding therefore gives
\[
\partial_\tau^q \phi_k \in C^{m,\alpha}(\underline{\mathcal{N}}_\varepsilon) \quad \text{for} \quad 0< \alpha\leq \frac{1}{2}, \; \quad m+3<s+q.
\]
Integrating this $q$ times with respect to $\tau$, one finds 
\begin{equation}
\phi_k = \sum_{q'=0}^{q-1} \frac{1}{q'!} (\partial_\tau^{q'}\phi_k)|_{\mathscr{I}^-}(\tau+1)^{q'}+ I^q(\partial_\tau^q\phi_k),
\label{ExpansionLowerDomain}
\end{equation}
where $I^q(\partial_\tau^q \phi_k) \in C^{m,\alpha}(\underline{\mathcal{N}}_\varepsilon)$ and $I$ denotes the operator
\[
f\mapsto I(f)= \int_{-1}^\tau f(\tau',\rho,t^\bmA{}_\bmB) \, \mathrm{d}\tau'.
\]

\begin{remark}
\label{Remark:ComputationExpansionsScriMinus}
Each term (except for the remainder) in the expansion \eqref{ExpansionLowerDomain} can be computed explicitly from the characteristic initial data $\phi_k|_{\scri^-}$ on $\scri^-$: see \Cref{Appendix:ExpansionsNullInfinity}. In particular, the regularity of these terms is known explicitly. 
\end{remark}


\subsection{Last slice argument}

As in the case of the upper domain $\mathcal{N}_1$, the existence of solutions in the lower domain $\underline{\mathcal{N}}_{\varepsilon}$ can be shown by a last slice argument. We provide a sketch of the argument below.

\begin{figure}[h]
\begin{center}
\includegraphics[scale=0.55]{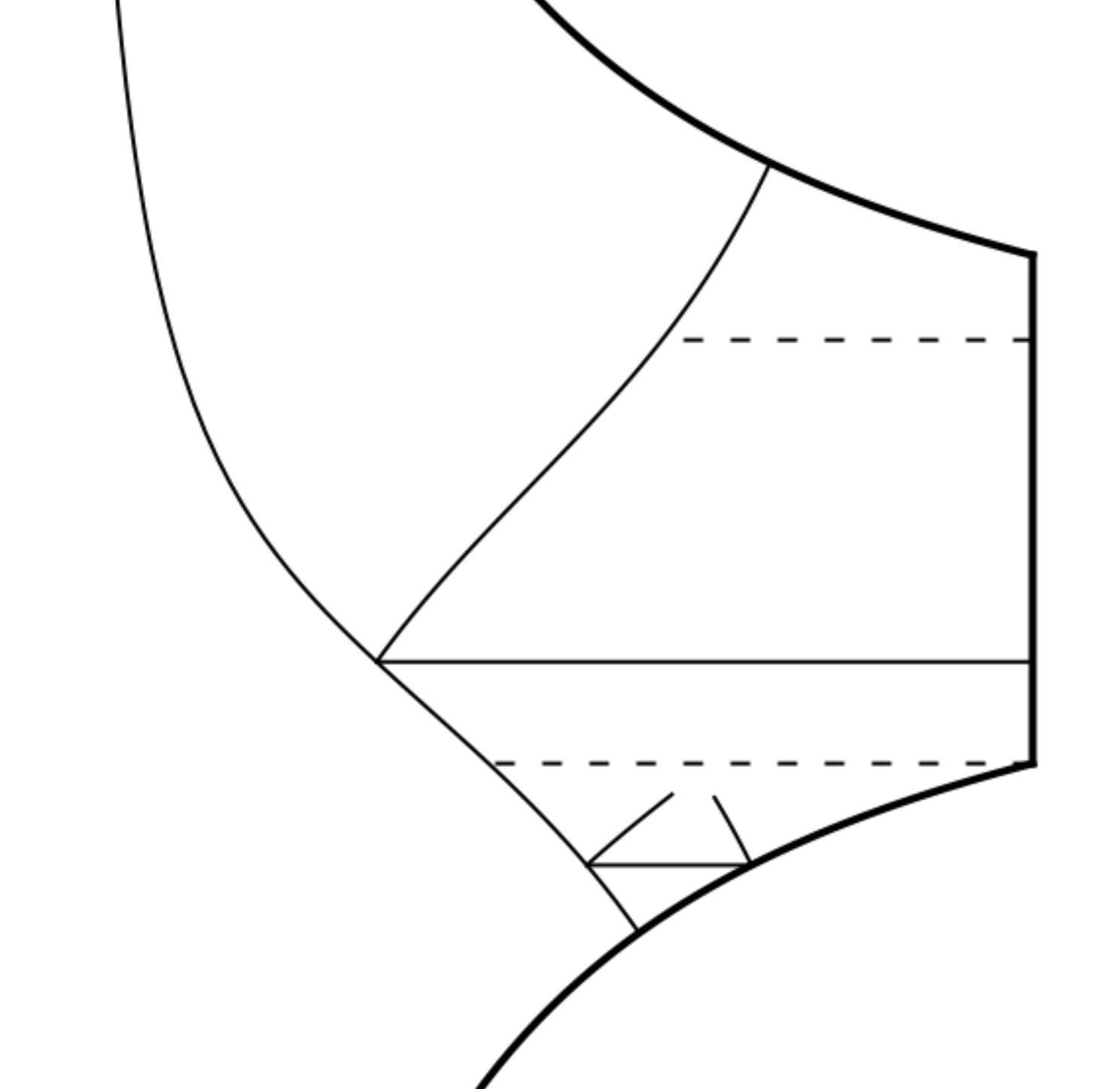}
\put(-70,200){$\mathscr{I}^+$}
\put(-70,40){$\mathscr{I}^-$}
\put(-10,120){$\mathcal{I}$}
\put(-10,175){$\mathcal{I}^+$}
\put(-10,65){$\mathcal{I}^-$}
\put(-130,50){$\underline{\mathcal{B}}_\varepsilon$}
\put(-120,150){$\mathcal{B}_1$}
\put(-70,95){$\hat{\mathcal{S}}_{-1+\varepsilon}$}
\put(-70,162){$\hat{\mathcal{S}}_{\hat{\tau}_\ast}$}
\put(-70,71){$\hat{\mathcal{S}}_{-1}$}
\put(-105,39){\small{$(a)$}}
\put(-98,52){\small{$(b)$}}
\put(-80,56){\small{$(b')$}}
\put(-118,58){\small{$(b'')$}}
\put(-50,162){\small{$(c)$}}
\put(-70,120){$\mathcal{N}_1$}
\put(-90,75){$\underline{\mathcal{N}}_\varepsilon$}
\end{center}
\caption{Schematic depiction of the region $\mathscr{D}=\underline{\mathcal{N}}_\varepsilon\cup \mathcal{N}_1$, now in the $(\hat{\tau}, \rho)$ coordinates. The last slice argument from $\scri^-$ is best carried out in a representation of spatial infinity where past and future null infinities are not horizontal. In region $(a)$, Rendall's theorem ensures the existence of a solution in a neighbourhood of the initial cut $\mathscr{I}^-\cap \underline{\mathcal{B}}_\varepsilon$. This neighbourhood contains a spacelike hypersurface of constant $\hat{\tau}$. Starting from this hypersurface, in region $(b)$ we can use the standard local existence result for symmetric hyperbolic systems to extend the solution. As this extension of the solution possesses Cauchy horizon in $\underline{\mathcal{N}}_\varepsilon$, we solve supplementary characteristic initial value problems in regions $(b')$ and $(b'')$ starting from data on $\scri^-$ and $\underline{\mathcal{B}}_\varepsilon$, and the induced data on the Cauchy horizon. The estimates in $\underline{\mathcal{N}}_\varepsilon$ then ensure that the solution can be extended up to $\hat{\mathcal{S}}_{-1+\varepsilon}$. Once we are in the upper domain $\mathcal{N}_1$, we proceed in a similar fashion. In region $(c)$ the assumption of the existence of a \emph{last slice} is in contradiction with the the estimates in terms of the data on $\hat{\mathcal{S}}_{-1+\varepsilon}$.}
\label{Fig:LastSliceArgument}
\end{figure}
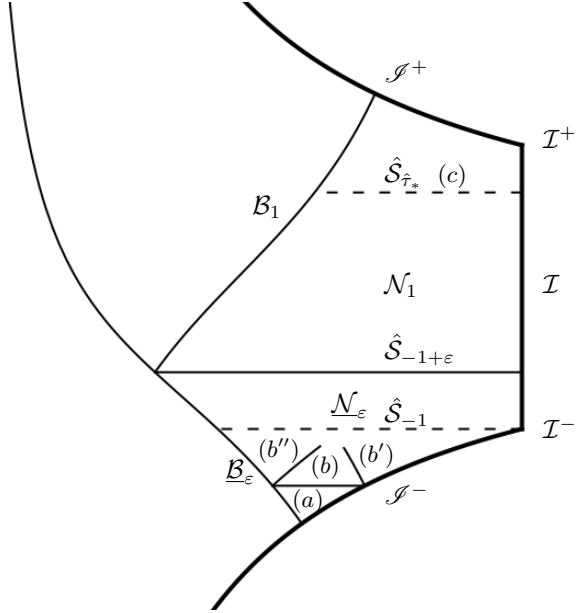

As for the resolution of the asymptotic characteristic initial value problem in \Cref{Section:CharacteristicProblem}, here it is more convenient to consider the massless spin-$s$ equations in the more general gauge discussed in \Cref{Appendix:MoreGeneralGauge}---rather than in the horizontal representation of \Cref{Section:CylinderMinkowski}---as in the latter the equations are singular on the whole of $\mathscr{I}^-$, not only at the critical set $\mathcal{I}^-$. In the more general gauge we have $x^0 = \hat{\tau} \kappa$, where $\kappa = \rho \mu$ and $\mu\not\equiv 1$, and the hypersurfaces of constant $\hat{\tau}$ 
give rise to a foliation of the lower domain $\underline{\mathcal{N}}_{\varepsilon}$, and indeed the region to the past of $\underline{\mathcal{B}}_\varepsilon$. Since $\scri^-$ is now given by $\hat{\tau} = - \mu^{-1}$ and $\mu(0) = 1$, the hypersurfaces $\hat{\mathcal{S}}_{\hat{\tau}}$ for $\hat{\tau} < -1$ terminate at $\scri^-$, whereas the ones for $ -1 < \hat{\tau} < 1$ terminate at $\mathcal{I}$. In particular, the hypersurface $\hat{\mathcal{S}}_{-1}$ terminates at $\mathcal{I}^-$---see \Cref{Fig:LastSliceArgument}. Moreover, there exists a hypersurface $\hat{\mathcal{S}}_{\tau_{\rho_\star}}$ which lies entirely in the past of $\underline{\mathcal{B}}_\varepsilon$, and which intersects $\scri^-$ at $\underline{\mathcal{B}}_\varepsilon\cap\mathscr{I}^-_{\rho_\star}$. We concentrate on the region $\underline{\mathcal{N}}_\varepsilon$, i.e. in what follows we refer to the parts of hypersurfaces $\hat{\mathcal{S}}$ in the future of $\underline{\mathcal{B}}_\varepsilon$.

As we are looking for solutions of the form \eqref{ExpansionLowerDomain} and the asymptotic expansion is formally known, we only need to ensure the existence of the remainder. We therefore consider an evolution system for the fields $\partial_\tau\phi_k^q$ for $m+3<s+q$. This system can be obtained by a repeated application of the derivative $\partial_\tau$ to the equations \cref{StandardEvolutionFirstGeneral,StandardEvolutionFirstGeneral}. The resulting system has the same structure as the original one, so it is a symmetric hyperbolic system for $\partial_\tau\phi_k^q$.

Now, for the asymptotic characteristic initial value problem for the fields $\partial_\tau\phi_k^q$, Rendall's theorem \cite{Ren90} ensures the existence of a solution in a neighbourhood $\mathcal{V}\subset \underline{\mathcal{N}}_\varepsilon\subset J^+(\underline{\mathcal{B}}_\varepsilon\cap\mathscr{I}^-)$ of the initial cut, region $(a)$ in \Cref{Fig:LastSliceArgument}. Next, there exists a value $\hat{\tau}_\delta < -1$ of the parameter $\hat{\tau}$ such that all leaves $\hat{\mathcal{S}}_{\hat{\tau}}$ with $\hat{\tau}_{\rho_\star} < \hat{\tau} < \hat{\tau}_\delta $ are contained entirely in the region $\mathcal{V}$. On one of these hypersurfaces, say $\hat{\mathcal{S}}_{\hat{\tau}_\circ}$, one can formulate a standard initial value problem for the symmetric hyperbolic system \eqref{MasslessSpinSSHS} implied by the spin-$s$ field equations. With data posed on $\hat{\mathcal{S}}_{\hat{\tau}_\circ}$, the standard theory of symmetric hyperbolic systems ensures the existence of a solution in $D^+({\hat{\mathcal{S}}_{\hat{\tau}_\circ}})$, region $(b)$ in \Cref{Fig:LastSliceArgument}. The solution in region $(b)$ possesses a Cauchy horizon $\mathcal{H}^+(\hat{S}_{\hat{\tau}_\circ}) = \partial D^+(\hat{\mathcal{S}}_{\hat{\tau}_\circ})$, depicted by the hypersurfaces separating the region $(b)$ from regions $(b')$ and $(b'')$. We extend the solution to $(b')$ and $(b'')$ by posing two further characteristic initial value problems: (i) one with data on $\mathscr{I}^-$ and the Cauchy horizon of $\mathcal{S}_{\tau_\circ}$, and (ii) a second one with data on $\underline{\mathcal{B}}_\varepsilon$ and the Cauchy horizon; the characteristic data on the Cauchy horizon is induced by the solution in the region $(b)$. For these characteristic initial value problems we again use Rendall's theorem to obtain a solution in a subset of the causal future of $\mathscr{I}^-\cap \mathcal{H}^+(\mathcal{S}_{\tau_\circ})$ (region $(b')$) and $\underline{\mathcal{B}}_\varepsilon\cap \mathcal{H}^+(\mathcal{S}_{\tau_\circ})$ (region $(b'')$), respectively. In this way we obtain the solution in a subdomain of $\underline{\mathcal{N}}_\varepsilon$ which contains a hypersurface $\hat{\mathcal{S}}_{\hat{\tau}}$ with $\hat{\tau}>\hat{\tau}_\circ$. This construction can then be repeated to obtain a larger extension domain for the solution.

Now assume that there exists $\hat{\tau}'_\ast\in(\hat{\tau}_{\rho_\star},-1)$ such that on the hypersurface $\hat{\mathcal{S}}_{\hat{\tau}'_\ast}$ the construction for the enlargement of the domain of existence as described in the previous paragraph fails. It is straightforward to show (as in the case of the upper existence domain) that this is then in contradiction with the estimates in \Cref{thm:radiation_field_estimate} and in \Cref{prop:bootstrap_bound}. In this way one ensures the existence of the solution up to the hypersurface $\hat{\mathcal{S}}_{-1}$.

Once one has reached the hypersurface $\hat{\mathcal{S}}_{-1}$, the domain extension procedure simplifies as now one only needs one supplementary characteristic initial value problem to complete the new slab---the one on the intersection of the Cauchy horizon and $\underline{\mathcal{B}}_\varepsilon$. Again, assuming the existence of a hypersurface $\hat{\mathcal{S}}_{\hat{\tau}_\ast}$  with $\hat{\tau}_\ast>-1$ beyond which it is no longer possible to formulate an initial value problem and extend the solution (region $(c)$ in \Cref{Fig:LastSliceArgument}) leads to a contradiction with the estimates in \Cref{thm:radiation_field_estimate} and \Cref{prop:bootstrap_bound}.

\section{Controlling the solutions from $\mathscr{I}^-$ to $\mathscr{I}^+$}
\label{Section:FromPastToFuture}
We have obtained two sets of estimates: 

\begin{itemize}
\item[(i)] estimates which control the solution on a spacelike hypersurface $\mathcal{S}_{-1+\varepsilon}$ in terms of initial data on $\mathscr{I}^-$ (\Cref{thm:radiation_field_estimate,prop:bootstrap_bound}), and
\item[(ii)] ones which control the solution up to $\mathscr{I}^+$ in terms of initial data on a spacelike hypersurface, say $\mathcal{S}_{-1+\varepsilon}$ (\Cref{thm:forward_estimates}). 
\end{itemize}

\noindent Each of the estimates (i) and (ii) allow the construction of a particular type of asymptotic expansion with its own requirements on the freely specifiable data and associated implications for the regularity of the solutions. In this section we show how these two sets of estimates can be stitched together to control the solutions to the massless spin-$s$ field equation in a neighbourhood of spatial infinity which contains parts of both $\mathscr{I}^-$ and $\mathscr{I}^+$ (\Cref{Fig:CylinderIntro}).

\begin{remark}
We refer to the domain in which estimates (i) hold as the lower domain, and the domain in which estimates (ii) hold as the upper domain. The parameters $m,\;p,\;q$ in the lower domain will be denoted by $m_-,\;p_-,\;q_-$, while those in the upper domain will be denoted by $m_+,\;p_+,\;q_+$.
\end{remark}


\subsection{From $\mathscr{I}^-$ to $\mathcal{S}_{-1+\varepsilon}$}
\Cref{thm:radiation_field_estimate} controls the regularity of the solution $\phi_k$ near $\scri^-$ in the sense that (setting $p_- = 0$, as remarked in \Cref{sec:expansions_near_past_null_infinity}) if the data on $\scri^-_{\rho_\star}\cup \underline{\mathcal{B}}_\varepsilon$ is such that
\[ 
(\partial_\tau^{q_-} \phi, \, \partial_\tau^{q_- + 1} \phi, \, \ldots, \, \partial_\tau^{q_- + m_- + 3} \phi ) \in H^{m_- + 3} \times H^{m_- +2} \times \dots \times L^2 
\]
for
\[ 
m_- + 3 < s + q_-, 
\]
then the solution is $\partial_\tau^{q_-} \phi \in H^{m_- + 3}(\underline{\mathcal{N}}_\varepsilon)$, and in particular for $0 < \alpha \leq \frac{1}{2}$
\[ 
\phi_k = \sum_{q' =0 }^{q_- - 1} \frac{1}{q'!}(\partial_\tau^{q'} \phi_k)|_{\scri^-} (\tau+1)^{q'} + C^{m_-, \alpha}(\underline{\mathcal{N}}_\varepsilon). 
\]

\subsection{From $\mathcal{S}_{-1+\varepsilon}$ to $\mathscr{I}^+$} 
On the other hand, \Cref{thm:forward_estimates} controls the regularity of the solution near $\scri^+$, in the sense that if the data on the Cauchy surface $\mathcal{S}_{-1+\varepsilon}$ is such that
\[ (\partial_\rho^{p_+} \phi, \, \partial_\rho^{p_+ + 1} \phi, \, \ldots, \, \partial_\rho^{p_+ + m_+ + 3} \phi) \in H^{m_+ + 3} \times H^{m_+ + 2} \times \dots \times L^2 \]
for
\[ p_+ \geq m_+ + s + 4, \]
then the solution is $\partial_\rho^{p_+} \phi \in H^{m_+ + 3}(\mathcal{N}_1)$, and in particular
\[ 
\phi_k = \sum_{p' = 0}^{p_+-1} \frac{1}{p'!} \phi_k^{(p')} \rho^{p'} + C^{m_+, \alpha}(\mathcal{N}_1). 
\]

\subsection{Prescribing the regularity at $\mathscr{I}^-$} 
For a desired regularity $m_+$ of the remainder at $\scri^+$, we may therefore put $p_+ = m_+ + \ceil{s} + 4$ for the minimum number of $\rho$-derivatives on $\mathcal{S}_{-1+\varepsilon}$ required by the estimate of \Cref{thm:forward_estimates}. In order to match the estimates in the two domains, in the lower domain we therefore require the remainder to be at least $C^{p_+}$, i.e. $m_- = p_+$. Combining the above bounds then gives
\begin{equation}
\label{bound_on_derivatives_past_null_infinity_by_remainder_regularity_in_future}
q_- \geq m_+ + 8. 
\end{equation}
Note that while in \eqref{bound_on_derivatives_past_null_infinity_by_remainder_regularity_in_future} the value $s$ of the spin cancels out, it still plays a role in length of the expansions of $\phi_k$ via $m_- = p_+ = m_+ + \ceil{s} + 4$. As $s$ increases, so does the length of the expansion before the remainder can be guaranteed to be $C^{m_+}$ near $\scri^+$.

\subsection{Main result}
We summarise the analysis of this article in the following theorem. 

\begin{theorem}
\label{MainTheorem}
Let real numbers $\rho_\star>0$, $\varepsilon>0$ and positive integers $m,\;q$ such that $q\geq m +8$ be given, and suppose that we have asymptotic characteristic initial data for the massless spin-$s$ equations on $\underline{\mathcal{B}}_\varepsilon\cup \mathscr{I}^-_{\rho_\star}$ such that
\[ 
(\partial_\tau^{q} \phi, \, \partial_\tau^{q + 1} \phi, \, \ldots, \, \partial_\tau^{q + m + 3} \phi ) \in H^{m + 3} \times H^{m +2} \times \dots \times L^2.
\]
Then in the domain $\mathscr{D} = \underline{\mathcal{N}}_\varepsilon \cup \mathcal{N}_1$ this data gives rise to a unique solution to the massless spin-$s$ equations \eqref{MasslessSpinS} which near $\scri^+$ admits the expansion
\[
\phi_k = \sum_{p' = 0}^{p-1} \frac{1}{p'!} \phi_k^{(p')} \rho^{p'} + C^{m, \alpha}(\mathcal{N}_1),
\]
where $p= m + \ceil{s} + 4$ and $0 < \alpha \leq \frac{1}{2}$. The coefficients $\phi_k^{(p')}$, $0\leq p'<p-1$, can be computed explicitly in terms of Jacobi polynomials in $\tau$ and spin-weighted spherical harmonics, and their regularity depends on the multipolar structure of the characteristic data at $\mathcal{I}^-$. In particular, the coefficients $\phi_k^{(p')}$ are generically polyhomogeneous at $\scri^+$.
\end{theorem}

\section{Concluding remarks}
\label{Section:Conclusions}
We have constructed estimates and asymptotic expansions for massless spin-$s$ fields on Minkowski space which control the behaviour of the solutions at future null infinity $\mathscr{I}^+$ in terms of data prescribed at past null infinity $\mathscr{I}^-$. In particular, we have established  a connection between the regularity of the characteristic initial data at past null infinity and the regularity of expansions near future null infinity, showing that the linear analogue of the logarithms found by Friedrich \cite{Fri98a} persists despite high regularity assumptions at past null infinity. Of course, the regularity of the data at $\scri^-$ in our setting is measured in Sobolev spaces with respect to the F-coordinates $\tau$ and $\rho$, and we make no claim as to the non-existence of other (perhaps comparatively large) function spaces which may eliminate the logarithmic terms. Nonetheless, Friedrich's coordinates appear natural insofar as they are constructed from conformal invariants.

It is interesting to note that for every logarithmic term near $\scri^+$, there is a corresponding logarithmic term near $\scri^-$ (see \eqref{JacobiSolution:Logarithmic}), however it generally carries a different tempering polynomial in front, cf. $(1-\tau)^{p-k+s} \log(1-\tau)$ vs. $(1+\tau)^{p+k-s} \log(1+\tau)$. In a sense, this is the statement that singular behaviour hidden deeper in the expansions near $\scri^-$ may resurface at higher orders near $\scri^+$. It is also interesting to note that, at the level of the Cauchy data, polyhomogeneity is \emph{not needed} to produce polyhomogeneity in the solutions near $\scri^\pm$.  A range of examples will be presented in a subsequent paper.

Finally, we mention that our assumptions on the asymptotic characteristic initial data have been made on the basis of ease of presentation. Our theory can, however, applied to a more general setting, e.g. for polyhomogeneous data at $\mathscr{I}^-$.

\section*{Appendices}

\appendix

\section{Geometry of $\mathrm{SU}(2)$}
\label{Appendix:SU2}
In this appendix we recall standard  results about the geometry and representation theory of the Lie group
$\mathrm{SU}(2)$.

\subsection{Basic properties}
We make use of spinorial notation to denote elements
in $\mathrm{GL}(2,\mathbb{C})$, so that $t^\bmA{}_\bmB$, $\bmA, \bmB \in \{0, 1 \}$, denotes an invertible $2\times 2$ matrix. The subgroups $\mathrm{SL}(2, \mathbb{C})$ and $\mathrm{SU}(2, \mathbb{C})$ may then be defined by
\begin{eqnarray*}
&& \mathrm{SL}(2,\mathbb{C}) =\left\{  t^\bmA{}_\bmB \in \mathrm{GL}(2,\mathbb{C}) \; | \;
    \epsilon_{\bmA\bmB} t^\bmA{}_\bmC t^\bmB{}_\bmD =
   \epsilon_{\bmC\bmD}  \right\}, \\
&& \mathrm{SU}(2,\mathbb{C}) = \left\{ t^\bmA{}_\bmB \in \mathrm{SL}(2,\mathbb{C}) \; | \;
    \tau_{\bmA\bmA'} t^\bmA{}_\bmB \bar{t}^{\bmA'}{}_{\bmB'} =    \tau_{\bmB\bmB'} \right\},
 \end{eqnarray*}
where $\tau_{\bmA \bmA'}$ in our frame is simply the $2\times 2$ identity matrix. It is classical that the group $\mathrm{SU}(2)$ is diffeomorphic to the 3-sphere $\mathbb{S}^3$, and its elements $t^\bmA{}_\bmB$ may be written in the explicit form

\[ t^\bmA{}_\bmB = \frac{1}{\sqrt{1 + |\zeta|^2}} \left( \renewcommand{\arraystretch}{1.2} \begin{array}{cc} e^{i\alpha} & i e^{-i \alpha} \zeta \\ i e^{i\alpha} \bar{\zeta} & e^{-i\alpha} \end{array} \right), \]
where $\zeta = x + iy \in \mathbb{C}$ and $\alpha \in \mathbb{R}$ coordinatize $\mathrm{SU}(2) \simeq \mathbb{S}^3$. This representation makes manifest the fact that there exist $\mathrm{U}(1)$ orbits in $\mathrm{SU}(2)$ generated by $\alpha$, with the coordinates $(x, y)$ constant along these orbits. Quotienting out the $\mathrm{U}(1)$ subgroup returns the $2$-sphere of the rescaled spacetime, $\mathrm{SU}(2) / \mathrm{U}(1) \simeq \mathbb{S}^2$, as described in \Cref{Section:NullFrameNearI0}.

\subsection{The vector fields \texorpdfstring{$\bmX_\pm$ and $\bmX$}{}}
Consider the basis
\[
\bmu_1 = \frac{1}{2}
\left(
\begin{array}{cc}
0 & i \\
i & 0
\end{array}
\right), \qquad
\bmu_2 =\frac{1}{2}
\left(
\begin{array}{cc}
0 & -1 \\
1 & 0
\end{array}
\right), \qquad
\bmu_3 = \frac{1}{2}
\left(
\begin{array}{cc}
i & 0 \\
0 & -i
\end{array}
\right),
\]
of the real Lie algebra $\mathfrak{s}\mathfrak{u}(2)$ of $\mathrm{SU}(2)$, where
$\bmu_3$ is the generator of the $\mathrm{U}(1)$ subgroup. These obey the commutation relations
\[
[\bmu_i, \, \bmu_j] =\epsilon_{ijk} \bmu_k.
\]
 Denote by $\bmZ_1$,
$\bmZ_2$ and $\bmZ_3$ the left-invariant vector fields on $\mathrm{SU}(2)$
generated by $\bmu_1$, $\bmu_2$ and $\bmu_3$, respectively, via
the exponential map. These inherit the commutation relations
\[ [ \bmZ_i, \, \bmZ_j ] = \epsilon_{ijk} \bmZ_k. \]
We then pass to the complexified Lie algebra $\mathfrak{su}(2) + i \, \mathfrak{su}(2) = \mathfrak{sl}(2; \mathbb{C})$ by setting
\[
\bmX_+ \equiv -(\bmZ_2 +\mbox{i} \bmZ_1), \qquad \bmX_-\equiv -(\bmZ_2
-\mbox{i} \bmZ_1), \quad \text{and} \quad \bmX \equiv -2 \mbox{i} \bmZ_3.
\]
The vector fields $\bmX_\pm$ and $\bmX$ then satisfy the commutation relations
\[
[\bmX, \, \bmX_+] = 2\bmX_+, \qquad [\bmX, \, \bmX_-]=-2\bmX_-, \quad \text{and} \quad [\bmX_+, \, \bmX_-]=-\bmX.
\]
The vector fields $\bmX_+$ and $\bmX_-$ are complex conjugates of each
other in the sense that
\[
\overline{\bmX_+ \phi} = \bmX_- \bar{\phi}
\]
for any sufficiently smooth function $\phi : \mathrm{SU}(2) \to \mathbb{C}$. The Casimir element of $\mathrm{SU}(2)$ is given by
\[ \bmC \equiv 4( \bmZ_1^2 + \bmZ_2^2 + \bmZ_3^2 )= 2 \{ \bmX_+, \, \bmX_- \} -  \bmX^2, \]
where $\{ \bmX_+, \, \bmX_- \} = \bmX_+ \bmX_- + \bmX_- \bmX_+$ is the anticommutator of $\bmX_+$ and $\bmX_-$. The $2 \{ \bmX_+, \, \bmX_- \} = 4( \bmZ_1^2 + \bmZ_2^2$ ) part of the Casimir corresponds to the Laplacian on the $2$-sphere $\mathbb{S}^2$ when acting on functions $f$ independent of $\alpha$, i.e. ones that satisfy $\bmX f = 0$.

\subsubsection{Coordinate expressions}
The vector fields $\bmpartial_x$, $\bmpartial_y$, $\bmpartial_\alpha$ at the
identity in $\mathrm{SU}(2)$ coincide, respectively, with the generators $\bmu_1$,
$\bmu_2$ and $\bmu_3$ of
$\mathfrak{s}\mathfrak{u}(2)$. More generally, writing 
\[
P\equiv \frac{1}{2}(1+ |\zeta|^2),
\]
one may express the vector fields $\bmZ_i$ in terms of $\bmpartial_x$, $\bmpartial_y$, $\bmpartial_\alpha$ as
\begin{eqnarray*}
&& \bmZ_1 = (P \cos 2\alpha ) \bmpartial_x + (P \sin2\alpha) \bmpartial_y +
   \frac{1}{2}\big( x\sin 2\alpha - y \cos2\alpha \big)
   \bmpartial_\alpha, \\
&& \bmZ_2 = - (P \sin 2\alpha) \bmpartial_x+ (P\cos 2\alpha) \bmpartial_y
   + \frac{1}{2}\big(y \sin 2\alpha + x \cos 2\alpha
   \big)\bmpartial_\alpha, \\
&& \bmZ_3 = \frac{1}{2}\bmpartial_\alpha.
\end{eqnarray*}
Then
\begin{eqnarray*}
    && \bmX_+ = - 2 \mathrm{i} P e^{2 \alpha\mathrm{i}} \bmpartial_\zeta - \frac{1}{2} \bar{\zeta} e^{2 \alpha\mathrm{i}} \bmpartial_\alpha, \\
    && \bmX_- = 2 \mathrm{i} P e^{- 2 \alpha\mathrm{i}} \bmpartial_{\bar{\zeta}} - \frac{1}{2} \zeta e^{- 2 \alpha\mathrm{i}} \bmpartial_\alpha, \\
    && \bmX = - \mathrm{i} \bmpartial_\alpha.
\end{eqnarray*}
When acting on functions $f$ such that $\partial_\alpha f = 0$, in these coordinates the Casimir $\bmC$ has the form
\[ 
\bmC f = 2 \{ \bmX_+, \, \bmX_- \} f  = 4 (1+ |\zeta|^2)^2 \partial_\zeta \partial_{\bar{\zeta}} f, \]
which is precisely the Laplacian on $\mathbb{S}^2$ in stereographic coordinates $\zeta = x + \mathrm{i} y$.

\subsubsection{A technical lemma}
\label{Appendix:TechnicalLemmaSU2C}
For a given multi-index $\alpha =( \alpha_1,\alpha_2,\alpha_3)$ of
non-negative integers $\alpha_1$, $\alpha_2$ and $\alpha_3$, set
\[
\bmZ^\alpha \equiv \bmZ^{\alpha_1}_1 \bmZ^{\alpha_2}_2
\bmZ^{\alpha_3}_3 
\]
with the convention that $\bmZ^\alpha =1$ if $|\alpha| = \alpha_1+\alpha_2+\alpha_3 =0$. We record the following technical lemma from \cite{Fri03b}, which will be useful to us when performing estimates in the main text.

\begin{lemma}
\label{Lemma:IntegrationSU2C}
For any smooth complex-valued functions $\phi$, $\psi$ on
$\mathrm{SU}(2)$, $k \in \{1, 2, 3 \}$ and $m \in \mathbb{N}$, one has
\[
\sum_{|\alpha|=m} \big( \bmZ^\alpha \bmZ_k \phi \bmZ^\alpha \psi +
\bmZ^\alpha\phi \bmZ^\alpha \bmZ_k \psi\big) = \sum_{|\alpha| = m} \bmZ_k \left(
  \bmZ^\alpha \phi \bmZ^\alpha\psi \right).
\]
In particular, one has that
\begin{equation}
\label{FriedrichsHigherOrderIntegrationByParts}
\sum_{|\alpha|=m} \int_{\mathrm{SU}(2)} \big( \bmZ^\alpha \bmX_\pm
\phi \bmZ^\alpha \psi + \bmZ^\alpha \phi \bmZ^\alpha \bmX_\pm \psi
\big)\d \mu =0.
\end{equation}

\end{lemma}

\begin{remark}
\Cref{Lemma:IntegrationSU2C} is an extension of the divergence theorem on $\mathrm{SU}(2)$. Indeed, for $m = 0$ the statement of \eqref{FriedrichsHigherOrderIntegrationByParts} reduces to 
\[
\int_{\mathrm{SU}(2)} \bmX_{\pm}(\phi \psi) \d \mu = 0.
\]
This, in turn, follows from the divergence theorem and the fact that left-invariant vector fields on unimodular Lie groups are divergence free.
\end{remark}

\subsection{The functions $T_m{}^j{}_k$}
\label{Section:Representations_SU2}

It will be useful to explicitly define the matrix elements of the unitary irreducible representations of $\mathrm{SU}(2)$ in the conventions of \cite{Fri86a}. The irreducible representations of $\mathrm{SU}(2)$ are uniquely labelled by the natural numbers $m \in \mathbb{N}$. For each $m \in \mathbb{N}$ there exists a unique unitary irreducible representation $T_m$, and the set of irreducible representations $\{ T_m \}_{m \in \mathbb{N}}$ contains all the unitary irreducible representations of $\mathrm{SU}(2)$. The dimension of each $T_m$ is $m+1$, and each $T_m$ is an eigenfunction of the Casimir operator $\bmC$ on $\mathrm{SU}(2)$ with eigenvalue $-m(m+2)$. The matrix elements of these representations are given by the complex-valued functions $T_m{}^j{}_k$ defined by
\[
\mathrm{SU}(2) \ni t^\bmA{}_\bmB \mapsto T_m{}^j{}_k
(t^\bmA{}_\bmB) \equiv \binom{m}{j}^{1/2} \binom{m}{k}^{1/2}
 t^{(\bmA_1}{}_{(\bmB_1}\cdots t^{\bmA_m)_j}{}_{\bmB_m)_k} 
\]
for $ j,\, k = 0, \,\ldots, \, m$,  $m=1,\, 2,\, 3,\ldots $, and
\[
T_0{}^0{}_0 =1,
\]
where the notation $(\bmA_1, \, \ldots, \, \bmA_m)_j$ indicates that the $\bmA$ indices are symmetrized, and then $j$ of them are set to $1$ and the remaining $m-j$ are set to $0$. The functions are real-analytic and the associated
representation $T_m$ is then given by
\[
\mathrm{SU}(2) \ni t^\bmA{}_\bmB \mapsto T_m (t^\bmA{}_\bmB) =
T_m{}^j{}_k(t^\bmA{}_\bmB) \in \mathrm{SU}(m+1).
\]
By the Schur orthogonality relations and the Peter--Weyl Theorem, the functions 
\[
\sqrt{m+1}\,T_m{}^j{}_k
\]
form a Hilbert basis for the Hilbert space $L^2(\mathrm{SU}(2), \mu)$ where $\mu$ denotes the normalised Haar
measure on $\mathrm{SU}(2)$. In particular, any complex analytic function $\phi$ on $\mathrm{SU}(2)$ admits the expansion
\[
\phi(t^\bmA{}_\bmB)= \sum^\infty_{m=0} \sum^m_{j=0} \sum^m_{k=0}
\phi_{m,k,j} T_m{}^j{}_k,
\]
with complex coefficients $\phi_{m,k,j}$ which decay rapidly as
$m\rightarrow \infty$. From the earlier definition, one can check that under complex conjugation we have
\[
\overline{T_m{}^j{}_k} =(-1)^{j+k} T_m{}^{m-j}{}_{m-k}.
\]
Moreover, one has for $0\leq j,\, k\leq m$, $m=0,\,1,2, \,\ldots$, that 
\begin{eqnarray*}
&& \bmX T_m{}^j{}_k = (m-2k) T_m{}^j{}_k, \\
&& \bmX_+ T_m{}^j{}_k = \beta_{m,k} T_m{}^j{}_{k-1},\\
&& \bmX_- T_m{}^j{}_k =  -\beta_{m,k+1} T_m{}^j{}_{k+1},
\end{eqnarray*}
where
\[
\beta_{m,k} \equiv \sqrt{k(m-k+1)}.
\]
In particular, the $2$-sphere Laplacian acts on the representations $T_m{}^j{}_k$ by
\[ 
\{ \bmX_+, \, \bmX_- \} T_m{}^j{}_k = (2k(k-m) - m ) T_m{}^j{}_k. 
\]
These formulae ensure that for a function $\phi$ with spin weight $\varsigma \in \frac{1}{2} \mathbb{N}$, i.e. satisfying
\[
\bmX \phi = 2 \varsigma \phi,
\]
the expansions above reduce to
\[
\phi = \sum_{m\leq |2\sigma|} \sum_{j=0}^m \phi_{m,j} T_m{}^j{}_{m/2-\sigma},
\]
where $m$ takes even values if $\sigma$ is an integer and odd values if
$\sigma$ is a half-integer \cite{Fri86a}. In particular, we note that for a given $k$ ($0 \leq k \leq 2s$), the component $\phi_k$ has spin weight $(k-s)$; that is, $\bmX \phi_k = 2(k-s) \phi_k$.

\begin{remark} The matrix coefficients $T_m{}^j{}_k$ as defined above are related to the perhaps more widely used spin-weighted spherical harmonics ${}_sY_{lm}$, as well as Wigner's $D$-matrices $\mathcal{D}^j_{m'm}$. One has
\begin{equation}
{}_s Y_{lm} = (-1)^{s+m} \sqrt{\frac{2l+1}{4 \pi}} T_{2l}{}^{l-m}{}_{l-s} 
\label{SpinWeightedHarmonicsCorrespondence}
\end{equation}
for $l \in \mathbb{N} \cup \{ 0 \}$, and $m, \, s \in \{ -l, -l+1, \, \ldots, \, l-1, l \}$ \cite{BeyDouFraWha12}.
Further,
\[ \mathcal{D}^j_{m' m} \propto T_{2j}{}^{j-m}{}_{j-m'}, \]
where $j \in \frac{1}{2} \mathbb{N}$ and $m', \, m \in \{ -j, -j+1 , \, \ldots, \, j - 1, j \}$, \cite{Wig59}.
\end{remark}

\section{Spin-$s$ equations}
\label{Appendix:SpinSEqns}

In this appendix we provide a derivation of the massless spin-$s$ equations in the F-gauge. Let
\[
\phi_{A_1\dots A_{2s}} =\phi_{(A_1\dots A_{2s})}
\]
denote a totally symmetric spinor of valence $2s$, with $s \in \frac{1}{2} \mathbb{N}$. The massless spin-$s$ equations for $\phi$ then read
\begin{equation}
\nabla^Q{}_{A'} \phi_{QA_1\dots A_{2s-1}}=0.
\label{SpinshalfEquation}
\end{equation}

\subsection{Hyperbolic reduction}
The system \eqref{SpinshalfEquation} is not manifestly symmetric hyperbolic, but a \emph{hyperbolic reduction} may be obtained by making use of the \emph{space-spinor} formalism (see e.g. \cite{CFEBook}, \S4) as follows. Let $\tau^{AA'}$ denote the spinorial counterpart of the timelike vector field $\tau^a$, with normalization $\tau_{AA'}\tau^{AA'}=2$. We define a spin dyad
\[
\{\epsilon_\bmA{}^A\} = \{ o^A,\, \iota^A  \}
\]
adapted to $\tau^{AA'}$ by requiring that
\[
\tau^{AA'}=o^A\bar{o}^{A'}
+\iota^A \bar{\iota}^{A'}.
\]
It then follows that 
\[
\tau_{AA'}\tau^{BA'}=\delta_A{}^B.
\]
Defining $\nabla_{AB}\equiv \tau_B{}^{A'}\nabla_{AA'}$, one has the decomposition
\[
\nabla_{AB} = \frac{1}{2}\epsilon_{AB}\mathcal{D}+\mathcal{D}_{AB},
\]
where 
\[
\mathcal{D} \equiv \tau^{AA'}\nabla_{AA'} \quad \text{and} \quad \mathcal{D}_{AB} \equiv \tau_{(A}{}^{A'}\nabla_{B)A'}
\]
denote, respectively, the so-called \emph{Fermi} and \emph{Sen} derivative operators. The operators $\mathcal{D}$ and $\mathcal{D}_{AB}$ correspond, respectively, to temporal and spatial parts of $\nabla_{AA'}$. Using this decomposition in equation \eqref{SpinshalfEquation}, one gets 
\[
\mathcal{D}\phi_{A_1\dots A_{2s}} -2 \mathcal{D}^Q{}_{A_1}\phi_{A_2\dots A_{2s} Q}=0.
\]
The above equation has two irreducible components: the totally symmetric part and the trace. That is, 
\begin{subequations}
\begin{eqnarray}
&& \mathcal{D}\phi_{A_1\dots A_{2s}} -2 \mathcal{D}^Q{}_{(A_1}\phi_{{A_2}\dots A_{2s})Q}=0, \label{SpinshalfEquationEvolution} \\
&& \mathcal{D}^{PQ} \phi_{PQA_1\dots A_{2s-2}}=0. \label{SpinshalfEquationConstraint}
\end{eqnarray}
\end{subequations}
Equations \eqref{SpinshalfEquationEvolution} and \eqref{SpinshalfEquationConstraint} will be referred to as the \emph{evolution} and \emph{constraint} equations, respectively. 

\begin{remark}
\label{Remark:PropagationConstraints}
It can be shown, through a standard propagation of constraints argument, that if \eqref{SpinshalfEquationConstraint} is satisfied on some initial hypersurface $\{\tau = 0\}$, then  it is also satisfied at later times whenever \eqref{SpinshalfEquationEvolution} holds.
\end{remark}

\begin{remark}
Note that the above decomposition into evolution and constraint equations fails when $s = \frac{1}{2}$. In this case $\phi_A$ has only two independent components and there are two corresponding equations,
\[ \mathcal{D} \phi_A - 2 \mathcal{D}^{Q}_{\phantom{Q}A} \phi_Q = 0, \]
so there are no constraints.
\end{remark}

\subsection{Transport equations along null geodesics}
\label{Appendix:TransportEqnsBasic}

Given a spinor $\mu_{AB\dots C}$, we denote its components with respect to the spin dyad $\{ \epsilon_\bmA{}^A\}$ by $\mu_{\bmA\bmB\dots \bmC}$. In particular, then
\begin{align*}
\nabla_{\bmA\bmA'}\mu_{\bmB\dots \bmC} &= \epsilon_\bmA{}^A\epsilon_{\bmA'}{}^A \epsilon_{\bmB}{}^B\dots\epsilon_{\bmC}{}^C \nabla_{AA'}\mu_{B\dots C} \\
& = \partial_{\bmA\bmA'} \mu_{\bmB\dots \bmC}  - \Gamma_{\bmA\bmA'}{}^\bmQ{}_\bmB \mu_{\bmQ\dots \bmC} - \dots -\Gamma_{\bmA\bmA'}{}^\bmQ{}_\bmD \mu_{\bmB\dots \bmQ},
\end{align*}
where 
\begin{equation}
\label{directionalderivatives}
\partial_{\bmA\bmA'} \equiv e_{\bmA\bmA'}{}^\mu \partial_\mu
\end{equation}
denotes the directional derivatives of the Newman--Penrose (NP) frame $\{ \bme_{\bmA\bmA'}\}$ associated to the spin dyad $\{\epsilon_\bmA{}^A \}$, and $\Gamma_{\bmA\bmA'}{}^\bmC{}_\bmD$ are the corresponding spin connection coefficients. In order to give more explicit expressions, it is convenient to define the following basis of symmetric $(0,2)$-spinors with unprimed indices,
\[
    \sigma^0_{AB} \equiv o_A o_B, \qquad \qquad \sigma^1_{AB} \equiv \iota_{(A} o_{B)}, \qquad \qquad \sigma^2_{AB} \equiv \iota_A \iota_B.
\]
The spinors $\sigma^i_{AB}$, $i \in \{0,1,2\}$, satisfy the orthogonality relations 
\begin{align*} \sigma^0_{AB} (\sigma^0)^{AB} = 0, && \sigma^0_{AB} (\sigma^1)^{AB} = 0, && (\sigma^0)^{AB} \sigma^2_{AB} = 1, \\
\sigma^1_{AB} (\sigma^1)^{AB} = - \frac{1}{2}, && \sigma^1_{AB} (\sigma^2)^{AB} = 0, && \sigma^2_{AB} (\sigma^2)^{AB} = 0.
\end{align*}
For higher valence spinors, we similarly define, for half-integer $s$, the basis
\[
\sigma^0_{A_1\dots A_{2s}} \equiv o_{(A_1}\dots o_{A_{2s})}, \quad \sigma^1_{A_1\dots A_{2s}} \equiv \iota_{(A_1}o_{A_2}\dots o_{A_{2s})}, \quad \dots \quad \sigma^{2s}_{A_1\dots A_{2s}} \equiv \iota_{(A_1}\dots \iota_{A_{2s})}. 
\]
By contracting with $\tau^{\bmA \bmA'}$, the directional derivatives \eqref{directionalderivatives} can be decomposed into temporal and spatial parts as
\[ \partial_{\bmA\bmA'} = \displaystyle\frac{1}{2}\tau_{\bmA\bmA'}\partial -\tau^\bmQ{}_{\bmA'}\partial_{\bmA\bmB}, \]
where we write $\partial\equiv \tau^{\bmA\bmA'}\partial_{\bmA\bmA'}$ and $\partial_{\bmA\bmB}\equiv \tau_{(\bmA}{}^{\bmA'}\partial_{\bmB)\bmA'}$.
In the particular case of the F-gauge used in the main text, one has that
\begin{eqnarray*}
&& \partial =\sqrt{2} \partial_\tau, \\
&& \partial_{\bmA\bmB} = \sqrt{2} \sigma^1_{\bmA\bmB} \left(-\tau \partial_\tau + \rho\partial_\rho \right) + \frac{1}{\sqrt{2}} \sigma^0_{\bmA\bmB} \bmX_+ - \frac{1}{\sqrt{2}} \sigma^2_{\bmA\bmB} \bmX_-, \\
&& \Gamma_{\bmA\bmA'\bmC\bmD}= -\frac{1}{\sqrt{2}}\tau_{\bmA\bmA'} \sigma^1_{\bmC\bmD}. 
\end{eqnarray*}
In particular, it follows that the spatial part $\Gamma_{\bmA \bmB \bmC \bmD} = \tau_{(\bmB}^{\phantom{\bmB}\bmB'} \Gamma_{\bmA) \bmB' \bmC \bmD}$ of the spin connection coefficients vanishes. Thus, the only non-trivial part of the connection is given by the \emph{acceleration} $f_{\bmA \bmB}\equiv -\tau^{\bmC\bmC'}\Gamma_{\bmC\bmC'\bmA\bmB}$. In fact it is easy to see that
\[
f_{\bmA \bmB} = \sqrt{2} \sigma^1_{\bmA \bmB}.
\]
As a consequence of $\Gamma_{\bmA \bmB \bmC \bmD} = 0$ and the fact that the components $\sigma^k_{\bmA_1\dots\bmA_{2s}}$ of $\sigma^k_{A_1\dots A_{2s}}$ are constants, it follows that
\[
\mathcal{D}_{\bmA\bmB} \sigma^k_{\bmA_1\dots \bmA_{2s}} = \partial_{\bmA\bmB} \sigma^k_{\bmA_1\dots\bmA_{2s}}=0,
\]
for all $0 \leq k \leq 2s$. On the other hand, a calculation shows that
\[
\mathcal{D}\sigma^k_{\bmA_1\dots \bmA_{2s}}=\frac{1}{\sqrt{2}}\sigma^k_{\bmA_1\dots \bmA_{2s}}.
\]
In order to write down the equations \eqref{SpinshalfEquationEvolution} and \eqref{SpinshalfEquationConstraint} in our gauge, we now expand $\phi_{A_1 \dots A_{2s}}$ in terms of the basis elements $\sigma^k_{A_1 \dots A_{2s}}$ as
\begin{equation}
\phi_{A_1 \dots A_{2s}} = \sum_{k=0}^{2s} (-1)^k \binom{2s}{k}  \phi_{2s-k} \sigma^k_{A_1 \dots A_{2s}},
\label{ExpansionSpinS}
\end{equation}
where
\[
\phi_0\equiv  \phi_{A_1\dots A_{2s}}o^{A_1}\dots o^{A_{2s}}, \quad \phi_1\equiv \phi_{A_1\dots A_{2s}}\iota^{A_1}o^{A_2}\dots o^{A_{2s}}, \quad \ldots  \quad \phi_{2s} \equiv \phi_{A_1\dots A_{2s}}\iota^{A_1}\dots \iota^{A_{2s}}
\]
are the components of the massless spinor $\phi_{A_1 \dots A_{2s}}$. That is, the component $\phi_k$ is obtained from $k$ contractions with $\iota^A$ and $2s-k$ contractions with $o^A$. By plugging in the expansion \eqref{ExpansionSpinS} into the equations \eqref{SpinshalfEquationEvolution} and \eqref{SpinshalfEquationConstraint}, one may derive the $2s+1$ scalar evolution equations and $2s-1$ scalar constraint equations satisfied by the components $\phi_k$, $0 \leq k \leq 2s$. These turn out to be
\begin{subequations}
\begin{eqnarray}
&& E_0\equiv (1-\tau) \partial_\tau \phi_0 +\rho\partial_\rho \phi_0 -\bmX_-\phi_1 - s\phi_0=0, \label{StandardEvolutionFirst}\\
&& E_k\equiv \partial_\tau \phi_k -\frac{1}{2}\bmX_+\phi_{k-1}-\frac{1}{2}\bmX_-\phi_{k+1}+(k-s)\phi_k=0, \qquad k=1, \, \dots, \, 2s-1,\label{StandardEvolutionMiddle}\\
&& E_{2s}\equiv (1+\tau) \partial_\tau \phi_{2s} -\rho\partial_\rho\phi_{2s} -\bmX_+ \phi_{2s-1} + s\phi_{2s}=0, \label{StandardEvolutionLast}
\end{eqnarray}
\end{subequations}
and
\[
C_k\equiv \tau \partial_\tau \phi_k -\rho\partial_\rho\phi_k +\frac{1}{2}\bmX_-\phi_{k+1}-\frac{1}{2}\bmX_+\phi_{k-1}=0, \qquad k=1, \, \dots, \, 2s-1.
\]

\begin{remark} Note that the equations $C_k = 0$ are termed \emph{constraints} despite containing $\partial_\tau$ derivatives. Indeed, the quantities $C_k$ are propagated as noted in \Cref{Remark:PropagationConstraints}.
\end{remark}

The analysis in the main text makes use of certain combinations of the above evolution and constraint equations. We set
\[ 
A_k \equiv E_{k+1} +C_{k+1} \quad \text{for} ~~ k = 0, \, \dots, \, 2s - 2 \quad \text{and} \quad A_{2s-1} \equiv E_{2s},  \]
and
\[ 
B_0 \equiv E_0 \quad \text{and} \quad B_k \equiv E_k-C_k \quad \text{for} ~~ k = 1, \, \dots, \, 2s - 1. 
\]
Explicitly, these are
\begin{subequations}
\begin{eqnarray}
&& A_k = (1+\tau)\partial_\tau \phi_{k+1} -\rho \partial_\rho \phi_{k+1} -\bmX_+ \phi_{k} +(k+1-s)\phi_{k+1} =0, \label{AEqnAppendix} \\
&& B_k = (1-\tau)\partial_\tau \phi_k + \rho\partial_\rho\phi_k -\bmX_-\phi_{k+1} +(k-s)\phi_k =0, \label{BEqnAppendix}
\end{eqnarray}
\end{subequations}
for $k=0, \, \dots, \, 2s-1$. The equations $A_k = 0$ and $B_k = 0$ are, respectively, transport equations along outgoing and incoming null geodesics in Minkowski space in the F-gauge (introduced in \Cref{Section:CylinderMinkowski}), and we shall refer to them as the \emph{outgoing equations} and \emph{incoming equations} respectively. A crucial feature of the outgoing and incoming equations is that they become degenerate at $\tau = -1$ and $\tau = + 1$ respectively. Further, we observe here that for a given $k$ the pair $(A_k, B_k)$ involves only the components $\phi_k$ and $\phi_{k+1}$.

\subsection{Wave equations}
\label{Appendix:WaveEquations}

It will be useful to note that the spinor $\phi_{A_1\dots A_{2s}}$ satisfies the wave equation. Applying $\nabla_P^{\phantom{P}A'}$ to equation \eqref{SpinshalfEquation} and using the decomposition 
\[
\nabla_{PA'} \nabla_Q{}^{A'}=\tfrac{1}{2}\epsilon_{PQ}\square +\square_{PQ},
\]
one finds that 
\[
\square \phi_{A_1\dots A_{2s-1} \textcolor{red}{P}}
+ 2\square_P{}^Q \phi_{QA_1\dots A_{2s-1}}=0,
\]
where $\square_{AB}\equiv \nabla_{A'(A}\nabla_{B)}{}^{A'}$ is the \emph{Penrose box} encoding the commutator of two $\nabla_{AA'}$ derivatives. Note that if one contracts $P$ with any of the remaining free indices in the above equation, the first term vanishes by the symmetry of $\phi_{A_1 \ldots A_{2s}}$, implying that the second term must also vanish. But, using the fact that on a spinor $\kappa^C$, $\Box_{AB}$ acts by $\Box_{AB} \kappa^C = X_{ABE}{}^C \kappa^E$ and $\Box_{A}{}^B \kappa_B = 3 \Lambda \kappa_A = \mathrm{R} \kappa_A / 8$, where $X_{ABCD} = \Psi_{ABCD} + \Lambda(\varepsilon_{AC} \varepsilon_{BD} + \varepsilon_{AD} \varepsilon_{BC})$, the second term may be expressed entirely in terms of products of $\phi$ with the Weyl spinor $\Psi_{ABCD}$. When $s>1$, one therefore discovers the well-known consistency condition\footnote{There are no such restrictions in the cases $s=1/2$ and $s=1$, i.e. for the Dirac and Maxwell fields. Moreover, while the spin-$2$ equation is unsatisfactory on a \emph{given} non-conformally flat spacetime, the Weyl spinor $\Psi_{ABCD}$ itself satisfies the spin-$2$ equation in \emph{dynamical} general relativity (in vacuum), when the above consistency condition becomes $\Psi_{ABM(C} \Psi_{D)}{}^{ABM} = 0$, which is trivial by the symmetries of $\Psi_{ABCD}$.} \cite{PenRin84} for higher spin fields,
\[ \phi_{ABM(C\ldots K} \Psi_{L)}{}^{ABM} = 0, \]
i.e. the requirement that the background spacetime be conformally flat.  Since we are working on conformally rescaled Minkowski space, this consistency condition is satisfied and one obtains the simple wave equation
\begin{equation}
\square \phi_{A_1\dots A_{2s}}=0.
\label{WaveEqnSpinS}
\end{equation}
Using the splitting
\[
\nabla_{AA'} =\frac{1}{2}\tau_{AA'}\mathcal{D} -\tau^Q{}_{A'} \mathcal{D}_{AQ},
\]
in the F-gauge the wave operator \eqref{WaveEqnSpinS} on $\phi_{A_1 \dots A_{2s}}$ may be written as
\[
\square \phi_{A_1\dots A_{2s}} =\frac{1}{2}\mathcal{D}^2 \phi_{A_1\dots A_{2s}}
 + \mathcal{D}^{AB}\mathcal{D}_{AB}\phi_{A_1\dots A_{2s}}- \sqrt{2} (\sigma^1)^{AB}\mathcal{D}_{AB}  \phi_{A_1\dots A_{2s}}.
 \]
Scalarising this equation and writing the derivative operators $\mathcal{D}$ and $\mathcal{D}_{AB}$ in terms of $\partial_\tau$ and $\partial_{AB}$, one finds, after a calculation, that
the components $\phi_k$ satisfy the wave equations
\[
\mathcal{W}_k[\bmphi] \equiv (1-\tau^2) \partial_\tau^2 \phi_k +2\tau\rho \partial_\tau\partial_\rho \phi_k -\rho^2\partial_\rho^2\phi_k -\frac{1}{2}\{\bmX_+, \bmX_-\}\phi_k + 2(s-k-\tau)\partial_\tau\phi_k + 2(s-k)^2\phi_k =0
\]
for $k=0,\, \dots, \, 2s$. For convenience, we also introduce the \emph{reduced wave operator} acting on a scalar $\zeta$ as
\begin{equation}
\blacksquare \zeta \equiv (1-\tau^2) \ddot\zeta +
2\tau\rho \dot{\zeta}' -\rho^2\zeta''
-2\tau \dot\zeta -\frac{1}{2}\{\bmX_+,  \bmX_-\} \zeta,
\label{FReducedWaveOperator}
\end{equation}
where we denote $\dot{\phantom{X}}\equiv \partial_\tau$ and $\phantom{X}^\prime\equiv \partial_\rho$. 
In this notation
\begin{equation}
\mathcal{W}_k[\bmphi] =\blacksquare \phi_k + \mathbf{L}\phi_k, \qquad k = 0,\, \dots, \, 2s,
\label{ReducedWaveEqnComponents}
\end{equation}
where $\mathbf{L}$ is a linear lower order operator such that $[\mathbf{L},\partial_\rho]=0$. 

\begin{remark}
\label{Remark:Data} In a standard (i.e. non-characteristic) initial value problem, the system \eqref{ReducedWaveEqnComponents} of wave equations needs to be supplemented with the initial data $(\phi_{k}^\star,\dot{\phi}_{k}^\star)$ where $\phi_{k}^\star \equiv \phi_k(\tau_\star)$ and $\dot{\phi}_{k}^\star \equiv \dot{\phi}_k(\tau_\star)$ for some initial hypersurface $\mathcal{S}_\star = \{\tau = \tau_* \}$,  $\tau_\star\in (-1,1)$. We observe that, since the system \eqref{ReducedWaveEqnComponents} arises from the first order system \eqref{StandardEvolutionFirst}--\eqref{StandardEvolutionLast}, here the time derivative part $\dot{\phi}_{k}^\star$ of the data is always expressible in terms of $\phi^\star_k$.
\end{remark}

\subsection{A more general gauge}
\label{Appendix:MoreGeneralGauge}

Certain arguments in the main text require a version of the F-gauge in which only the critical sets $\mathcal{I}^\pm$ are singular sets of the evolution equations, and $\scri^\pm$ is given by a non-horizontal hypersurface in a generalized coordinate plane $(\hat{\tau}, \rho)$.

Proceeding as in \Cref{Section:CylinderMinkowski}---but instead of writing $x^0 = \tau \rho$---we define a new time coordinate $\hat{\tau}$ by
\[
x^0 = \hat{\tau} \kappa, \quad \mbox{where} \quad \kappa = \rho \mu,
\]
and $\mu$ is a smooth function of $\rho$ such that $\mu(0)=1$, but $\mu \not\equiv 1$. The specific version of the F-gauge introduced in \Cref{Section:CylinderMinkowski} corresponds to the choice $\mu \equiv 1$. This more general choice of the coordinate $\hat{\tau}$ leads to the conformal factor
\begin{equation}
\hat{\Theta} = \frac{\rho}{\mu}\big(1-\mu^2 \hat{\tau}^2 \big) = \kappa^{-1}\Xi,
\label{ConformalFactorThetaHat}
\end{equation}
which, in turn, gives rise to the unphysical metric
\begin{eqnarray*}
&& \hat{\bmeta} = \hat{\Theta}^2 \tilde{\bmeta},\\
&& \phantom{\hat{\bmeta}}= \mathbf{d}\hat{\tau}\otimes \mathbf{d}\hat{\tau} + \frac{\hat{\tau}\kappa'}{\kappa}\big( \mathbf{d}\hat{\tau}\otimes\mathbf{d}\rho + \mathbf{d}\rho\otimes \mathbf{d}\hat{\tau} \big) -\frac{(1-\hat{\tau}^2\kappa^{\prime 2})}{\kappa^2}\mathbf{d}\rho\otimes\mathbf{d}\rho - \frac{1}{\mu^2}\bmsigma.
\end{eqnarray*}
This conformal metric is supplemented with the following choice of frame:
\begin{eqnarray*}
&& \hat{\bme}_{\bmzero\bmzero'} = \frac{1}{\sqrt{2}}\big((1-\kappa'\hat{\tau})\big) \bmpartial_{\hat{\tau}} + \kappa \bmpartial_\rho, \\
&& \hat{\bme}_{\bmone\bmone'} = \frac{1}{\sqrt{2}}\big((1+\kappa'\hat{\tau})\big) \bmpartial_{\hat{\tau}} - \kappa \bmpartial_\rho, \\
&& \hat{\bme}_{\bmzero\bmone'} =- \frac{1}{\sqrt{2}}\mu \bmX_+, \\
&& \hat{\bme}_{\bmone\bmzero'} =-\frac{1}{\sqrt{2}}\mu \bmX_-,
\end{eqnarray*}
with associated non-vanishing spin connection coefficients
\[
\hat{\Gamma}_{\bmzero\bmzero'\bmzero\bmone}=\hat{\Gamma}_{\bmone\bmone'\bmzero\bmone}=-\frac{1}{2\sqrt{2}} \kappa' \quad \text{and} \quad \hat{\Gamma}_{\bmzero\bmone'\bmone\bmone}=\hat{\Gamma}_{\bmone\bmzero'\bmzero\bmzero}=\frac{1}{\sqrt{2}}\rho\mu'.
\]
The above is equivalent to the expression
\[
\Gamma_{\bmA\bmA'\bmC\bmD}= \frac{1}{\sqrt{2}}\rho\mu' \tau^\bmB{}_{\bmA'}\epsilon_{\bmA\bmC}\sigma^1_{\bmB\bmD}+\frac{1}{\sqrt{2}}\rho\mu'\tau_{\bmD\bmA'}\sigma^1_{\bmA\bmC}-\frac{1}{\sqrt{2}} (\mu+\rho\mu')\tau_{\bmA\bmA'}\sigma^1_{\bmC\bmD}.
\]
The space spinor counterpart of the above expressions is given by
\begin{eqnarray*}
&& \bmpartial =\sqrt{2} \bmpartial_\tau, \\
&& \bmpartial_{\bmA\bmB} = \sqrt{2} \sigma^1_{\bmA\bmB} \left(-\tau \kappa'  \bmpartial_\tau + \kappa\bmpartial_\rho \right) + \frac{1}{\sqrt{2}} \sigma^0_{\bmA\bmB} \bmX_+ - \frac{1}{\sqrt{2}} \sigma^2_{\bmA\bmB} \bmX_-, \\
&& \Gamma_{\bmA\bmB\bmC\bmD} = -\frac{1}{\sqrt{2}}\rho \mu' (\epsilon_{\bmA\bmC}\sigma^1_{\bmB\bmD} + \epsilon_{\bmB\bmD}\sigma^1_{\bmA\bmC}) -\frac{1}{\sqrt{2}}(\mu+\rho \mu')\epsilon_{\bmA\bmB}\sigma^1_{\bmC\bmD}.
\end{eqnarray*}
In particular, the acceleration is given by
\[
f_{\bmA\bmB}=\sqrt{2}(\mu +\rho\mu') \sigma^1_{\bmA\bmB}.
\]

Now, recalling that the conformal factor used in \Cref{Section:CylinderMinkowski} is given by $\Theta = \rho^{-1} \Xi$,  it follows by comparison with equation \eqref{ConformalFactorThetaHat} that
\[
\hat{\bmeta} = \varpi^2 \bmeta, \qquad \varpi\equiv  \frac{1}{\mu}.
\]
The associated transformation of the antisymmetric spinor is then given by
\[
\epsilon_{AB} =\mu \hat{\epsilon}_{AB},
\]
with a scaling of the spin basis given by 
\begin{equation}
o_A = \mu^{1/2} \hat{o}_A, \qquad \iota_A = \mu^{1/2} \hat{\iota}_A.
\label{DemocraticScalingSpinBasis}
\end{equation}
The unphysical spin-$s$ fields are related to the physical one by
\[
\phi_{A_1\cdots A_{2s}} = \Theta^{-1} \tilde{\phi}_{A_1\cdots A_{2s}}, \qquad \hat{\phi}_{A_1\cdots A_{2s}} = \hat{\Theta}^{-1} \tilde{\phi}_{A_1\cdots A_{2s}},
\]
so that, in fact, one has
\[
\hat{\phi}_{A_1\cdots A_{2s}} = \mu \phi_{A_1\cdots A_{2s}}.
\]
The scaling \eqref{DemocraticScalingSpinBasis} then implies that 
\[
\hat{\phi}_i = \mu^{s+1}\phi_i, \qquad i=0, \, \ldots, \, 2s.
\]

\begin{remark}
{Observe that given that $\mu$ is assumed to be a smooth function of $\rho$, it follows that the regularity of the components of the spin-$s$ field is not affected by the rescaling.}
\end{remark}

\noindent Following an approach similar to the one used in \Cref{Appendix:TransportEqnsBasic}, one obtains the equations
\begin{subequations}
\begin{eqnarray}
&& \hspace{-5mm} E_0\equiv (1-\kappa'\tau) \partial_\tau \phi_0 +\kappa\partial_\rho \phi_0 -\mu\bmX_-\phi_1 - s(\mu+\rho\mu')\phi_0=0, \label{StandardEvolutionFirstGeneral}\\
&& \hspace{-5mm} E_k\equiv \partial_\tau \phi_k -\frac{1}{2}\mu\bmX_+\phi_{k-1}-\frac{1}{2}\mu\bmX_-\phi_{k+1}+(k-s)(\mu+\rho\mu')\phi_k=0, \qquad k=1, \, \dots, \, 2s-1,\label{StandardEvolutionMiddleGeneral}\\
&& \hspace{-5mm} E_{2s}\equiv (1+\kappa'\tau) \partial_\tau \phi_{2s} -\kappa\partial_\rho\phi_{2s} -\mu\bmX_+ \phi_{2s-1} + s(\mu+\rho\mu')\phi_{2s}=0, \label{StandardEvolutionLastGeneral}
\end{eqnarray}
\end{subequations}
and
\[
C_k\equiv \kappa'\tau \partial_\tau \phi_k -\kappa\partial_\rho\phi_k +\frac{1}{2}\mu\bmX_-\phi_{k+1}-\frac{1}{2}\mu\bmX_+\phi_{k-1}+ s\mu'\phi_k=0, \qquad k=1, \, \dots, \, 2s-1.
\]

\section{F-expansions}
\label{Appendix:SolutionJets}

In this appendix we provide a detailed overview of the construction of \emph{F-expansions} for the solutions to the spin-$s$ field equations: expansions which exploit the cylinder at spatial infinity $\mathcal{I}$ being a total characteristic of the evolution equations \cite{Fri98a,Val03a,Fri04}.

\subsection{Interior equations on \texorpdfstring{$\mathcal{I}$}{I}}

The total characteristic nature of the cylinder at spatial infinity for the spin-$s$ equations is reflected in the fact that the reduced wave operator $\blacksquare$---essentially the principal and sub-principal parts of the full wave operator $\Box$ acting on $\phi_{A_1 \ldots A_{2s}}$ (see \Cref{Appendix:WaveEquations})---reduces, upon evaluation on $\mathcal{I}$, to the interior operator $\accentset{\circ}{\blacksquare}\equiv\blacksquare|_{\mathcal{I}}$, which acts on a scalar 
$\zeta$ by
\[
\accentset{\circ}{\blacksquare} \zeta \equiv (1-\tau^2) \ddot\zeta
-2\tau \dot\zeta -\frac{1}{2}\{ \bmX_+, \, \bmX_- \} \zeta,
\]
where $\dot{\zeta} = \partial_\tau \zeta$ and $\{\bmX_+, \, \bmX_- \}$ denotes the anticommutator of $\bmX_+$ and $\bmX_-$. Computing the commutator
\[ [ \partial_\rho^p, \, \blacksquare] \zeta =
2\tau p \partial_\rho^p\dot{\zeta} -2 p\rho \partial_\rho^p \zeta' - p(p-1)\partial_\rho^p\zeta \]
and using \cref{ReducedWaveEqnComponents}, we obtain
\[
\blacksquare \partial^p_\rho \phi_k + 2\tau p \partial^p_\rho\dot{\phi}_k-2 \rho p \partial^p_\rho \phi'_s -p(p-1)\partial_\rho^p\phi_k + \mathbf{L}\partial^p_\rho\phi_k=0, 
\]
where $\mathbf{L}$ is a first order differential operator which commutes with $\partial_\rho$. Evaluating on $\mathcal{I}$, one therefore finds the following set of interior equations on $\mathcal{I}$ for all $p \geq 0$,
\begin{equation}
\label{InteriorEquationsOnCylinder}
\accentset{\circ}{\blacksquare} \phi^{(p)}_k + 2\tau p \dot{\phi}^{(p)}_k -p(p-1)\phi^{(p)}_k + \mathbf{L}\phi^{(p)}_k=0. 
\end{equation}
The equations \eqref{InteriorEquationsOnCylinder} are linear, and in fact solvable explicitly in terms of a basis of functions on $\mathrm{SU}(2)$.

\subsection{Expansions in terms of $T_m{}^j{}_k$}
\label{Section:Expansions_In_Harmonics}

Given that the spin weight of component $\phi_k$ is $(k-s$), and writing $m=2q$, we look, for $p \geq |s-k|$, for solutions $\phi_k^{(p)}$ to \eqref{InteriorEquationsOnCylinder} of the form
\begin{equation}
\phi_k^{(p)}= \sum_{q=|s-k|}^{p} \sum_{j=0}^{2q} a_{k,p;q,j} T_{2q}{}^j{}_{q+s-k},
\label{ExpansionHarmonics}
\end{equation}
where the coefficients $a_{k,p;q,j}$ are functions of $\tau$.

\begin{remark}
This Ansatz for the coefficients $\phi_k^{(p)}$ is motivated by analogy to the analysis in \cite{Fri98a,Val03a} of time symmetric initial data sets for the Einstein field equations admitting a conformal metric which is analytic at spatial infinity. Depending on the particular application at hand, the Ansatz can be suitably generalised. For example, the analysis in \cite{MagVal22} of BMS charges for spin-1 and spin-2 fields considers, e.g. for the coefficient $\phi^{(0)}_2$ of the spin-2 field, the expansion
\[
\phi_2^{(0)}= \sum_{q=0}^{\infty} \sum_{j=0}^{2q} a_{2,0;q,j} T_{2q}{}^j{}_{q+s-2}.
\]
In that particular case one finds that the coefficient $a_{2,0;q,j}$ decomposes into a sum of a regular part and a part which has logarithmic divergences at $\tau=\pm1$. The part of the solution with logarithmic terms can be eliminated by fine-tuning the initial data.
\end{remark}

\begin{remark}
In the case $s=\frac{1}{2}$ the expansion \eqref{ExpansionHarmonics} takes the particular form
\[
\phi_k^{(p)}=\sum_{q=|k-\frac{1}{2}|}^{p} \sum_{j=0}^{2q} a_{k,p;q,j} T_{2q}{}^j{}_{q+\frac{1}{2}-k}, \qquad k \in \{0, 1 \},
\]
with the understanding that $q$ is a proper half-integer (i.e. it does not simplify to an integer). In particular, then the indices of $T_i{}^j{}_k$ are always integers. Comparing the above expansion with the relation \eqref{SpinWeightedHarmonicsCorrespondence} shows that in this case one has expansions in term of the harmonics ${}_{\pm \frac{1}{2}}Y_{lm}$. The expansions for fields with higher half-integer spins are analogous.
\end{remark}

From the expansion \eqref{ExpansionHarmonics} and equation \eqref{InteriorEquationsOnCylinder}, it follows then that the coefficient $a_{k,p;q,j}$ satisfies the ODE
\begin{equation}
    (1-\tau^2)\ddot{a}_{k,p;q,j} +2 \big((p-1)\tau +s-k  \big)\dot{a}_{k,p;q,j} + \big( q^2+q -p^2-p  \big)a_{k,p;q,j} =0,
\label{JacobiEqn:Raw}
\end{equation}
where the integers $(k,p,q,j)$ are such that
\[
0 \leq k \leq 2s, \qquad |s-k|\leq p, \qquad |s-k|\leq q \leq p, \quad \text{and} \quad 0\leq j\leq 2q. 
\]
Equation \eqref{JacobiEqn:Raw} is an example of a \emph{Jacobi ordinary differential equation}. Jacobi equations are usually parametrised in the form
\begin{equation}
    D_{(n,\alpha,\beta)}a \equiv (1-\tau^2) \ddot{a} +\big(\beta-\alpha -(\alpha+\beta+2)\tau \big)\dot{a} + n(n+\alpha+\beta) a =0.
    \label{JacobiEqn:Model}
\end{equation}
A direct comparison between equations \eqref{JacobiEqn:Raw} and \eqref{JacobiEqn:Model} gives
\begin{subequations}
\begin{eqnarray}
&& \alpha = -p +(k-s), \label{Jacobi:alpha}\\
&& \beta = -p -(k-s), \label{Jacobi:beta} \\
&& n= n_1\equiv p+q, \qquad \mbox{or} \qquad n=n_2\equiv p-q-1. \label{Jacobi:n}
\end{eqnarray}
\end{subequations}
As we shall see in the following subsection, the qualitative nature of the solutions to equation \eqref{JacobiEqn:Raw} differs depending on whether $|s-k|\leq q<p$ or $q=p$; the case $q=p$ corresponds to the harmonic which acquires logarithmic singularities at $\tau = \pm 1$.


\subsubsection{Properties of the Jacobi differential equation}
\label{Section:JacobiODE}

An extensive discussion of the solutions of the Jacobi equation can be found in the monograph \cite{Sze78} from which we borrow a number of identities. The solutions to  \eqref{JacobiEqn:Model} are given by the \emph{Jacobi polynomials} $P_n^{(\alpha, \beta)}(\tau)$, of degree $n$, defined by
\[
P_n^{(\alpha,\beta)}(\tau) \equiv \sum_{l=0}^n
\binom{n+\alpha}{l}\binom{n+\beta}{n-l}\left( \frac{\tau-1}{2}\right)^{n-l}\left(\frac{\tau+1}{2}\right)^l,
\]
where for $z \in \mathbb{C}$, $r \in \mathbb{N}$ the binomial coefficient is defined by 
\[ \binom{z}{r} = \begin{cases} \displaystyle\frac{\Gamma(z+1)}{\Gamma(r+1)\Gamma(z-r+1)} & r \geq 0, \\ \hfil 0 & r < 0. \end{cases}
\]
In particular, one has
\[
P_0^{(\alpha,\beta)}(\tau)=1,
\]
and
\[
P_n^{(\alpha,\beta)}(-\tau) =(-1)^n P_n^{(\beta,\alpha)}(\tau).
\]
The differential operator defined by \eqref{JacobiEqn:Model} exhibits the following symmetries,
\begin{subequations}
\begin{eqnarray}
&&\hspace{-1.5cm} D_{(n,\alpha,\beta)} \left( \left( \frac{1-\tau}{2}
   \right)^{-\alpha}a(\tau)  \right) = \left( \frac{1-\tau}{2}
   \right)^{-\alpha} D_{(n+\alpha,-\alpha,\beta)} a(\tau), \label{JacobiEqnIdentity1}\\
&&\hspace{-1.5cm} D_{(n,\alpha,\beta)} \left( \left( \frac{1+\tau}{2}
    \right)^{-\beta}a(\tau)  \right) = \left( \frac{1+\tau}{2}
    \right)^{-\beta} D_{(n+\beta,\alpha,-\beta)} a(\tau), \label{JacobiEqnIdentity2}\\
&&\hspace{-1.5cm} D_{(n,\alpha,\beta)} \left( \left( \frac{1-\tau}{2}
   \right)^{-\alpha}\left( \frac{1+\tau}{2}
    \right)^{-\beta}a(\tau)  \right) = \left( \frac{1-\tau}{2}
   \right)^{-\alpha}\left( \frac{1+\tau}{2}
    \right)^{-\beta} D_{(n+\alpha+\beta,-\alpha,-\beta)} a(\tau), \label{JacobiEqnIdentity3}
\end{eqnarray}
\end{subequations}
which hold for $|\tau|<1$, arbitrary $C^2$-functions $a(\tau)$, and
arbitrary real values of the parameters $\alpha$, $\beta$, $n$. An alternative definition of the Jacobi polynomials, convenient for
verifying when the functions vanish identically, is given by
\[
P_n^{(\alpha,\beta)}(\tau) = \frac{1}{n!}\sum_{k=0}^n c_k \left( \frac{\tau-1}{2} \right)^k,
\]
with
\begin{eqnarray*}
&& \hspace{-1cm}c_0\equiv (\alpha+1)(\alpha+2)\cdots(\alpha +n), \\
&& \hspace{2cm}\vdots \\
&& \hspace{-1cm} c_k \equiv \frac{n!}{k!(n-k)!}(\alpha+k+1)(\alpha+k+2)\cdots
   (\alpha+n) \\
&& \hspace{45pt} \times (n+1 +\alpha +\beta) (n+2+\alpha+\beta)\cdots
   (n+k+\alpha+\beta),\\
&& \hspace{2cm}\vdots \\
&& \hspace{-1cm}c_n\equiv (n+1+\alpha+\beta)(n+2+\alpha+\beta) \cdots (2n+\alpha+\beta).
\end{eqnarray*}

\subsubsection{Solutions for $|s-k|\leq q <p$}
In the case $|s-k|\leq q <p$, one has from direct inspection of the formulae above, that the polynomial $P_{n_1}^{(\alpha,\beta)}(\tau)$ with $(n_1,\,\alpha,\,\beta)$ as given by \cref{Jacobi:alpha,Jacobi:beta,Jacobi:n} vanishes identically, while 
\[
Q_2(\tau) \equiv P_{n_2}^{(\alpha,\beta)}(\tau)
\]
gives a polynomial of degree $n_2$. A further non-trivial solution can be written down using the identity \eqref{JacobiEqnIdentity1}; one finds a polynomial of degree $n_1$ given by
\[
Q_1(\tau) \equiv \left(\frac{1-\tau}{2}  \right)^{p-k+s}P_{q+k-s}^{(-\alpha,\beta)}(\tau).
\]
Since $n_2<n_1$, the solutions $Q_1$ and $Q_2$ are linearly independent. Yet another solution can be obtained using identity \eqref{JacobiEqnIdentity2}, namely
\[
Q_3(\tau)\equiv \left( \frac{1+\tau}{2} \right)^{p+k-s}P_{q-k+s}^{(\alpha,-\beta)}(\tau),
\]
which, again, is a polynomial of degree $n_1$. It can be verified that $Q_1$ and $Q_3$ are also linearly independent. Making use of these solutions one can write the general solution to equation \eqref{JacobiEqn:Raw}, for $|s-k|\leq q <p$, in the symmetric form
\begin{equation}
\label{a_coefficient_q_less_p}
a_{k,p;q,j}(\tau) = \mathfrak{c}_{k,p;q,j}\left(\frac{1-\tau}{2}  \right)^{p-k+s}P_{q+k-s}^{(-\alpha,\beta)}(\tau) + \mathfrak{d}_{k,p;q,j}\left( \frac{1+\tau}{2} \right)^{p+k-s}P_{q-k+s}^{(\alpha,-\beta)}(\tau),
\end{equation}
with $\mathfrak{c}_{_{k,p;q,j}}$ and $\mathfrak{d}_{_{k,p;q,j}}$ constants to be determined from the initial conditions. In particular, we have the following lemma:

\begin{lemma}
For $|s-k|\leq q <p$, the solutions to the Jacobi equation \eqref{JacobiEqn:Raw} are analytic at $\tau=\pm 1$.
\end{lemma}

\begin{remark}
A direct inspection of the formulae given above show that they do not give non-vanishing solutions if $p=q$.
\end{remark}

 \subsubsection{Solutions for $p=q$}
\label{Appendix:LogarithmicTerms}
 
 In order to obtain solutions in the case $p=q$, we make use of identity \eqref{JacobiEqnIdentity3} with $n=n_1$ and look for solutions of the form
 \[
 a_{k,p;p,j}(\tau)= \left(\frac{1-\tau}{2}\right)^{-\alpha}\left(\frac{1+\tau}{2} \right)^{-\beta} b(\tau),
 \]
 with $b(\tau)$ satisfying the equation
 \[
 D_{(0,p-k+2,p+k-s)}b(\tau) = (1-\tau^2) \ddot{b}(\tau) + 2\big(k-s -(p+1)\tau \big)\dot{b}(\tau)=0.
 \]
This can be integrated to give
 \begin{equation}
 b(\tau) = \mathfrak{c}_{k,p;p,j} + \mathfrak{d}_{k,p;p,j} \int_0^\tau \frac{\mathrm{d}\varsigma}{(1+\varsigma)^{p-s+k+1}(1-\varsigma)^{p+s-k+1}},
 \label{Solution:b}
 \end{equation}
 with $\mathfrak{c}_{_{k,p;p,j}}$ and $\mathfrak{d}_{k,p;p,j}$ constants. 
 Thus, the general solution to \eqref{JacobiEqn:Raw} for $p=q$ can be written as 
 \begin{equation}
 a_{k,p;p,j}(\tau)=\left( \frac{1-\tau}{2}\right)^{p-k+s}\left(\frac{1+\tau}{2} \right)^{p+k-s} \left(\mathfrak{c}_{k,p;p,j} + \mathfrak{d}_{k,p;p,j} \int_0^\tau \frac{\mathrm{d}\varsigma}{(1+\varsigma)^{p-s+k+1}(1-\varsigma)^{p+s-k+1}}\right).
 \label{JacobiSolution:Logarithmic}
 \end{equation}
 Now, expanding the integrand of \eqref{Solution:b} in partial fractions, one sees that $ a_{k,p;p,j}(\tau)$ contains terms of the form
 \[
 (1-\tau)^{p-k+s}\log(1-\tau) \quad \text{and} \quad (1+\tau)^{p+k-s}\log(1+\tau),
 \]
 which are, respectively, $C^{p-k+s-1}$ and $C^{p+k-s-1}$ at $\tau=\pm 1$. These are the only singular terms in the solution \eqref{JacobiSolution:Logarithmic}. The rest of the solution is polynomial in $\tau$, and thus analytic at $\tau=\pm 1$. The solutions in the case $p=q$ are therefore not smooth at the critical sets $\mathcal{I}^\pm$ of the conformal boundary. 
 
 \begin{remark}
The logarithmic divergences in \eqref{JacobiSolution:Logarithmic} can be set to vanish by fine-tuning initial data so that $\mathfrak{d}_{k,p;p,j}=0$. 
 \end{remark}

 \section{Expansions near $\mathscr{I}^-$}
 \label{Appendix:ExpansionsNullInfinity}
 
 The approach in the main text makes use of expansions not only near $\mathcal{I}$ but also near $\mathscr{I}^-$. The expansions near $\scri^-$ are computed from characteristic data. While the construction of the expansions near $\mathcal{I}$ in \Cref{Appendix:SolutionJets} is fairly general, it turns out that to construct the expansions near $\scri^-$ it is convenient to require the fields to possess a certain amount of decay towards the critical set  $\mathcal{I}^-$.

\subsection{Leading order terms}
We begin by observing that the component $\phi_0$ encodes the freely specifiable characteristic initial data on $\mathscr{I}^-$. This can be easily seen by evaluating the equations $A_k=0$ at $\mathscr{I}^- = \{ \tau = -1 \}$. One gets
\[
\rho \partial_\rho \phi_{k+1}|_{\scri^-} + \bmX_+\phi_k|_{\scri^-} -(k+1-s)\phi_{k+1}|_{\scri^-}=0, \qquad 0 \leq k \leq 2s - 1,
\]
so that the components $\phi_{k}$ for $1 \leq k \leq 2s$ can be computed in a hierarchical manner by solving ODEs along the generators of $\mathscr{I}^-$. 

We assume from the outset that $\phi_0|_{\scri^-}$ is bounded as $\rho \to 0$. In particular, $\bmX_+ \phi_0 = \mathcal{O}(1)$ as $\rho \to 0$. We then set
\[ 
\phi_1|_{\scri^-} = - \rho^{-(s-1)} \int_0^\rho \varrho^{s-2} \bmX_+ \phi_0|_{\scri^-} (\varrho) \d \varrho, 
\]
which, by the above assumption, satisfies $ \phi_1|_{\scri^-} = \mathcal{O}(1)$ as $\rho \to 0$. This therefore gives a solution to the above equation with $k=0$, which is continuous on $\scri^-_{\rho_\star} \cup \mathcal{I}^-$. We write down the other components hierarchically to obtain
\[
\phi_{k+1}|_{\scri^-} = - \rho^{-(s-k-1)} \int_0^\rho \varrho^{s-k-2} \bmX_+ \phi_k |_{\scri^-} (\varrho) \d \varrho 
\]
for $0 \leq k \leq 2s - 1$. We observe that all components $\phi_k|_{\scri^-}$ defined in this way inherit the decay rate towards $\mathcal{I}^-$ prescribed for $\phi_0$. For instance, if $\phi_0|_{\scri^-} \sim \rho^\gamma$ for some $\gamma \geq 0$, then $\phi_k|_{\scri^-} \sim \rho^\gamma$ for all $0 \leq k \leq 2s$.
 With the knowledge of all components $\phi_k|_{\scri^-}$, $0 \leq k \leq 2s$, one can then proceed to compute $(\partial_\tau \phi_0)|_{\mathscr{I}^-}$. For this we evaluate the equation $B_0=0$ at $\tau=-1$ to obtain
\[
2(\partial_\tau \phi_0)|_{\mathscr{I}^-} +\rho\partial_\rho \phi_0|_{\mathscr{I}^-} -\bmX_-\phi_1|_{\mathscr{I}^-}-s \phi_0|_{\mathscr{I}^-}=0,
\]
which is just an algebraic equation for $(\partial_\tau \phi_0)|_{\mathscr{I}^-}$. Note that, with reference to the decay $\gamma$ assumed for $\phi_0|_{\scri^-}$, the time derivative term decays at the same rate, $\partial_\tau \phi_0|_{\scri^-} \sim \rho^\gamma$.

\subsection{Higher order terms}
\label{Appendix:Expansion_Near_Past_Null_Infinity}

More generally, suppose $(\partial^q_\tau \phi_0)|_{\scri^-}$ is known for some $q \geq 1$. Taking $q$ $\tau$-derivatives of $A_k = 0$ and evaluating at $\mathscr{I}^-$, one finds 
\[
\rho\partial_\rho (\partial^q_\tau \phi_{k+1})|_{\mathscr{I}^-}+\bmX_+ (\partial^q_\tau \phi_{k})|_{\mathscr{I}^-} -(k+1+q-s) (\partial^q_\tau \phi_{k+1})|_{\mathscr{I}^-} =0,
\]
$0 \leq k \leq 2s-1$, and so one can compute the derivatives $(\partial^q_\tau \phi_{k+1})|_{\mathscr{I}^-}$ by solving ODEs along the generators of $\mathscr{I}^-$. Specifically, we set
\[ 
\partial_\tau^q\phi_{k+1}|_{\scri^-} = - \rho^{k+1+q-s} \int_0^\rho \varrho^{-(k+2+q-s)} \bmX_+ (\partial_\tau^q \phi_k )|_{\scri^-}(\varrho) \d \varrho 
\]
for $0 \leq k \leq 2s- 1$. Now, once $(\partial^q_\tau \phi_{k+1})|_{\mathscr{I}^-}$ for $ 0 \leq k \leq 2s$ are known, one may use the condition
\[
(\partial^q_\tau B_0)|_{\mathscr{I}^-}=0
\]
to solve for $(\partial^{q+1}_\tau \phi_{0})|_{\mathscr{I}^-}$. Observe that, as before, the decay rate of all derivatives $(\partial_\tau^q \phi_k)|_{\scri^-}$ is inherited from the decay rate prescribed for $\phi_0$, e.g. $(\partial_\tau^q \phi_k)|_{\scri^-} \sim \rho^\gamma$ with reference to the above example. In summary, one obtains a formal expansion near $\scri^-$ of the form
\[
\phi_k = \sum_{q'=0}^{q-1} \frac{1}{q'!} (\partial^{q'}_\tau \phi_k)|_{\mathscr{I}^-}(1 + \tau)^{q'} + I^q_k, 
\]
$0 \leq k \leq 2s$, where $I^q_k$ is the remainder at order $q \geq 0$ for the component $\phi_k$.

 

\printbibliography

\end{document}